\documentclass[12pt,oneside]{memoir}

\makeatletter
\makechapterstyle{thesis}{
	
	\setlength{\beforechapskip}{0pt}
	\setlength{\midchapskip}{0pt}
	\setlength{\afterchapskip}{0pt}

}
\makeatother

\makeatletter
\makechapterstyle{thesis2}{
	\setlength{\beforechapskip}{0pt}
	\setlength{\midchapskip}{0pt}
	\setlength{\afterchapskip}{0pt}

}
\makeatother

\makeatletter
\makechapterstyle{thesis3}{
	\setlength{\beforechapskip}{0pt}
	\setlength{\midchapskip}{0pt}
	\setlength{\afterchapskip}{0pt}

}
\makeatother

\makeatletter
\newdimen\rh@wd
\newdimen\rh@hta
\newdimen\rh@htb
\newbox\rh@box
\def\rh@measure#1{\setbox\rh@box=\hbox{$#1$}\rh@wd=\wd\rh@box \rh@hta=\ht\rh@box}

\def\widecheck#1{\rh@measure{#1}%
	\setbox\rh@box=\hbox{$\widehat{\vrule height \rh@hta width\z@ \kern\rh@wd}$}%
	\rh@htb=\ht\rh@box \advance\rh@htb\rh@hta \advance\rh@htb\p@
	\ooalign{$\vrule height \ht\rh@box width\z@ #1$\cr
		\raise\rh@htb\hbox{\scalebox{1}[-1]{\box\rh@box}}\cr}}
\makeatother

\maxsecnumdepth{subsubsection}
\maxtocdepth{subsection}

\usepackage[percent]{overpic}
\usepackage{color}
\definecolor{blue1}{rgb}{0.03, 0.27, 0.49}
\usepackage{sidecap}
\usepackage{caption}
\usepackage{cmap}
\usepackage[T1]{fontenc}
\usepackage{amsthm}
\usepackage{amsmath}
\usepackage{amssymb}
\usepackage[utf8]{inputenc}
\usepackage{mathtools}

\usepackage{graphicx}
\usepackage[top=1.5in, bottom=1.5in, left=1.25in, right=1.25in]{geometry}
\usepackage{setspace}
	
\input{qcircuit}

\def\rA{{\mathrm{A}}} \def\rB{{\mathrm{B}}} \def\rC{{\mathrm{C}}} \def\rD{{\mathrm{D}}} \def\rE{{\mathrm{E}}} \def\rF{{\mathrm{F}}} \def\rG{{\mathrm{G}}} \def\rH{{\mathrm{H}}} \def\rI{{\mathrm{I}}} \def\rL{{\mathrm{L}}} \def\rM{{\mathrm{M}}} \def\rN{{\mathrm{N}}} \def\rO{{\mathrm{O}}} \def\rP{{\mathrm{P}}} \def\rR{{\mathrm{R}}} \def\rS{{\mathrm{S}}} \def\rX{{\mathrm{X}}} \def\rY{{\mathrm{Y}}} \def\rZ{{\mathrm{Z}}}
\def\tA{{\mathcal{A}}} \def\tB{{\mathcal{B}}} \def\tC{{\mathcal{C}}} \def\tD{{\mathcal{D}}} \def\tE{{\mathcal{E}}} \def\tF{{\mathcal{F}}} \def\tG{{\mathcal{G}}} \def\tI{{\mathcal{I}}} \def\tL{{\mathcal{L}}} \def\tR{{\mathcal{R}}} \def\tT{{\mathcal{T}}} \def\tU{{\mathcal{U}}} \def\tV{{\mathcal{V}}} \def\tW{{\mathcal{W}}} \def\tX{{\mathcal{X}}} \def\tY{{\mathcal{Y}}} \def\tZ{{\mathcal{Z}}}

\def\Tr{{\mathsf{Tr}}}

\newcommand\st[2][]{\mathsf{St}_{\mathbb{#1}}(\mathrm{#2})} \newcommand\eff[2][]{\mathsf{Eff}_{\mathbb{#1}}(\mathrm{#2})}
\newcommand\transf[3][]{\mathsf{Transf}_{\mathbb{#1}}(\mathrm{#2}\rightarrow\mathrm{#3})}
	
\usepackage[italian,english]{babel}
\newtheorem{theorem}{Theorem}
\newtheorem{defn}{Definition}

\newtheorem{lemma}{Lemma}

\newtheorem{prop}{Proposition}
\newtheorem{corollary}{Corollary}
\newtheorem{axiom}{Axiom}

\begin{document}

	\newgeometry{top=1in, bottom=1in, left=1in, right=1in}
\begin{titlingpage}
	\begin{center}
		
		\vspace{0.3cm}	
		\includegraphics[width=0.25\textwidth]{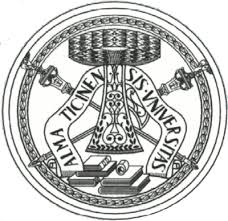}\par\vspace{0.4cm}	
		{\huge\scshape universit\'{a} degli studi di pavia \par}
		{\huge\scshape \textsc{dipartimento di fisica} \par}\vspace{0.1cm}
		{\Large\scshape \textsc{corso di laurea magistrale in scienze fisiche} \par}
		\vspace{3cm}
		{\Huge\bfseries Bit Commitment in Operational Probabilistic Theories\par}
		\vspace*{0.2cm}
	\end{center}
	\vspace{2cm}
	{\large \text{Tesi per la Laurea Magistrale di} \par}
	{\large \textbf{Lorenzo Giannelli} \par}		
	\vspace{2cm}
	
	\raggedleft	
	{\large \text{Relatore}\par}
	{\large \textbf{Chiar.mo Prof. Giacomo Mauro D'Ariano}\par}
	{\large \text{Correlatore}\par}
	{\large \textbf{Dott. Alessandro Tosini}\par}
	\raggedleft	
	\vspace{1cm}
	\begin{center}
		{\large Anno Accademico 2019-2020\par}
	\end{center}
\end{titlingpage}
\restoregeometry


\bigskip\begin{quotation}\begin{center}\begin{em}
			\thispagestyle{empty}
			\hfill A Piero
			\par\end{em}\end{center}\end{quotation}

\cleartoverso 

\renewenvironment{abstract}
{\cleardoublepage\thispagestyle{empty}\null\vfill\begin{center}
		\bfseries Abstract \end{center}}
{\vfill\null}

\begin{abstract}
	\noindent The aim of this thesis is to investigate the bit commitment protocol in the framework of operational probabilistic theories. In particular a careful study is carried on the feasibility of bit commitment in the non-local boxes theory and in order to do this new aspects of the theory are presented.\\
	
	\hfill
	
	\noindent Lo scopo di questa tesi è di investigare il protocollo di bit commitment all'interno delle teorie probabilistiche operazionali. In particolare si è analizzato attentamente la fattibilità del protocollo all'interno della teoria dei non-local boxes e i nuovi aspetti della teoria emersi in questa analisi sono presentati.
	
\end{abstract}
\clearpage
\thispagestyle{empty}
\phantom{a}
\newpage

\linespread{1.3}

\frontmatter 

\pagestyle{plain}
\tableofcontents

\mainmatter 
\pagestyle{Ruled}

\chapterstyle{thesis2}
\chapter*{Introduction}

\addcontentsline{toc}{chapter}{Introduction}
\thispagestyle{plain}

The study of quantum foundations is a discipline of science that seeks to understand the most characterizing aspects of quantum theory, to reformulate it and even propose new generalizations. An active area of research in quantum foundations is therefore to find alternative formulations of quantum theory which rely on physically compelling principles in attempt to find a re-derivation of the quantum formalism in terms of operational axioms. One of the most interesting effort in this direction is made investigating the relations between the current operational axioms and the main results of quantum information theory.\\
As a recent example it has been proved in Ref.~\cite{niwd} that the no information without disturbance (NIWD) theorem, i.e. the impossibility in quantum theory to extract information without disturbing the state of the system or its correlations with other systems, is independent of both local discriminability and purification, two of the defining axioms of quantum theory. Especially the latter, as we will see thoroughly in the first Chapter, is considered as a characteristic and distinctive quantum trait but now NIWD can be exhibited in absence of it and also of most of the principles of quantum theory.\\

The NIWD property spawns other no-go theorems, which represent some of the most famous and classic results in quantum information theory. Among these results one can certainly list the no-cloning theorem, the no-programming and the impossibility of perfectly secure bit commitment. In our thesis our efforts are focused on the latter.\\
A bit commitment (BC) protocol is meant to allow one party, Alice, to send a bit to a second party, Bob, in such a way that Bob cannot read the bit until Alice allows for its disclosure, while Alice cannot change the value of the bit after she encoded it. The bit commitment protocol is a very important primitive in cryptography, and perfectly secure protocols are known to be impossible in classical information theory. This is the case in quantum theory as well. The proof involves a very important characterization theorem for general theories with purification \cite{chiribella}.\\

\thispagestyle{plain}

Numerous bit commitment protocols have been proposed in literature and this cryptography primitive has been deeply studied, especially the possibility of unconditionally secure bit commitment, both in quantum and classical information theory, due to its importance in practical applications. We would like to investigate the relations between the theorem of no-bit commitment and the operational axioms that characterize quantum theory.\\

A general strategy to apprehend the nature of these links is to test the validity of the theorem in a theory that lacks one or more principles; naturally the first try has to be made excluding the purification, the most quantum feature.\\

The answer of how to start in the analysis comes directly from the literature on bit commitment protocol itself. In fact, after it was proved to be impossible in quantum theory the protocol has been tested in more general scenarios, in particular in more non-local scenarios.\\
To understand what it means we need to take a step back to the two parties of the protocol, Alice e Bob. Assume that they are not able to communicate but have access to physical states that they can use to generate joint correlations. In this experiment the outcome of the measurements on the state of their local systems are given by random variables. Obviously, causality constrains the correlations to be non signalling, and on the other side quantum theory prevents the strength of the non-local correlation to violate Bell's inequalities \cite{bell}, where the maximal value is known as Cirel'son's bound \cite{cirelson}. A well-known variant of a Bell inequality is the Clauser, Horne, Shimony \& Holt (CHSH) inequality \cite{CHSH}, which can be expressed as \cite{vanDam}
\begin{equation*}
	\sum_{x,y\in\{0,1\}}\mathsf{Pr}(a_x\oplus b_y=x\cdot y)\le2\sqrt{2}\,.
\end{equation*}
Where $x$ and $y$ denote the choice of Alice's and Bob's measurement, respectively, $a_x\in\{0,1\}$, $b_y\in\{0,1\}$ the respective binary outcomes, and $\oplus$ addition modulo 2.\\
However, if we only care about the causality constrain, Cirel'son's bound can be violated up to the maximal value of 4. Popescu and Rohrlich, who first noticed \cite{PRorigial}, raised the question of why nature is not more non-local and why does quantum mechanics not allow a stronger violation of the CHSH inequality.\\
Following this lead, bit commitment has been studied in Popescu Rohrlich (PR) non-local boxes theory.\\

It has been claimed, as in Ref.~\cite{wolf}, that bit commitment was admissible in PR-box theory but a counter-proof reached from Short, Gisin, and Popescu \cite{short-gisin-popescu}. It was formulated on the behalf that, since non-locality is the main reason to prevent bit commitment in quantum theory, it would not be possible that in a theory more non-local than the quantum one BC would work. However, thereafter a new protocol was proposed by Buhrman \textit{et al.} \cite{Buhrman_2006} and the argument of the counter-proof of Short, Gisin and Popescu is not able to deny it.\\

The first aim of our work is to bring some clarity, pointing out the limitations of the framework of validity about the results that have been claimed until now. In our path new aspects of PR-box theory emerged and have been studied. Even if the theory is still far to be considered complete, important characterizations have been made.\\
We will show that all the BC protocols proposed so far respect the limitations of PR-box theory under which we are able to prove a no-bit commitment theorem. Furthermore we are able to create, with very similar arguments, a counter-proof of the protocol proposed in Ref.~\cite{Buhrman_2006}.\\

Finally, a surprising result is observed. Relaxing the constraint of PR-box theory and including some recent developments, a scheme of bit commitment that is perfectly secure seems possible. Even if there are actually no "operational" evidences within the theory to deny it, by the same fundamental reason expressed by Short, Gisin, and Popescu, we think that future advancements of the theory would reconsider the protocol as cheatable. However, as a matter of fact, now the answer to the existence of a theory with entanglement that admits bit commitment, seems to be ``\textit{Yes}!''.\\

\noindent A synopsis of the thesis is the following:\\

\noindent In the first Chapter the framework of operational probabilistic theories (OPTs) is presented by first introducing the operational language that expresses the possible connections between events, and then by dressing the elements of the language with a probabilistic structure. After that, the principles for the OPT of quantum mechanics are stated (given the framework, the rule of connectivity among events are given).\\

\noindent In the second Chapter, we will discuss the bit commitment protocol. We will start with an historical perspective and then we will rigorously define the protocol in the language of OPT. We will mainly deal only with perfectly secure bit commitment since our analysis is carried in the OP framework and there are not yet the technical tools to analyse unconditionally secure BC in OPTs. In the third section of the Chapter, we will study the proof of the impossibility of perfectly secure bit commitment in quantum theory done in Ref.~\cite{chiribella}. This is a very elegant and solid demonstration and our purpose is to adapt the proof in order to comprehend other theories than the quantum one. As we will see, non-locality and entanglement are the key reasons that plays in favor of the impossibility of BC and so it seems reasonable to try to extend the impossibility proof also to other non-local theories. \\

\noindent In the third Chapter, we will analyse the probabilistic theory corresponding to the popular PR-box model. A comprehensive study has never been made and an organic theory is not available. We propose to bring clarity to the actual model considering only bipartite correlated boxes, highlighting its limitations, and some progress by analysing new aspects such as perfect discriminability and states purification. We will also add some considerations and prevision on $N$-partite correlated boxes.\\

\noindent Finally, in the fourth Chapter we propose a proof of impossibility of perfectly secure bit commitment in PR-boxes. Even if the PR-box model does not constitute a proper theory it has been often used in numerous application in literature, such as in protocol of perfectly or unconditionally secure BC. However this led to neglect important elements of the theory and in return numerous results published can be proved false when contextualized in the PR-box OPT (also if it is still incomplete). \\
As the last remark, we conjecture that including tripartite correlated boxes in the theory, secure bit commitment would be possible. The reason is the following. We will see that what prevents BC in quantum theory is that every state can be purified and its purification is unique up to local reversible transformations. Exactly these local operations allow the cheating in the protocol. When we deal with PR-box theory limited to bipartite correlated boxes, only one internal state can be purified, but its purification is unique and this is enough to allow cheating. On the contrary, when we admit tripartite correlated boxes, the uniqueness of purification is lost and exactly this uncommon feature could open the door to the possibility of secure bit commitment.

\thispagestyle{plain}

\chapterstyle{thesis}

\chapter{Operational Probabilistic Theory}
\label{chap:opt}
The purpose of this Chapter is to introduce the framework of operational probabilistic theories (OPTs) and to express quantum theory as an OPT. In this Chapter we will follow Ref.~\cite{d'ariano-libro}.\\

The framework of operational probabilistic theories consists of two distinct conceptual ingredients: an operational structure, describing circuits that produce outcomes, and a probabilistic structure, which assigns probabilities to the outcomes in a consistent way. The operational structure summarizes all the possible circuits that can be constructed in a given physical theory, in this setting a rigorous formulation of the elements of the circuits: systems, transformations, and their composition is given, which constitutes the grammar for the probabilistic description of an experiment. However, it is only the probabilistic structure that promotes the operational language from a merely descriptive tool to a framework for predictions, the predictive power being the crucial requirement for any scientific theory and for its testability - the essence of science itself. Different OPTs will have different rules for assigning the joint probabilities of events. Working in this framework allows us to deal with a wide range of probabilistic theories, including not only quantum and classical theories, but also the theory of Popescu-Rohrlich (PR), or non-local, boxes.\\

In the first part of the Chapter the framework is provided by first introducing the operational language that expresses the possible connections between events, and then by dressing the elements of the language with a probabilistic structure.\\
Then, in the second part of the Chapter we formulate the principles for the OPT of quantum theory. In fact, once the framework is defined, the rule of connectivity among events are given.\\

\section{The Framework}

A theory for making predictions about joint events depending on their reciprocal connections is what we call an operational probabilistic theory. We see that OPT is a non-trivial extension of probability theory.\\
To the joint events we associate not only their joint probability, but also a circuit that connects them. When the events in the circuit have a well-defined order, the circuit is mathematically described by a \textit{directed acyclic graph} (a graph with directed edges and without loops).\\
The basic element of an OPT - the notion of \textit{event} - gets dressed with \textit{wires} that allow us to connect it with other events. Such wires are the \textit{systems} of the theory. In agreement with the directed nature of the graph, there are input and output systems. The events are the \textit{transformations}, whereas the transformations with no input system are the states (corresponding to preparations of systems), and those with no output system are the effects (corresponding to observations of systems). Since the purpose of a single event is to describe a process connecting an input with an output, the full circuit associated to a probability is a closed one, namely a circuit with no input and no output.\\
The circuit framework is mathematically formalized in the language of \textit{category theory}. In this language, an OPT is a category, whose systems and events are objects and arrows, respectively. Every arrow has an input and an output object, and arrows can be sequentially composed. The associativity, existence of a trivial system, and commutativity of the parallel composition of systems of quantum theory technically correspond to having a \textit{strict symmetric monoidal category}.

\subsection{Primitive notions and notation}

The primitive notions of any operational theory are those of \textit{test}, \textit{event}, and \textit{system}. A test $\{\tA_i\}_{i\in \rI}$ is the collection of events $\tA_i$, where $i$ labels the element of the outcome space I. In addition to comprising a collection of events, the notion of test carries also the event connectivity of the theory that is achieved by the systems. These can represent the input and the output of the test. The resulting representation of a test is the following diagram:
\begin{equation*}
	\begin{aligned}
		\Qcircuit @C=1em @R=.7em {
			&
			\poloFantasmaCn{\rA}\qw&
			\gate{\{\tE_x\}_{x\in \rX}}&
			\poloFantasmaCn{\rB}\qw&
			\qw
		}\;.
	\end{aligned}
\end{equation*}
The wire on the left labeled as A represents the input system, whereas the wire on the right labeled as B represents the
output system. The same diagrammatic representation is also used for any of the events, namely for $x\in \rX$
\begin{equation*}
	\begin{aligned}
		\Qcircuit @C=1em @R=.7em {
			&
			\poloFantasmaCn{\rA}\qw&
			\gate{\tE_x}&
			\poloFantasmaCn{\rB}\qw&
			\qw
		}\;.
	\end{aligned}
\end{equation*}
In the following, the systems will be denoted by capital Roman letters A,B,...,Z, whereas the events by capital calligraphic letters $\mathcal{A,B,}...,\mathcal{Z}$.\\
Different tests can be combined in a \textit{circuit}, which is a directed acyclic graph where the links are the systems (oriented from left to right, namely from input to output) and the nodes are the boxes of the tests. The same graph can be built up for a single test instance, namely with the network nodes being events instead of tests, corresponding to a joint outcome for all tests.\\
The circuit graph is obtained precisely by using the following rules.	
\begin{description}
	\item[Sequential Composition of Test] When the output system of test $\{\mathcal{C}_x\}_{x\in \rX}$ and the input system of test $\{\tD_y\}_{y\in \rY}$ coincide, the two tests can be composed in sequence as follows:
	\begin{equation*}
		\begin{aligned}
			\Qcircuit @C=1em @R=.7em {
				&
				\poloFantasmaCn{\rA}\qw&
				\gate{\{\tC_x\}_{x\in \rX}}&
				\poloFantasmaCn{\rB}\qw&
				\gate{\{\tD_y\}_{y\in \rY}}&
				\poloFantasmaCn{\rC}\qw&
				\qw
			}
			\, \eqqcolon \,
			\Qcircuit @C=1em @R=.7em {
				&
				\poloFantasmaCn{\rA}\qw&
				\gate{\{\tE_{(x,y)}\}_{(x,y)\in \rX\times \rY}}&
				\poloFantasmaCn{\rC}\qw&
				\qw
			}\;,
		\end{aligned}
	\end{equation*}
	resulting in the test $\{\tE_{(x,y)}\}_{(x,y)\in \rX\times \rY}$ called \textit{sequential composition} of $\{\tC_x\}_{x\in \rX}$ and $\{\tD_y\}_{y\in \rY}$.
	In formulas we will also write $\tE_{(x,y)}:=\tD_y\tC_x$.
	
	\item[Identity Test] For every system A, one can perform the \textit{identity test} (shortly identity) that "leaves the system alone". Formally, this is the deterministic test $\mathcal{I}_\textnormal{A}$ with the property
	\begin{equation*}
		\begin{aligned}
			\Qcircuit @C=1em @R=.7em @! R {
				&
				\poloFantasmaCn{\rA}\qw&
				\gate{\mathcal{I}_\rA}&
				\poloFantasmaCn{\rA}\qw&
				\gate{\mathcal{C}}&
				\poloFantasmaCn{\rB}\qw&
				\qw
			}
		\end{aligned}
		\, =\,
		\begin{aligned}
			\Qcircuit @C=1em @R=.7em @! R {
				&
				\poloFantasmaCn{\rA}\qw&
				\gate{\mathcal{C}}&
				\poloFantasmaCn{\rB}\qw&
				\qw
			}
		\end{aligned}\,,
	\end{equation*}
	\begin{equation*}
		\begin{aligned}
			\Qcircuit @C=1em @R=.7em @! R {
				&
				\poloFantasmaCn{\rB}\qw&
				\gate{\mathcal{D}}&
				\poloFantasmaCn{\rA}\qw&
				\gate{\mathcal{I}_\rA}&
				\poloFantasmaCn{\rA}\qw&
				\qw
			}
		\end{aligned}
		\, =\, 
		\begin{aligned}
			\Qcircuit @C=1em @R=.7em @! R {
				&
				\poloFantasmaCn{\rB}\qw&
				\gate{\mathcal{D}}&
				\poloFantasmaCn{\rA}\qw&
				\qw
			}
		\end{aligned}\,,
	\end{equation*}
	where the above identities must hold for any event $ \Qcircuit @C=.5em @R=.5em { & \poloFantasmaCn{\rA} \qw & \gate{\mathcal{C}} & \poloFantasmaCn{\rB} \qw &\qw} $
	and $ \Qcircuit @C=.5em @R=.5em { & \poloFantasmaCn{\rB} \qw & \gate{\mathcal{D}} & \poloFantasmaCn{\rA} \qw &\qw} $,respectively. The sub-index A will be dropped from $\mathcal{I}_\text{A}$ where there is no ambiguity.
	
	\item[Operationally Equivalent Systems] We say that two systems A and $\text{A}^\prime$ are \textit{operationally equivalent} - denoted as $\text{A}^\prime\simeq \text{A}$ or just $\text{A}^\prime=\text{A}$ - if there exist two deterministic events $ \Qcircuit @C=.5em @R=.5em { & \poloFantasmaCn{\rA} \qw & \gate{\mathcal{U}} & \poloFantasmaCn{\rA^\prime} \qw &\qw} $ and	$ \Qcircuit @C=.5em @R=.5em { & \poloFantasmaCn{\rA^\prime} \qw & \gate{\mathcal{V}} & \poloFantasmaCn{\rA} \qw &\qw} $ such that
	\begin{equation*}
		\begin{aligned}
			\Qcircuit @C=1em @R=.7em @! R {& \poloFantasmaCn{\rA} \qw & \gate{\mathcal{U}} & 	\poloFantasmaCn{\rA^\prime} \qw & \gate{\mathcal{V}} & \poloFantasmaCn{\rA} \qw & \qw}
		\end{aligned}
		\ =\ 
		\begin{aligned}
			\Qcircuit @C=1em @R=.7em @! R {& \poloFantasmaCn{\rA} \qw & \gate{\mathcal{I}} & \poloFantasmaCn{\rA} \qw & \qw}
		\end{aligned}\,,
	\end{equation*}
	\begin{equation*}
		\begin{aligned}
			\Qcircuit @C=1em @R=.7em @! R {& \poloFantasmaCn{\rA^\prime} \qw & \gate{\mathcal{V}} & \poloFantasmaCn{\rA} \qw & \gate{\mathcal{U}} & \poloFantasmaCn{\rA^\prime} \qw & \qw}
		\end{aligned}
		\ =\ 
		\begin{aligned}
			\Qcircuit @C=1em @R=.7em @! R {& \poloFantasmaCn{\rA^\prime} \qw & \gate{\mathcal{I}} & \poloFantasmaCn{\rA^\prime} \qw & \qw}
		\end{aligned}\,.
	\end{equation*}
	Accordingly, if $\{\tC_x\}_{x\in \rX}$ is any test for system A, performing an equivalent test on system $\text{A}^\prime$ means performing the test $\{\tC^\prime_x\}_{x\in \rX}$ defined as
	
	\begin{equation*}
		\begin{aligned}
			\Qcircuit @C=1em @R=.7em @! R {& \poloFantasmaCn{\rA^\prime} \qw & \gate{\mathcal{C}^\prime_x} & \poloFantasmaCn{\rA^\prime} \qw & \qw}
		\end{aligned}
		\ =\ 
		\begin{aligned}
			\Qcircuit @C=1em @R=.7em @! R {& \poloFantasmaCn{\rA^\prime} \qw & \gate{\mathcal{V}} & \poloFantasmaCn{\rA} \qw & \gate{\mathcal{C}_x} & \poloFantasmaCn{\rA} \qw & \gate{\mathcal{U}} & \poloFantasmaCn{\rA^\prime} \qw & \qw}
		\end{aligned}\,.
	\end{equation*}
	
	\item[Composite System] Given two systems A and B, one can join them into the single composite system AB. As a rule, the system AB is operationally equivalent to the system BA, and we will identify them in the following. This means that the system composition is commutative,
	\begin{equation*}
		\text{AB}=\text{BA}.
	\end{equation*}
	We will call a system \textit{trivial system}, reserving for it the letter I, if it corresponds to the
	identity in the system composition, namely
	\begin{equation*}
		\text{AI}=\text{IA}=\text{A}.
	\end{equation*}
	The trivial system corresponds to having no system, namely I carries no information.\\
	Finally we require the composition of systems to be associative, namely
	\begin{equation*}
		\text{A(BC)}=\text{(AB)C}.
	\end{equation*}
	In other words, if we iterate composition on many systems we always end up with	a composite system that only depends on the components, and not on the particular composition sequence according to which they have been composed. Systems then make an Abelian monoid. A test with input system AB and output system CD represents an interaction process (see the parallel composition of tests in the following).
	
	\item[Parallel Composition of Tests] Any two tests $ \Qcircuit @C=.5em @R=.5em { & \poloFantasmaCn{\rA} \qw & \gate{\{\mathcal{C}_x\}_{x\in \rX}} & \poloFantasmaCn{\rB} \qw &\qw} $ $ \Qcircuit @C=.5em @R=.5em { & \poloFantasmaCn{\rC} \qw & \gate{\{\mathcal{D}_y\}_{y\in \rY}} & \poloFantasmaCn{\rD} \qw &\qw} $ can be composed in parallel as follows:
	\begin{equation*}
		\begin{aligned}
			\Qcircuit @C=1em @R=.7em @! R {& \poloFantasmaCn{\rA} \qw & \gate{\{\mathcal{C}_x\}_{x\in \rX}} & \poloFantasmaCn{\rB} \qw &\qw\\ & \poloFantasmaCn{\rC} \qw & \gate{\{\mathcal{D}_y\}_{y\in \rY}} & \poloFantasmaCn{\rD} \qw &\qw}
		\end{aligned}
		\ =:\ 
		\begin{aligned}
			\Qcircuit @C=1em @R=.7em @! R {
				& \poloFantasmaCn{\rA\rC} \qw & \gate{\{\tF_{(x,y)}\}_{(x,y)\in \rX\times \rY}} & \poloFantasmaCn{\rB\rD} \qw &\qw }
		\end{aligned}\,.
	\end{equation*}
	The test $ \Qcircuit @C=.8em @R=.5em { & \poloFantasmaCn{\rA\rC} \qw & \gate{\{\tF_{(x,y)}\}_{(x,y)\in \rX\times \rY}} & \poloFantasmaCn{\rB\rD} \qw &\qw } $ is the \textit{parallel composition} of tests $ \Qcircuit @C=.5em @R=.5em { & \poloFantasmaCn{\rA} \qw & \gate{\{\mathcal{C}_x\}_{x\in \rX}} & \poloFantasmaCn{\rB} \qw &\qw} $ and $ \Qcircuit @C=.5em @R=.5em { & \poloFantasmaCn{\rC} \qw & \gate{\{\mathcal{D}_y\}_{y\in \rY}} & \poloFantasmaCn{\rD} \qw &\qw} $, where $\{\tF_{(x,y)}\}_{(x,y)\in \rX\times \rY}\equiv\{\tC_x\otimes\tD_y\}_{(x,y)\in \rX\times \rY}$. Parallel and sequential composition of tests commute, namely one has	
	\begin{equation}
		\label{c:parallel composition}
		\begin{aligned}
			\Qcircuit @C=1em @R=.7em @! R {
				& \poloFantasmaCn{\rA} \qw & \gate{\tC_z} & \poloFantasmaCn{\rB} \qw & \gate{\tA_x} & \poloFantasmaCn{\rC} \qw & \qw \\
				& \poloFantasmaCn{\rD} \qw & \gate{\tB_y} & \poloFantasmaCn{\rE} \qw & \gate{\tD_w} & \poloFantasmaCn{\rF} \qw & \qw \gategroup{1}{3}{2}{3}{.7em}{--} \gategroup{1}{5}{2}{5}{.7em}{--}}
		\end{aligned}
		\ =\ 
		\begin{aligned}
			\Qcircuit @C=1em @R=1.2em @! R {
				& \poloFantasmaCn{\rA} \qw & \gate{\tC_z} & \poloFantasmaCn{\rB} \qw & \gate{\tA_x} & \poloFantasmaCn{\rC} \qw & \qw \\
				& \poloFantasmaCn{\rD} \qw & \gate{\tB_y} & \poloFantasmaCn{\rE} \qw & \gate{\tD_w} & \poloFantasmaCn{\rF} \qw & \qw \gategroup{1}{3}{1}{5}{.7em}{--} \gategroup{2}{3}{2}{5}{.7em}{--}}
		\end{aligned}\,.
	\end{equation}
	
	When one of the two operation is the identity, we wull omit the identity box and drawn only a straight line:
	\begin{equation*}
		\begin{aligned}
			\Qcircuit @C=1em @R=.7em @! R {
				& \poloFantasmaCn{\rA} \qw & \gate{\tC_x} & \poloFantasmaCn{\rB} \qw &  \qw \\
				& \qw & \poloFantasmaCn{\rC} \qw & \qw & \qw
			}\,.
		\end{aligned}
	\end{equation*}
	
	Therefore, as a consequence of commutation between sequetial and parallel composition, we have the following identity:
	\begin{equation*}
		\begin{aligned}
			\Qcircuit @C=1em @R=.7em @! R {
				& \poloFantasmaCn{\rA} \qw & \gate{\tC_x} & \qw & \poloFantasmaCn{\rB} \qw & \qw \\
				& \poloFantasmaCn{\rC} \qw & \qw & \gate{\tD_y} & \poloFantasmaCn{\rD} \qw & \qw
			}
		\end{aligned}
		\ =\ 
		\begin{aligned}
			\Qcircuit @C=1em @R=1.2em @! R {
				& \poloFantasmaCn{\rA} \qw & \qw & \gate{\tC_x} & \poloFantasmaCn{\rB} \qw & \qw \\
				& \poloFantasmaCn{\rC} \qw & \gate{\tD_y} & \qw & \poloFantasmaCn{\rD} \qw & \qw
			}
		\end{aligned}\,.
	\end{equation*}
	
	\item[Preparation Tests and Observation Tests] Tests with a trivial input system are called \textit{preparation tests}, and tests with a trivial output system are called \textit{observation tests}. They will be represented as follows:
	\begin{equation*}
		\begin{aligned}
			\Qcircuit @C=1em @R=.7em @! R {& \prepareC
				{\{\rho_x\}_{x\in \rX}} & \poloFantasmaCn{\rB} \qw & \qw
			}
		\end{aligned}
		\ :=\ 
		\begin{aligned}
			\Qcircuit @C=1em @R=.7em @! R {    
				& \poloFantasmaCn{\rI} \qw & \gate{\{\rho_x\}_{x\in \rX}} & \poloFantasmaCn{\rB} \qw & \qw
			}
		\end{aligned}\,,
	\end{equation*}
	\begin{equation*}
		\begin{aligned}
			\Qcircuit @C=1em @R=.7em @! R { & & \poloFantasmaCn{\rA} \qw & \measureD{\{a_y\}_{y\in \rY}}
			}
		\end{aligned}
		\ :=\ 
		\begin{aligned}
			\Qcircuit @C=1em @R=.7em @! R {    
				& \poloFantasmaCn{\rA} \qw & \gate{\{a_y\}_{y\in \rY}} & \poloFantasmaCn{\rI} \qw & \qw
			}
		\end{aligned}\,.
	\end{equation*}
	The corresponding events will be called \textit{preparation events} and \textit{observation events}. In formulas we will also write $|\rho_i)_\text{A}$ to denote a preparation event and $(a_j|_\text{A}$ to denote an observation event.
	
	\item[Closed Circuits] Using the above rules we can build up \textit{closed circuits}, i.e. circuits with no input and no output system. An example is given by the following circuit:
	\begin{equation}
		\label{c:closed1}
		\begin{aligned}
			\Qcircuit @C=1em @R=.7em @! R {
				\multiprepareC{3}{\{\Psi_{i}\}}&
				\qw\poloFantasmaCn{\rA}&
				\multigate{1}{\{\tA_{j}\}}&
				\qw\poloFantasmaCn{\rB}&
				\gate{\{\tC_{l}\}}&
				\qw\poloFantasmaCn{\rC}&
				\multigate{1}{\{\tE_{n}\}}&
				\qw\poloFantasmaCn{\rD}&
				\multimeasureD{2}{\{\tG_{q}\}}
				\\
				\pureghost{\{\Psi_{i}\}}&
				\qw\poloFantasmaCn{\rE}&
				\ghost{\{\tA_{j}\}}&\qw\poloFantasmaCn{\rF}&
				\multigate{1}{\{\tD_{m}\}}&
				\qw\poloFantasmaCn{\rG}&
				\ghost{\{\tE_{n}\}}
				\\
				\pureghost{\{\Psi_{i}\}}&
				\qw\poloFantasmaCn{\rH}&
				\multigate{1}{\{\tB_{k}\}}&
				\qw\poloFantasmaCn{\rL}&
				\ghost{\{\tD_{m}\}}&\qw\poloFantasmaCn{\rM}&
				\multigate{1}{\{\tF_{p}\}}&
				\qw\poloFantasmaCn{\rN}&\pureghost{\{\tG_{q}\}}\qw
				\\
				\pureghost{\{\Psi_{i}\}}&
				\qw\poloFantasmaCn{\rO}&
				\ghost{\{\tB_{k}\}}&				
				\qw&
				\qw\poloFantasmaCn{\rP}&
				\qw&
				\ghost{\{\tF_{p}\}}
				\\
			}
		\end{aligned}
	\end{equation}
	where we omitted the probability spaces of each test.
	
	\item[Independent Systems] For any (generally open) circuit constructed according to the above rules we call a set of systems \textit{independent} if for each couple of systems in the set the	two are not connected by a unidirected path (i.e. following the arrow from the input to the output). For example, in Eq.~\eqref{c:closed1} the sets \{A,E\}, \{H,O\}, \{A,E,H,O\}, \{A,L\}, \{A,E,L,P\} are independent, whereas e.g. the sets \{A,M\}, \{A,B\}, \{A,E,N\} are not. A maximal set of independent systems is called a \textit{slice}.
\end{description}

We are now in position to move towards the general purpose of an operational probabilistic theory: predicting and accounting for the joint probability of events corresponding to a particular circuit of connections.\\
Given a closed circuit, as in Eq. \eqref{c:closed1}, we are left with just a joint probability distribution. Therefore, to a closed circuit of event as the following: 
\begin{equation}
	\label{c:closed2}
	\begin{aligned}
		\Qcircuit @C=1em @R=.7em @! R {
			\multiprepareC{3}{\Psi_{i}}&
			\qw\poloFantasmaCn{\rA}&
			\multigate{1}{\tA_{j}}&
			\qw\poloFantasmaCn{\rB}&
			\gate{\tC_{l}}&
			\qw\poloFantasmaCn{\rC}&
			\multigate{1}{\tE_{n}}&
			\qw\poloFantasmaCn{\rD}&
			\multimeasureD{2}{\tG_{q}}
			\\
			\pureghost{\Psi_{i}}&
			\qw\poloFantasmaCn{\rE}&
			\ghost{\tA_{j}}&\qw\poloFantasmaCn{\rF}&
			\multigate{1}{\tD_{m}}&
			\qw\poloFantasmaCn{\rG}&
			\ghost{\tE_{n}}
			\\
			\pureghost{\Psi_{i}}&
			\qw\poloFantasmaCn{\rH}&
			\multigate{1}{\tB_{k}}&
			\qw\poloFantasmaCn{\rL}&
			\ghost{\tD_{m}}&\qw\poloFantasmaCn{\rM}&
			\multigate{1}{\tF_{p}}&
			\qw\poloFantasmaCn{\rN}&\pureghost{\tG_{q}}\qw
			\\
			\pureghost{\Psi_{i}}&
			\qw\poloFantasmaCn{\rO}&
			\ghost{\tB_{k}}&				
			\qw&
			\qw\poloFantasmaCn{\rP}&
			\qw&
			\ghost{\tF_{p}}
			\\
		}
	\end{aligned}
\end{equation}
we will associate a joint probability $p(i,j,k,l,m,n,p,q)$ which we will consider as \textit{parametrically dependent} on the circuit, namely, for a different choice of events and/or different connections we will have a different joint probability.\\
Since we are interested only in the joint probabilities and their corresponding circuits, we will build up probabilistic equivalence classes, and define:

\begin{center}
	\textit{Two events from system A to system B are equivalent if they occur with the same joint probability with the other events within any circuit.}
\end{center}

We will call \textit{transformation from} A \textit{to} B - denoted as $\tA\in\mathsf{Transf}(\rA\rightarrow\rB)$ - the equivalence class of events from A to B that are equivalent in the above sense. Likewise we will call \textit{instrument} an equivalence class of tests, \textit{state} an equivalence class of preparation events, and \textit{effect} an equivalence class of observation events. We will denote the set of states of system A as $\mathsf{St}(\rA)$, and the set of its effects as $\mathsf{Eff}(\rA)$. Clearly, the input systems belonging to two different elements of an equivalence class will be operationally equivalent, and likewise for output systems.\\
We now can define an operational probabilistic theory as follows:

\begin{center}
	\textit{An operational probabilistic theory (OPT) is a collection of systems and transformations, along with rules for composition of systems and parallel and sequential composition of transformations. The OPT assigns a joint probability to each closed circuit.}
\end{center}

Therefore, in an OPT every test from the trivial system I to itself is a probability distribution $\{p_i\}_{i\in \rX}$ for the set of joint outcomes X, with $p(i):=p_i\in\left[0,1\right]$ and
$\sum_{i\in\rX}p(i)=1$. Compound events from the trivial system to itself are \textit{independent}, namely their joint probability is given by the product of the respective probabilities for both the parallel and the sequential composition, namely
\begin{equation*}
	\begin{aligned}
		\Qcircuit @C=1em @R=.7em @! R {& \prepareC{\rho_{i_1}} & \poloFantasmaCn{\rA} \qw & \measureD{a_{i_2}}\\
			& \prepareC{\sigma_{j_1}} & \poloFantasmaCn{\rB} \qw & \measureD{b_{j_2}}
		}
	\end{aligned}
	\ =\ 
	\begin{aligned}
		\Qcircuit @C=1em @R=.7em @! R {    
			\prepareC{\rho_{i_1}} & \poloFantasmaCn{\rA} \qw & \measureD{a_{i_2}} & \prepareC{\sigma_{j_1}} & \poloFantasmaCn{\rB} \qw & \measureD{b_{j_2}} 
		}
	\end{aligned}
	=p(i_1,i_2)q(j_1,j_2)
	\,.
\end{equation*}
A special case of OPT is the \textit{deterministic} OPT, where all probabilities are 0 or 1.

\subsection{States and effects}

Using the parallel and sequential composition of transformation it follows that any closed circuit can be regarded as the composition of a preparation event and an observation event, for example the circuit in Eq.~\eqref{c:closed2} can be cut along a slice as follows:

\begin{equation*}
	\begin{aligned}
		\Qcircuit @C=1em @R=.7em @! R {
			\multiprepareC{3}{\Psi_{i}}&
			\qw\poloFantasmaCn{\rA}&
			\multigate{1}{\tA_{j}}&
			\qw\poloFantasmaCn{\rB}&
			\qw
			\\
			\pureghost{\Psi_{i}}&
			\qw\poloFantasmaCn{\rE}&
			\ghost{\tA_{j}}&\qw\poloFantasmaCn{\rF}&
			\multigate{1}{\tD_{m}}&
			\qw\poloFantasmaCn{\rG}&
			\\
			\pureghost{\Psi_{i}}&
			\qw\poloFantasmaCn{\rH}&
			\multigate{1}{\tB_{k}}&
			\qw\poloFantasmaCn{\rL}&
			\ghost{\tD_{m}}&\qw\poloFantasmaCn{\rM}&
			\multigate{1}{\tF_{p}}&
			\qw
			\\
			\pureghost{\Psi_{i}}&
			\qw\poloFantasmaCn{\rO}&
			\ghost{\tB_{k}}&				
			\qw&
			\qw\poloFantasmaCn{\rP}&
			\qw&
			\ghost{\tF_{p}}
			\\
		}
	\end{aligned}
	\ + \
	\begin{aligned}
		\Qcircuit @C=1em @R=.7em @! R {
			& \qw\poloFantasmaCn{\rB}&
			\gate{\tC_{l}}&
			\qw\poloFantasmaCn{\rC}&
			\multigate{1}{\tE_{n}}&
			\qw\poloFantasmaCn{\rD}&
			\multimeasureD{2}{\tG_{q}}
			\\
			& & & \qw\poloFantasmaCn{\rG}&
			\ghost{\tE_{n}}
			\\
			& & & & & \qw\poloFantasmaCn{\rN}&
			\pureghost{\tG_{q}}\qw
			\\
			&&&&&&
			\\
		}
	\end{aligned}
\end{equation*}
and thus is equivalent to the following state-effect circuit:
\begin{equation*}
	\begin{aligned}
		\Qcircuit @C=1em @R=.7em @! R {& \prepareC{(\Psi_i,\tA_j,\tB_k,\tD_m,\tF_p)} & \poloFantasmaCn{\rB\rG\rN} \qw & \measureD{(\tC_l,\tE_n,\tG_q)}
		}
	\end{aligned}
\end{equation*}
Therefore, a state $\rho\in\st{A}$ is a functional over effects $\eff{A}$, the functional being denoted with the pairing $(a|\rho)$ with $a\in\eff{A}$ and analogously an effect $a\in\eff{A}$ is a functional over states $\st{A}$.\\
By taking linear combinations of functionals we see that $\st[R]{A}:=\mathsf{Span}_\mathbb{R}\left[\st{A}\right]$ and $\eff[R]{A}:=\mathsf{Span}_\mathbb{R}\left[\eff{A}\right]$ are dual spaces, and states are positive linear functionals over effects, and effects are positive linear functional over states ($\st[R]{A}$ and $\eff[R]{A}$ are assume finite dimensional and we denote as $D_\rA:=\text{dim}\st[R]{A}\equiv\text{dim}\eff[R]{A}$ also called size of system A). In the following we also denote by $\mathsf{St}_1(\rA)$ and $\mathsf{Eff}_1(\rA)$ the sets of deterministic states and effects, respectively.\\
According to the above definition, two states are different if and only if there exists an effect which occurrs on them with different joint probabilities. We also have that two effects are different if and only if there exists a state on which they have different probabilities.\\
In particular, given two states $\rho_0\ne\rho_1\in\st{A}$ we will say that an effect $a\in\eff{A}$ \textit{separates} the state $\rho_0$ and $\rho_1$ when $(\rho_1|a)\ne(\rho_0|a)$, namely when the effect occurs with different joint probabilities over the two states (the analogous relation holds for separable states respect to effects). Therefore we conclude that:
\begin{center}
	\textit{States are separating for effects and effects are separating for states.}
\end{center}
It is possible to demonstrate that in any convex OPT if two states (effects) $\rho_0$,$\rho_1\in\st{A}$ ($a_0,a_1\in\eff{A}$) are distinct, then one can discriminate them with error probability strictly smaller than $\frac{1}{2}$.

\subsection{Transformations}

From what we said before, the following circuit is a state of system BFHO:
\begin{equation*}
	\begin{aligned}
		\Qcircuit @C=1em @R=.7em @! R {
			\multiprepareC{3}{\Psi}&
			\qw\poloFantasmaCn{\rA}&
			\multigate{1}{\tA}&
			\qw\poloFantasmaCn{\rB}&
			\qw
			\\
			\pureghost{\Psi}&
			\qw\poloFantasmaCn{\rE}&
			\ghost{\tA}&\qw\poloFantasmaCn{\rF}&
			\qw
			\\
			\pureghost{\Psi}&
			\qw\poloFantasmaCn{\rH}&
			\qw
			\\
			\pureghost{\Psi}&
			\qw\poloFantasmaCn{\rO}&				
			\qw&	
			\\
		}
	\end{aligned}
\end{equation*}
This means that any transformation connected to some output systems of a state maps the state into another state of generally different systems. Thus, while states and effects are linear functionals over each other, we can always regard a transformation as a map between states. In particular, a transformation $\tT\in\transf{A}{B}$ is always associated to a map $\hat{\tT}$ from $\st{A}$ to $\st{B}$, uniquely defined as
\begin{equation*}
	\hat{\tT}:|\rho)\in\st{A}\mapsto\hat{\tT}|\rho)=|\tT\rho)\in\st{B}.
\end{equation*}
Similarly the transformation can be associated to a map from $\eff{A}$ to $\eff{B}$. The map $\hat{\tT}$ can be linearly extended to a map from $\st[R]{A}$ to $\st[R]{B}$. Notice that the linear
extension of $\tT$ (which we will denote by the same symbol) is well defined. In fact, a linear combination of states of A is null - in formula $\sum_{i}c_i|rho_i)=0$ - if and only if $\sum_{i}c_i(a|rho_i)=0$ for every $a\in\eff{A}$, and since for every $b\in\eff{B}$ we have $(b|\tT\in\eff{A}$, then $(b|\tT(\sum_{i}c_i|\rho_i))=\sum_{i}c_i(b|\tT|\rho_i)=(b|\sum_{i}c_i\tT|\rho_i)=0$, and finally
$\sum_{i}c_i\tT|\rho_i)=0$.\\
We want to stress that if two transformations $\tT,\tT^\prime\in\transf{A}{B}$ correspond to the same map $\hat{\tT}$ from $\st{A}$ to $\st{B}$, this does not mean that the two transformations are the same, since as an equivalence class, they must occur with the same joint probability in all possible circuits. In terms of state mappings, the same definition of the transformation as equivalence class corresponds to say that $\tT,\tT^\prime\in\transf{A}{B}$ as maps from states of AR to states of BR are the same for all possible systems R of the theory, namely $\tT=\tT^\prime\in\transf{A}{B}$ if and only if 
\begin{equation}
	\label{eq:equal maps}
	\forall\rR,\;\forall\Psi\in\st{AR}\quad
	\begin{aligned}
		\Qcircuit @C=1em @R=.7em @! R {
			\multiprepareC{1}{\Psi}&
			\qw\poloFantasmaCn{\rA}&
			\gate{\tT}&
			\qw\poloFantasmaCn{\rB}&
			\qw
			\\
			\pureghost{\Psi}&
			\qw\poloFantasmaCn{\rR}&
			\qw&\qw&\qw
		}
	\end{aligned}
	\, = \,
	\begin{aligned}
		\Qcircuit @C=1em @R=.7em @! R {
			\multiprepareC{1}{\Psi}&
			\qw\poloFantasmaCn{\rA}&
			\gate{\tT^\prime}&
			\qw\poloFantasmaCn{\rB}&
			\qw
			\\
			\pureghost{\Psi}&
			\qw\poloFantasmaCn{\rR}&
			\qw&\qw&\qw
		}
	\end{aligned}\,.
\end{equation}
Indeed, there exist cases of OPT where there are transformations $\tT=\tT^\prime\in\transf{A}{B}$ corresponding to the same map when applied to $\st{A}$ and not when applied to $\st{AR}$ for some system R, a relevant example is fermionic theory \cite{fermionic1,fermionic2}.\\
Since we can take linear combinations of linear transformations, $\transf{A}{B}$ can be embedded in the vector space $\transf{A}{B}$. The deterministic transformations, whose set will be denoted as $\mathsf{Transf}_1(\rA\rightarrow\rB)$, will be also called channels.\\
Finally, a transformation $\tU\in\transf{A}{B}$ is \textit{reversible} if there exists another transformation $\tU^{-1}\in\transf{B}{A}$ such that $\tU^{-1}\tU = \tI_\rA$ and $\tU\tU^{-1}=\tI_\rB$. The set of reversible transformations from A to B will be denoted by $\textsf{RevTransf}(\rA\rightarrow\rB)$.

\subsection{Coarse-graining and refinement}

When dealing with probabilistic events, a natural notion is that of \textit{coarse-graining}, corresponding to merging events into a single event. According to probability theory, the probability of a coarse-grained event $\rS\subseteq\rX$ subset of the outcome space X is the sum of probabilities of the elements of S, namely $p(\rS)=\sum_{i\in\rS}p(i)$. We then correspondingly have that the coarse-grained event $\tT_\rS$ of a test $\{\tT_i\}_{i\in\rX}$ will be given by 
\begin{equation}
	\label{eq:coarse-grained event}
	\tT_\rS=\sum_{i\in\rS}\tT_i.
\end{equation}
We stress that the equal sign in Eq.~\eqref{eq:coarse-grained event} is to be meant in the sense of Eq.~\eqref{eq:equal maps}. In addition to the notion of coarse-grained event we have also that of coarse-grained test, corresponding to the collection of a coarse-grained events $\{\tT_{\rX_l}\}_{l\in\rZ}$ from a partition $\rX=\cup_{l\in\rZ}\rX_l$ of the outcome space $\rX$, with $\rX_i\cap\rX_j=\emptyset$ for $i\neq j$.\\
The converse procedure of coarse-graining is what we call \textit{refinement}. If $\tT_\rS$ the coarse-graining in Eq.~\eqref{eq:coarse-grained event}, we call any sum $\sum_{i\in\rS^\prime}\tT_i$ with $\rS^\prime\subseteq\rS$ a refinement of $\tT_\rS$. The same notion can be analogously considered for a test. Intuitively, a test that refines another is a test that extracts more detailed information, namely it is a test with better "resolving power".\\
The notion of refinement is translated to transformations (hence also to states, and effects), as equivalence classes of events. Refinement and coarse-graining define a partial ordering in the set of transformations $\transf{A}{B}$, writing $\tD\prec\tC$ if $\tD$ is a refinement of $\tC$. A transformation $\tC$ is atomic if it has only trivial refinement, namely $\tC_i$ refines $\tC$ implies that $\tC_i=p\tC$ for some probability $p\ge0$. A test that consists of atomic transformations is a test whose "resolving power" cannot be further improved.\\
It is often useful to refer to the set of all possible refinements of a given event $\tC$. This set is called refinement set of the event $\tC\in\transf{A}{B}$, and is denoted by $\textsf{RefSet}(\tC)$.
In formula, $\textsf{RefSet}(\tC):=\{\tD\in\transf{A}{B}|\tD\prec\tC\}.$\\
In the special case of states, we will use the word \textit{pure} as a synonym of atomic. A pure state describes an event providing maximal knowledge about the system's preparation, namely a knowledge that cannot be further refined (we will denote with $\mathsf{PurSt}(\rA)$ the set of pure states of system $\rA$).\\
As usual, a state that is not pure will be called \textit{mixed}. An important notion is that of \textit{internal state}. A state is called internal when any other state can refine it: precisely, $\omega\in\st{A}$ is internal if for every $\rho\in\st{A}$ there is a non-zero probability $p>0$ such that $p\rho$ is a refinement of $\omega$, i.e. $p\rho\in\mathsf{RefSet}(\omega)$. The adjective "internal" has a precise geometric connotation, since the state cannot belong to the border of $\st{A}$. An internal state describes a situation in which there is no definite knowledge about the system preparation, namely a priori we cannot in principle exclude any possible preparation.

\section{Quantum Theory as an OPT}
\label{sec:axioms}
In this section we provide an overview of the six principles used for constructing quantum theory as an OPT. All features of quantum theory - ranging from the superposition principle, entanglement, no cloning, teleportation, Bell's inequalities violation, quantum cryptography - can be understood and
proved using only the principles, without using Hilbert spaces. However, our aim is only to introduce the principle and analyse them from an operative point of view.\\
All the six principles are operational, in that they stipulate whether or not certain tasks can be accomplished: they set the rules of the game for all the experiments and all the protocols that can be carried out in the theory. They also provide a great insight into the worldview at which quantum theory hints.\\
We review the list of the principles:
\begin{enumerate}
	\item Atomicity of composition
	\item Perfect discriminability
	\item Ideal compression
	\item Causality
	\item Local discriminability
	\item Purification
\end{enumerate}
All six principles, with the exception of purification, express standard features that are shared by both classical and quantum theory. The principle of purification picks up uniquely quantum theory among the theories allowed by the first five, partly explaining the magic of quantum information.

\subsection{Atomicity of composition}

In the general framework we encountered the notions of coarse-grained and atomic operation. A coarse-grained operation is obtained by joining together outcomes of a test, corresponding to neglect some information. The inverse process of coarse-graining is that of refining. An atomic operation is one where no information has been neglected, namely an operation that cannot be refined. When the operation is atomic, the experimenter has maximal knowledge of what’s happening in the lab. A test consisting of atomic operations represents the highest level of control achievable according to our theory.\\
The principle of atomicity of composition states that it possible to maintain such a level of control throughout a sequence of experiments, stating precisely what follows:
\begin{axiom}[Atomicity of composition]
	\label{axiom:atomicity}
	The sequence of two atomic operations is an atomic operation.
\end{axiom}
One of the immediate consequences granted by atomicity of composition is the following:
\begin{corollary}[Parallel composition of pure states]
	\label{cor:parallelcompositionofpurestates}
	Given two pure states $\alpha\in\st{A}$ and $\beta\in\st{B}$, the parallel composition of $\alpha$ and $\beta$ is a pure state of $\st{AB}$.
\end{corollary}

\subsection{Perfect discriminability}

Two deterministic states $\rho_0$ and $\rho_1$ are perfectly discriminable if there exists a measurement $\{m_y\}_{y\in\{0,1\}}$ such that $$(m_y|\rho_x)=\delta_{xy}\quad\forall x,y\in\{0,1\}.$$
The existence of perfectly discriminable states is important, because these states can be used to communicate classical information without errors. In a communication protocol, the sender can encode the value of a bit $x$ into the state $\rho_x$ and then transmit the system to the receiver, who can decode the value of the bit using the measurement $\{m_y\}_{y\in\{0,1\}}$.\\
The perfect discriminability axiom ensures that our ability to discriminate states is as sharp as it could possibly be: except for trivial cases, every state can be perfectly discriminated from some other state. The "trivial cases" are those states that cannot be discriminated from anything else because they contain every other state in their convex decomposition. We can call them internal, or completely mixed.
\begin{axiom}[Perfect discriminability]
	Every deterministic state that is not completely mixed is perfectly discriminable from some other state.
\end{axiom}
As anticipated, the perfect discriminability axiom guarantees that every non-trivial
system has at least two perfectly discriminable states:
\begin{prop}
	In a theory satisfying perfect discriminability, every physical system has
	at least two perfectly discriminable states, unless the system is trivial (i.e. it has only one
	deterministic state).
\end{prop}
\begin{proof}
	Pick a pure state $\alpha\in\st{A}$. If $\alpha$ is not internal, then perfect discriminability guarantees that $\alpha$ is perfectly discriminable from some other state $\alpha^\prime$, hence A has two perfectly discriminable states. If $\alpha$ is internal every pure state belongs to its refinement set.	Moreover, since it is also pure, i.e. extremal, one has that every other deterministic state	$\rho_1\in\mathsf{St}_1(\rA)$ must be equal to $\alpha$, i.e. A has only one deterministic state.
\end{proof}
An easy consequence of this result is that the theory can describe noiseless classical communication.

\subsection{Ideal compression}

Ideal compression garantes that information can be transferred faithfully from one system to another. Namely, suppose that Alice has a preparation device, which prepares system A in some state $\alpha$. Alice does not know the state $\alpha$, but she knows that on average the device prepares the deterministic state $\rho\in\mathsf{St}_1(\rA)$. Now, suppose Alice wants to transfer the state of her system to Bob's laboratory, but unfortunately she cannot send system A directly. Instead, she has to encode the state $\alpha$ into the state of another system B, by applying a suitable deterministic operation $\tE$ (the \textit{encoding}), which transforms the state $\alpha$ into the state $$\beta:=\tE\alpha.$$
We say that the encoding is \textit{lossless} for the state $\rho$ iff there exists another deterministic operation $\tD$ (the \textit{decoding}) such that $$\tD\tE\alpha=\alpha\quad\forall\alpha\in F_\rho$$ where $F_\rho$ is the refinement set of $\rho$, which is made of the set of all states $\alpha$ that are compatible with $\rho$ (on the convex set of states this would be the face to which $\rho$ belongs).\\
This third axiom establishes the possibility of a particular type of lossless encoding, called \textit{ideal compression}. The ultimate limit to the lossless compression of a given state $\rho$ is reached when every state of the encoding system B is a codeword for some state in $F_\rho$, namely
when every state $\beta\in\st{B}$ is of the form $\tE\alpha$ for some $\alpha\in F_\rho$. When this is the case, we say that the compression is \textit{efficient}, and we call the triple $(\rB,\tE,\tD)$ an \textit{ideal compression protocol}.

\begin{axiom}[Ideal Compression]
	Every state can be compressed in a lossless and efficient way.
\end{axiom}

\subsection{Causality}

The causality axiom identifies the input–output ordering of a circuit with the direction along which information flows, identifying such ordering with a proper-time arrow, corresponding to the request that future choices cannot influence the present.
\begin{axiom}[Causality]
	\label{axiom:causality}
	The probability of the outcome of a preparation test is independent of the choice of observation tests connected at its output.
\end{axiom}
To better understand the statement it is useful to consider the joint test consisting of a preparation test $\tX=\{\rho_i\}_{i\in\rX}\subset\st{A}$ followed by the observation test $\tY=\{a_j\}_{j\in\rY}\subset\eff{A}$ performed on system A:
\begin{equation*}
	\begin{aligned}
		\Qcircuit @C=1em @R=.7em @! R { 
			\prepareC{\tX}&\qw&
			\qw\poloFantasmaCn{\rA}&\qw&
			\measureD{\tY}
		}\,.
	\end{aligned}
\end{equation*}
The joint probability of preparation $\rho_i$ and observation $a_j$ is given by
\begin{equation*}
	\begin{aligned}
		p(i,j|\tX,\tY)\coloneqq(a_j|\rho_i)\equiv
		\Qcircuit @C=1em @R=.7em @! R { 
			\prepareC{\rho_i}&\qw&
			\qw\poloFantasmaCn{\rA}&\qw&
			\measureD{a_j}
		}\,.
	\end{aligned}
\end{equation*}
The marginal probability of the preparation alone does not depend on the outcome $j$. Yet, it generally depends on which observation test $\tY$ is performed, namely $$\sum_{a_j\in\rY}(a_j|\rho_i)\eqqcolon p(i|\tX,\tY)\,.$$
The marginal probability of preparation $\rho_i$ is then generally conditioned on the choice of the observation test $\tY$. What the causality axiom states is that $p(i|\tX,\tY)$ is actually independent of $\tY$, namely for any two different observation tests $\tY=\{a_j\}_{j\in\rY}$ and $\tZ=\{b_k\}_{k\in\rZ}$ one has $$p(i|\tX,\tY)=p(i|\tX,\tZ)=p(i|\tX)\,.$$

In a causal OPT the choice of a test on a system can be conditioned on the outcomes of a preceding test, since causality guarantees that the probability distribution of the preceding test is independent of the choice of the following test. This leads us to introduce the notion of \textit{conditioned test}.
\begin{defn}[Conditioned test]
	If $\{\tA_i\}_{i\in\rX}$ is a test from $\rA$ to $\rB$, and $\{\tB_j^{(i)}\}_{j\in\rY_i}$ is a test from $\rB$ to $\rC$ for every $i\in\rX$, then the conditioned test is a test from $\rA$ to $\rC$, with outcomes $(i,j)\in\rZ:=\cup_i\{\{i\}\times\rY_i\}$, and events $\{\tB_j^{(i)}\circ\tA_i\}_{(i,j)\in\rZ}$. Diagrammatically, the events $\tB_j^{(i)}\circ\tA_i$ are represented as follows:
	\begin{equation*}
		\begin{aligned}
			\Qcircuit @C=1em @R=.7em @! R { &
				\qw\poloFantasmaCn{\rA}&
				\gate{\tA_i}&
				\qw\poloFantasmaCn{\rB}&
				\gate{\tB_j^{(i)}}&
				\poloFantasmaCn{\rC}\qw&
				\qw
			}\,.
		\end{aligned}
	\end{equation*}
\end{defn}
Among conditioned test, a special role is played by the \textit{observe-and-prepare test},where the "connecting" system is the null system I. They are thus made of a preparation test conditioned by an observation test, as follows:
\begin{equation*}
	\begin{aligned}
		\Qcircuit @C=.5em @R=.7em @! R { &
			\qw\poloFantasmaCn{\rA}&
			\measureD{l_i}&
			&\prepareC{\omega^{(i)}}&
			\poloFantasmaCn{\rC}\qw&
			\qw
		}\,,
	\end{aligned}
\end{equation*}
which can be also represented as $\{|\omega^{(i)})(l_i|\}_{i\in\rX}$.\\

Another remarkable way to characterize causal theories is to require the unicity of the deterministic effect, the equivalence of this formulation with Axiom \ref{axiom:causality} is given by the following lemma.
\begin{lemma}
	\label{lemma:unique deterministic effect}
	An OPT is causal if and only if for every system $\rA$ there is a unique deterministic effect.
\end{lemma}
\begin{proof}
	We will prove the two directions separately, namely: (1) if the probability of preparation of states is independent of the observation test, then the deterministic effect is unique; (2) vice versa. (1) The probability of the preparation $\rho$ is given by the marginal of the joint probability with the observation, namely $p(\rho)=\sum_{i\in\rX}(a_i|\rho)$. Upon denoting the deterministic effects of two different tests as $a=\sum_{i\in\rX}a_i$ and $b=\sum_{j\in\rY}b_j$, the statement that the preparation probability is independent of the observation tests translates to $(a|\rho) = (b|\rho)$ for every preparation $\rho\in\st{A}$, which implies that $a=b$, since the set of states is separating for events. (2) Uniqueness of the deterministic effect implies that the
	preparation probability of each state is independent of the test, since the effect $a=\sum_{i\in\rX}a_i$ for any test $\{a_i\}_{i\in\rX}$ is deterministic, and $(a|\rho)$ for any deterministic effect $a\in\eff{a}$ is the probability of preparation $\rho$.
\end{proof}

We will denote the unique deterministic effect for system A as $e_\rA$, and the subindex will be dropped when no confusion can arise.\\
In the following we will use the notation $\le$ to denote the partial ordering between effects, defined as follows: $$a,b\in\eff{A},\,a\le b\quad\Leftrightarrow\quad(a|\rho)\le(b|\rho),\,\forall\rho\in\st{A}\,.$$
It is immediate to show that the causality condition of Lemma \ref{lemma:unique deterministic effect} spawns the following lemmas.
\begin{lemma}
	Causality is equivalent to the following statements regarding tests:
	\begin{enumerate}
		\item \textbf{Completeness of observation tests:} For any system $\rA$ and for every observation test $\{a_i\}_{i\in\rX}$ one has $$\sum_{i\in\rX}a_i=e_\rA\,.$$
		\item \textbf{Completeness of tests:} For any systems $\rA$, $\rB$ and for every test $\{\tC_i\}_{i\in\rX}$ from $\rA$ to $\rB$ one has $$\sum_{i\in\rX}(e_\rB|\tC_i=(e_\rA\,.$$
		\item \textbf{Domination of transformations:} For any systems $\rA$, $\rB$ a transformation $\tC\in\transf{A}{B}$ satisfies the condition $$(e_\rB|\tC\le(e_\rA\,,$$ with the equality if and only if $\tC$ is a channel, i.e. a deterministic transformation corresponding to a single-outcome test.
		\item \textbf{Domination of effects:} For any system $\rA$ all effects are dominated by a unique effect $e_\rA$ which is deterministic$$\forall a\in\eff{A},\quad0\le a\le e_\rA\,.$$
	\end{enumerate}
\end{lemma}
An immediate consequence of uniqueness of the deterministic effect is the identification of all transformations of the form $$\forall\tB\in\transf{B}{C},\quad\sum_{i\in\rX}a_i\otimes\tB=e_\rA\otimes\tB,$$for any observation test $\{a_i\}_{i\in\rX}$ of system A. In particular, we have the factorization of the deterministic effect of composite systems $$e_{\rA\rB}=e_\rA\otimes e_\rB.$$

The uniqueness of the deterministic effect naturally leads to the relevant notion of \textit{marginal state} or also called \textit{local state}.
\begin{defn}[Marginal state]
	The marginal state of $|\sigma)_{\rA\rB}$ on system $\rA$ is the state $$|\rho)_\rA:=(e|_\rB|\sigma)_{\rA\rB}$$represented by the diagram
	\begin{equation*}
		\begin{aligned}
			\Qcircuit @C=1.2em @R=.7em @! R { 
				\multiprepareC{1}{\sigma_{\rA\rB}}&
				\poloFantasmaCn{\rA}\qw&
				\qw
				\\
				\pureghost{\sigma_{\rA\rB}}&
				\poloFantasmaCn{\rB}\qw&
				\measureD{e_\rB}
				\\
			}
		\end{aligned}\, \eqqcolon \,
		\begin{aligned}
			\Qcircuit @C=1.2em @R=.7em @! R { 
				\prepareC{\rho_{\rA}}&
				\poloFantasmaCn{\rA}\qw&
				\qw
			}\,.
		\end{aligned}
	\end{equation*}
\end{defn}

Finally, the last implication of causality we would outline is the impossibility of signaling without interaction, i.e. by just performing local tests.
\begin{theorem}[No signaling without interaction]
	In a causal OPT it is impossible to send signals by performing only local tests.
\end{theorem}
\begin{proof}
	Suppose the general situation in which two "distant" parties Alice and Bob share a bipartite state $|\Psi)_{\rA\rB}$ of systems A and B. Alice performs her local test $\{\tA_i\}_{i\in\rX}$ on system A and similarly Bob performs his local test $\{\tB_j\}_{j\in\rY}$ on system B. The joint probability of their outcomes is $$p_{ij}=(e|_{\rA\rB}(\tA_i\otimes\tB_j)|\Psi)_{\rA\rB}.$$The marginal probabilities $p_i^\rA$ at Alice and $p_j^\rB$ at Bob are given by $$p_i^\rA\coloneqq\sum_{j}p_{ij},\quad p_j^\rB\coloneqq\sum_{i}p_{ij}.$$Alice's marginal does not depend on the choice of test $\{\tB_j\}$ of Bob, since
	\begin{equation*}
		\begin{aligned}
			p_i^\rA=&\sum_{j}(e|_{\rA}(e|_{\rB}(\tA_i\otimes\tB_j)|\Psi)_{\rA\rB}=(e|_{\rA}\left(\tA_i\otimes\left[\sum_{j}(e|_{\rB}\tB_j\right]\right)|\Psi)_{\rA\rB}\\
			&(e|_{\rA}\tA_i|\rho)_{\rA},\quad\quad|\rho)_\rA\coloneqq(e|_\rB|\Psi)_{\rA\rB},
		\end{aligned}
	\end{equation*}
	where we used Eq.~\eqref{c:parallel composition} and the normalization condition $\sum_{j}(e|_{\rB}\tB_j=(e|_{\rB}$. The same argument holds for Bob’s marginal.
	
\end{proof}

\subsection{Local discriminability}

Now we introduce the principle of local discriminability, which stipulates the possibility of discriminating states of composite systems via local measurements on the component systems.
\begin{axiom}[Local discriminability]
	It is possible to discriminate any pair of states of composite systems using only local measurements.
\end{axiom}
Mathematically the axiom asserts that for every two joint states $\rho,\sigma\in\st{AB}$, with $\rho\ne\sigma$, there exist effects $a\in\eff{A}$ and $b\in\eff{B}$ such that the joint probabilities for the two states are different, namely, in circuits
\begin{equation}
	\label{eq:loc-discr}
	\begin{aligned}
		\Qcircuit @C=1.2em @R=.7em @! R { 
			\multiprepareC{1}{\rho}&
			\poloFantasmaCn{\rA}\qw&
			\qw
			\\
			\pureghost{\rho}&
			\poloFantasmaCn{\rB}\qw&
			\qw
			\\
		}
	\end{aligned}
	\,\ne\,
	\begin{aligned}
		\Qcircuit @C=1.2em @R=.7em @! R { 
			\multiprepareC{1}{\sigma}&
			\poloFantasmaCn{\rA}\qw&
			\qw
			\\
			\pureghost{\sigma}&
			\poloFantasmaCn{\rB}\qw&
			\qw
			\\
		}
	\end{aligned}
	\, \Longrightarrow \, 
	\begin{aligned}	
		\Qcircuit @C=1.2em @R=.7em @! R { 
			\multiprepareC{1}{\rho}&		
			\poloFantasmaCn{\rA}\qw&
			\measureD{a}
			\\
			\pureghost{\rho}&
			\poloFantasmaCn{\rB}\qw&
			\measureD{b}
			\\
		}
	\end{aligned}
	\,\ne\,
	\begin{aligned}	
		\Qcircuit @C=1.2em @R=.7em @! R { 
			\multiprepareC{1}{\sigma}&		
			\poloFantasmaCn{\rA}\qw&
			\measureD{a}
			\\
			\pureghost{\sigma}&
			\poloFantasmaCn{\rB}\qw&
			\measureD{b}
			\\
		}
	\end{aligned}\,.
\end{equation}
We can now prove one of the main theorem following from the principle of local discriminability.
\begin{theorem}[Product rule for composite system]
	\label{theorem:product-rule-for-composite-systems}
	A theory satisfies local discriminability if and only if, for every composite system $\rA\rB$, one has
	\begin{equation}
		\label{eq:product-rule-for-composite-systems}
		D_{\rA\rB}=D_\rA D_\rB\,.
	\end{equation}
\end{theorem}
\begin{proof}
	By Eq.~\eqref{eq:loc-discr}, a theory satisfies local discriminability if and only if local effects $a\otimes b\in\eff{AB}$, with $a\in\eff{A}$ and $b\in\eff{B}$, are separating for joint states $\st{AB}$. Equivalently, the set $T\coloneqq\{a\otimes b|a\in\eff{A},b\in\eff{B}\}$ is a spanning set for $\mathsf{Eff}_\mathbb{R}(\rA\rB)$. Since the dimension of $\mathsf{Span}_\mathbb{R}(T)$ is $D_\rA D_\rB$ and the spaces of states and effects have the same dimension, we have $D_{\rA\rB}=D_\rA D_\rB$. Conversely, if Eq.~\eqref{eq:product-rule-for-composite-systems} holds, then the product effects are a spanning set for the vector space $\mathsf{Eff}_\mathbb{R}(\rA\rB)$, hence they are separating, and local discriminability holds.
\end{proof}
Along with the axiom of local discriminability we introduce the notion of \textit{entangled} and \textit{separable} states, where entangled states are defined, by negation, as those states that are not separable.
\begin{defn}[Separable states]
	Given $n$ systems $\rA_1,\rA_2,\dots,\rA_n$, the separable states of the composite system $\rA_1\rA_2\dots\rA_n$ are those of the form
	\begin{equation*}
		\begin{aligned}	
			\Qcircuit @C=1.2em @R=.7em @! R { 
				\multiprepareC{3}{\Sigma}&		
				\poloFantasmaCn{\rA_1}\qw&
				\qw
				\\
				\pureghost{\Sigma}&
				\poloFantasmaCn{\rA_2}\qw&
				\qw
				\\
				\pureghost{\Sigma}&
				\vdots&&
				\\
				\pureghost{\Sigma}&
				\poloFantasmaCn{\rA_n}\qw&
				\qw
				\\
			}
		\end{aligned}
		\,=\, \sum_{i\in\rX}p_i
		\begin{aligned}	
			\Qcircuit @C=1.2em @R=.7em @! R { 
				\prepareC{\alpha_{1_i}}&		
				\poloFantasmaCn{\rA}\qw&
				\qw
				\\
				\prepareC{\alpha_{2_i}}&		
				\poloFantasmaCn{\rA}\qw&
				\qw
				\\
				\vdots
				\\
				\prepareC{\alpha_{n_i}}&		
				\poloFantasmaCn{\rA}\qw&
				\qw
			}
		\end{aligned}
	\end{equation*}
	where for $j=1,2,\dots,n$, $\alpha_{j_i}\in\st{A_j}$ $\forall i\in\rX$.
\end{defn}

\subsection{Purification}

Purification is the really distinctive and fundamental trait of quantum theory, in the sense that purification allows to distinguish it between all the other possible theories (all the ones we can think of).\\
The statement of the axiom is the following.
\begin{axiom}[Purification]
	\label{axiom:purification}
	For every system $\rA$ and for every state $\rho\in\st{A}$, there exists a system $\rB$ and a pure state $\Psi\in\mathsf{PurSt}(\rA\rB)$ such that
	\begin{equation*}
		\begin{aligned}	
			\Qcircuit @C=1.2em @R=.7em @! R { 
				\prepareC{\rho}&
				\poloFantasmaCn{\rA}\qw&
				\qw
			}
		\end{aligned}
		\,=\, 
		\begin{aligned}	
			\Qcircuit @C=1.2em @R=.7em @! R { 
				\multiprepareC{1}{\Psi}&		
				\poloFantasmaCn{\rA}\qw&
				\qw
				\\
				\pureghost{\Psi}&		
				\poloFantasmaCn{\rB}\qw&
				\measureD{e}
				\\
			}
		\end{aligned}\,.
	\end{equation*}
	If two pure states $\Psi$ and $\Psi^\prime$ satisfy
	\begin{equation*}
		\begin{aligned}	
			\Qcircuit @C=1.2em @R=.7em @! R { 
				\multiprepareC{1}{\Psi^\prime}&
				\poloFantasmaCn{\rA}\qw&
				\qw\\
				\pureghost{\Psi^\prime}&
				\poloFantasmaCn{\rB}\qw&
				\measureD{e}\\
			}
		\end{aligned}
		\,=\, 
		\begin{aligned}	
			\Qcircuit @C=1.2em @R=.7em @! R { 
				\multiprepareC{1}{\Psi}&		
				\poloFantasmaCn{\rA}\qw&
				\qw
				\\
				\pureghost{\Psi}&		
				\poloFantasmaCn{\rB}\qw&
				\measureD{e}
				\\
			}
		\end{aligned}\,,
	\end{equation*}
	then there exists a reversible transformation $\tU$, acting only on system $\rB$, such that
	\begin{equation}
		\label{eq:uniq purification}
		\begin{aligned}	
			\Qcircuit @C=1.2em @R=.7em @! R { 
				\multiprepareC{1}{\Psi^\prime}&
				\poloFantasmaCn{\rA}\qw&
				\qw\\
				\pureghost{\Psi^\prime}&
				\poloFantasmaCn{\rB}\qw&
				\qw\\
			}
		\end{aligned}
		\,=\, 
		\begin{aligned}	
			\Qcircuit @C=1.2em @R=.7em @! R { 
				\multiprepareC{1}{\Psi}&
				\qw&		
				\poloFantasmaCn{\rA}\qw&
				\qw&
				\qw
				\\
				\pureghost{\Psi}&		
				\poloFantasmaCn{\rB}\qw&
				\gate{\tU}&
				\poloFantasmaCn{\rB}\qw&
				\qw\\
			}
		\end{aligned}\,\,.
	\end{equation}
\end{axiom}
Here we say that $\Psi$ is a \textit{purification} of $\rho$ and that $\rB$ is the \textit{purifying system}.\\
The property stated in Eq.~\ref{eq:uniq purification} is called \textit{uniqueness of purification} and refers to the case in which the two purifications have the same purifying system. It can be easily generalized (the purifying systems are different):
\begin{prop}
	If two pure states $\Psi\in\mathsf{PurSt}(\rA\rB)$ and $\Psi^\prime\in\mathsf{PurSt}(\rA\rB^\prime)$ are purifications of the same mixed state, then
	\begin{equation*}
		\begin{aligned}	
			\Qcircuit @C=1.2em @R=.7em @! R { 
				\multiprepareC{1}{\Psi^\prime}&
				\poloFantasmaCn{\rA}\qw&
				\qw\\
				\pureghost{\Psi^\prime}&
				\poloFantasmaCn{\rB^\prime}\qw&
				\qw\\
			}
		\end{aligned}
		\,=\, 
		\begin{aligned}	
			\Qcircuit @C=1.2em @R=.7em @! R { 
				\multiprepareC{1}{\Psi}&
				\qw&		
				\poloFantasmaCn{\rA}\qw&
				\qw&
				\qw
				\\
				\pureghost{\Psi}&		
				\poloFantasmaCn{\rB}\qw&
				\gate{\tC}&
				\poloFantasmaCn{\rB^\prime}\qw&
				\qw\\
			}
		\end{aligned}
	\end{equation*}
	for some deterministic transformation $\tC$ transforming system $\rB$ into system $\rB^\prime$.
\end{prop}
\begin{proof}
	Pick two pure states $\beta\in\mathsf{PurSt}(\rB)$ and $\beta^\prime\in\mathsf{PurSt}(\rB^\prime)$. Since $\Psi\otimes\beta^\prime$ and $\Psi^\prime\otimes\beta$	are purifications of the same state on $\rA$, the uniqueness of purification implies
	\begin{equation*}
		\begin{aligned}	
			\Qcircuit @C=1.2em @R=.7em @! R { 
				\multiprepareC{1}{\Psi^\prime}&
				\poloFantasmaCn{\rA}\qw&
				\qw\\
				\pureghost{\Psi^\prime}&
				\poloFantasmaCn{\rB^\prime}\qw&
				\qw\\
				\prepareC{\beta}&
				\poloFantasmaCn{\rB}\qw&
				\qw\\
			}
		\end{aligned}
		\,=\, 
		\begin{aligned}	
			\Qcircuit @C=1.2em @R=.7em @! R { 
				\multiprepareC{1}{\Psi}&
				\qw&		
				\poloFantasmaCn{\rA}\qw&
				\qw&
				\qw\\
				\pureghost{\Psi}&		
				\poloFantasmaCn{\rB}\qw&
				\multigate{1}{\tU}&
				\poloFantasmaCn{\rB^\prime}\qw&
				\qw\\
				\prepareC{\beta^\prime}&
				\poloFantasmaCn{\rB^\prime}\qw&
				\ghost{\tU}&
				\poloFantasmaCn{\rB}\qw&
				\qw\\
			}
		\end{aligned}
	\end{equation*}
	for some reversible transformation $\tU$ (we have also assumed local discriminability).\\
	Discarding system $\rB$ on both sides we then obtain $\Psi^\prime = (\tI_\rA\otimes\tC)\Psi$, where $\rC$ is the deterministic transformation defined by
	\begin{equation*}
		\begin{aligned}	
			\Qcircuit @C=1.2em @R=.7em @! R { & 
				\poloFantasmaCn{\rB}\qw&
				\gate{\tC}&
				\poloFantasmaCn{\rB^\prime}\qw&
				\qw
			}
		\end{aligned}
		\,\coloneqq \, 
		\begin{aligned}	
			\Qcircuit @C=1.2em @R=.7em @! R { &		
				\poloFantasmaCn{\rB}\qw&
				\multigate{1}{\tU}&
				\poloFantasmaCn{\rB^\prime}\qw&
				\qw\\
				\prepareC{\beta}&		
				\poloFantasmaCn{\rB^\prime}\qw&
				\ghost{\tU}&
				\poloFantasmaCn{\rB}\qw&
				\measureD{e}\\
			}
		\end{aligned}\,.
	\end{equation*}
\end{proof}


\chapter{Bit Commitment}
\label{chap:bit commitment}
Bit commitment is a cryptographic primitive involving two mistrustful parties, conventionally called Alice and Bob. Alice is supposed to submit an encoded bit of information to	Bob in such a way that Bob has (in principle) no chance to identify the bit before Alice later decodes it for him, whereas Alice has (in principle) no way of changing the value of the bit once she has submitted it. In other words, Bob is interested in \textit{binding} Alice to some commitment, whereas Alice would like to \textit{conceal} her commitment from Bob.\\

In the first two sections of this Chapter we will describe the protocol: we will start from an historical perspective, from the first article published by Blum in 1983 \cite{blum} to the most recent developments of the last decade that focus on the impossibility of bit commitment in quantum theory \cite{d'ariano2007,chiribella-d'ariano2009}. Then, we will rigorously define the protocol in the language of OPTs.\\

As every other cryptographic primitive, bit commitment does not need to be \textit{perfectly secure}, i.e. probability of cheating equals to 0 for Alice and 1/2 for Bob (who can always randomly guess), to be efficient. In fact, even if a greater probability of error is admitted, iterating the protocol the error probability can be generally asymptotically reduced. A protocol that admits this possibility is called \textit{unconditionally secure}, with a vast literature on the subject.\\
However in our thesis we will only deal with perfectly secure bit commitment. Since we are considering the protocol within the operational framework, there are not yet the technical tools to analyse unconditionally secure bit commitment in OPTs and, anyway, the perfectly secure protocol should be the starting point for a rigorous analysis. The exceptional thing is that in our analysis we will be anyway able to build a cheating scheme also for some unconditionally secure bit commitment protocol, as we will see in Section~\ref{sec:cheating-BC}.\\

Finally, in the third part of the Chapter, we will study the proof of the impossibility of perfectly secure bit commitment in  OP quantum theory done in Ref.~\cite{chiribella}. We will review if all of the six axioms of quantum theory introduced in Section~\ref{sec:axioms} are really necessary conditions. Some of them will be neglected and other will be replaced with weaker hypothesis. With the new sufficient conditions that we will find, the proof of impossibility of perfectly secure bit commitment can be extended to other theories than the quantum one. Some applications hypothesis are the fermionic theory and the real quantum theory and in particular the PR-box theory. To the latter will be focused the next two Chapter, in fact, in literature numerous bit commitment protocols (and more generally quantum key distributions protocols) have been studied in the context of non-local correlated box, i.e. PR-boxes. So our analysis will provide a solid (operational probabilistic) point of view from which study all these protocols that have been prosed in the past years.

\section{From Coin Tossing to no-go Theorem}
\label{sec:fromcointossingtonogotheorem}
The bit commitment protocol was conceived for the first time by Blum in 1983 as a building block for secure coin tossing. To cite the abstract of the original work \cite{blum}:
\begin{center}
	\textit{Alice and Bob want to flip a coin by telephone. (They have just divorced, live in different cities, want to decide who gets the car.) Bob would not like to tell Alice HEADS and hear Alice (at the other end	of the line) say ``Here goes... I'm flipping the coin .... You lost!''}.
\end{center}

A standard example to illustrate bit commitment is for Alice to write the bit down on a piece of paper, which is then locked in a safe and sent to Bob, whereas Alice keeps the key. At a later time, she will unveil it by handing over the key to Bob. However, Bob has a well-equipped toolbox at home and may have been able to open the safe in the meantime. So while this scheme may offer reasonably good practical security, it is in principle insecure. Yet all bit commitment schemes that have wide currency today rely on such technological constraints: not on strongboxes and keys, but	on unproven assumptions that certain computations are hard to perform.\\

The first example of quantum bit commitment was first proposed by Bennet and Brassard in their famous paper of 1984 \cite{bb84} as a primitive for implementing coin tossing.\\
In their scheme, Alice commits to a bit value by preparing a sequence of photons in either of two mutually unbiased bases, in a way that the resulting quantum states are indistinguishable to Bob. The authors show that their protocol is secure against so-called passive cheating, in which Alice initially commits to the bit value $k$ and then tries to unveil $1-k$ later. However, they also prove that Alice can cheat with a more sophisticated strategy, in which she initially prepares pairs of maximally entangled states instead, keeps one particle of each pair in her laboratory and sends the second particle to Bob. For the first time, entanglement is recognized as a crucial factor in preventing perfectly secure bit commitment.\\
Subsequent proposals for bit commitment schemes tried to evade this type of attack by forcing the players to carry out measurements and communicate classically as they go through the protocol. At a 1993 conference Brassard \textit{et al.} presented a bit commitment protocol \cite{brassard91} that was claimed and generally accepted to be unconditionally secure.\\

In 1996 Lo and Chau \cite{lo-chau}, and Mayers \cite{mayers} realized that all previously proposed bit commitment protocols were vulnerable to a generalized version of the (EPR) attack that renders the BB84 proposal insecure, a result that they slightly extended to cover quantum bit commitment protocols in general. In OP terms, the determinant factor in the impossibility of bit commitment shifts from entanglement to the purification principle, this will be fully proved in Ref~\cite{chiribella}.\\
Their basic argument is the following. At the end of the commitment
phase, Bob will hold one out of two quantum states $\psi_k$ as proof
of Alice's commitment to the bit value $k\in\{0, 1\}$. Alice holds its
purification $\Psi_k$, which she will later pass on to Bob to
unveil. For the protocol to be concealing, the two states $k$ should
be (almost) indistinguishable, $\psi_0\approx\psi_1$. But Uhlmann's
theorem then implies the existence of a unitary transformation $\tU$
that (nearly) rotates the purification of $\psi_0$ into the
purification of $\psi_1$. Since $\tU$ is localized on the purifying system only, which is entirely under Alice's control, Lo-Chau-Mayers argue that Alice can switch at will between the two states, and is not in any way bound to her commitment. As a consequence, any concealing bit commitment protocol is argued to be necessarily non-binding.\\

Starting from 2000 the Lo-Chau-Mayers no-go theorem has been continually challenged, arguing that the impossibility proof does not exhaust all conceivable quantum bit commitment protocols. Several protocols have been proposed and claimed to circumvent the no-go theorem. These protocols seek to strengthen Bob's position with the help of "secret parameters" or "anonymous states", so that Alice lacks some information needed to cheat successfully: while Uhlmann's theorem would still imply the existence of a unitary cheating transformation as described above, this transformation might be unknown to Alice.\\

However, the above attempts to build up a secure quantum bit commitment protocol have motivated the thorough analysis of Ref.~\cite{d'ariano2007}, which provided a strengthened and explicit impossibility proof exhausting all conceivable protocols in which classical and quantum information is exchanged between two parties, including the possibility of protocol aborts and resets. This proof encompasses protocols even with unbounded number of communication rounds (it is only required that the expected number of rounds is finite), and with quantum systems on infinite-dimensional Hilbert spaces. However, the considerable length of this proof made it hard to follow, lacking a synthetic intuition of the impossibility proof. Finally in 2009 Chiribella \textit{et al.} in Ref.~\cite{chiribella-d'ariano2009} provided a new short impossibility proof of quantum bit commitment. In Ref.~\cite{chiribella} a similar demonstrative structure is used to prove the impossibility of perfectly secure bit commitment in an operational framework, not only in quantum theory, but in a wide range of theories with purification. This proof will be the one that we will study in the third Section of this Chapter.		

\section{A Formal Definition}				
\label{sec:defn-bit-commitment}
A rigorous definition is in order for a twofold reason. If it is true that it is necessary to define a framework where to operate, it is at the same time important to clearly remark the definition of bit commitment to which the statement \textit{"bit commitment is impossible"} refers to.\\

Despite we already referred to the word \textit{protocol}, we did not linger to define what we mean with it, therefore we will start by the notion of protocol and we will state the bit commitment in the OPT language and its key properties only thereafter.\\

\subsection{The protocol}
A protocol regulates the exchange of messages between participants, defining what are the honest strategies that they can adopt, so that at every stage it is clear what type of message is expected from the participants, although, of course, their content is not fixed. The expected message types can be either classical or quantum or a combination thereof.\\

In any bit commitment protocol, we can distinguish two main phases: the first is the \textit{commitment phase}, in which Alice and Bob exchange classical and quantum messages in order to commit the bit. The second is the \textit{opening phase} where Alice will send to Bob some classical or quantum information in order to to reveal the bit value.
\begin{description}
	\item[Commitment phase] this phase can end either with a successful commitment, or with an abort, in which the two parties irrevocably give up the purpose of committing the bit (of course, in a well designed protocol, if both parties are honest the probability of abort should be vanishingly small). If no abort took place, the bit value is considered to be committed to Bob but, supposedly, concealed from him. Since bit commitment is a two-party protocol and trusted third parties are not allowed, the starting state necessarily has to be originated by one of the two parties. Moreover, since we can always include in the protocol null steps (in which no information, classical or quantum, is exchanged), without loss of generality, we can restrict	our attention to protocols that are started by Alice.
	\item[Opening phase] in the case of abort during the commitment, this is just a null step, whereas, in the case of successful commitment, at the opening Alice will send to Bob some classical or quantum information in order to to reveal the bit value. Taking both Alice's message and his own (classical and quantum) records, Bob will then perform a suitable verification measurement. His measurement will result in either a successful readout of the committed bit, or in a failure, e.g. due to the detection of an attempted cheat. Again, in a well-designed protocol the probability of failure should be vanishingly small.
\end{description}

\subsection{OP bit commitment}
\label{sec:op-bit-commitment}
Alice wants to commit a classical bit $b\in\{0,1\}$ to Bob.  
\begin{enumerate}
	\item as we have seen before, the first phase is the commitment phase, in which Alice and Bob can perform any sequence of operations. Depending on whether Alice intended to commit $b=0$ or $b=1$, the state $\Psi_0,\,\Psi_1\in\st{AB}$ is selected, respectively.	We will assume that Alice and Bob's systems at the end of the commitment phase are A and B respectively, so that the two possible pure states that can be transmitted to Bob are $(e|_\rA\Psi_0,\,(e|_\rA\Psi_1\in\st{B}$;
	\item between the commitment and the opening phase Alice and Bob can perform only local operations on their systems A and B, respectively (together with every other ancillary system they control);
	\item during the opening phase, Alice transmits her system A to Bob who can perform any measurement on the joint system AB to know the bit $b$ and to check if it is compatible with the commitment of Alice. In general, to do this Bob can perform a two outcome measurement (Positive Operator Valued Measure, POVM) $\{a_0,a_1\}$ on the joint system AB.
\end{enumerate}
From now on, when we will refer to bit commitment we mean a protocol that is included in the previous scheme.\\

Already at the beginning of this Chapter we intuitively mention the two way of cheating that can occur. Being a two party protocol, we can have Alice's cheating, i.e. she changes the bit after the commitment (the protocol is not binding), and Bob's cheating, i.e. Bob discovers the bit committed before the opening (the protocol is not concealing). In this way we have now identified two key properties of any bit commitment protocol. However there is a third key property that is often omitted in the literature: the \textit{correctness} of the protocol that guarantees the correct verification of the bit committed in the opening phase.\\
To recapitulate with a properer language, we will say that a bit commitment protocol is
\begin{itemize}
	\item\textbf{binding:} if, for honest Bob, Alice should not be able to change the bit she committed. More precisely, assume that a possibly dishonest Alice committed $b$ but she wants to reveal $b^\prime\ne b$, then it must be
	\begin{equation}
		\label{eq:binding}
		\textsf{Pr}\left[\text{ Bob accepts }|\text{ Alice reveals }b^\prime\,\right]<1\,;
	\end{equation}
	\item\textbf{concealing:} if, for honest Alice, Bob should not be able to know the bit that Alice committed until she reveals it;
	\item\textbf{correct:} for honest Alice and Bob, if Alice commits $b$ and later reveals $b$ to Bob, then Bob accepts with probability grater then $1/2$:
	\begin{equation}
		\label{eq:correctness}
		\textsf{Pr}\left[\text{ Bob accepts }|\text{ Alice reveals }b\,\right]>\frac{1}{2}\,.
	\end{equation}		
\end{itemize}
Furthermore, we will say that a bit commitment protocol is \textit{perfect} if
\begin{itemize}
	\item it is perfectly binding, namely for honest Bob, Alice cannot switch between $\Psi_0$ and $\Psi_1$ in such a way that Bob cannot detect the switch with certainty (i.e. Eq.~\eqref{eq:binding} becomes $\textsf{Pr}\left[\text{ Bob accepts }|\text{ Alice reveals }b^\prime\,\right]=0$). Namely it does not exists a reversible channel $\tU$ such that
	\begin{equation}
		\label{eq:perfectly-binding}
		(\tU\otimes\tI_\rB)|\Psi_0)_{\rA\rB}=|\Psi_1)_{\rA\rB}\,;
	\end{equation}
	\item it is perfectly concealing, namely for honest Alice, Bob would not be able to perform some measurement of his system B and gain at least partial information about which bit Alice committed
	\begin{equation}
		\label{eq:perfectly-concealing}
		(e|_\rA|\Psi_0)_{\rA\rB}=(e|_\rA|\Psi_1)_{\rA\rB}\,;
	\end{equation}
	\item it is correct with probability one, namely Bob accepts with probability one. So Eq.~\eqref{eq:correctness} becomes
	\begin{equation*}
		\textsf{Pr}\left[\text{ Bob accepts }|\text{ Alice reveals }b\,\right]=1\,.
	\end{equation*}
\end{itemize}
In the perfect implementation of the protocol the probability of cheating of Bob must be $1/2$, since in the worst case he can always make a random guess of the committed bit. The cheating probability of Alice is instead equal to zero.\\

As we have anticipated in Section \ref{sec:fromcointossingtonogotheorem}, perfectly secure bit commitment protocol is impossible in quantum and classic theory. However, if the perfect concealing and perfect biding conditions are relaxed as follows
\begin{equation*}
	\begin{aligned}
		(\tU\otimes\tI_\rB)|\Psi_0)_{\rA\rB}&\approx|\Psi_1)_{\rA\rB}\\
		(e|_\rA|\Psi_0)_{\rA\rB}\approx(&e|_\rA|\Psi_1)_{\rA\rB}
	\end{aligned}
\end{equation*}
interesting level of security relevant for concrete applications emerges. Within this scenario the main effort is in quantifying the cheating probabilities and their trade-off in operational terms. The goal is to achieve a protocol that is asymptotically binding and concealing (both the Alice and Bob cheating probability can be made arbitrarily small) against an adversary that has no restrictions on the computational resources. In this case one has unconditionally secure bit commitment.

\section{Lightening No Bit Commitment}
\label{sec:lightening-noBC}
In this Section we will analyse the necessary assumptions to prove the impossibility of perfectly secure bit commitment. In the literature, a no-go theorem in quantum theory has been proven in Ref.~\cite{d'ariano2007,chiribella-d'ariano2009} but we found that not all of the axioms of quantum theory are necessary. So, using weaker assumption, we will generalize the impossibility of bit commitment to a larger set of operational probabilistic theories.\\

Clearly the causality Axiom~\ref{axiom:causality} is essential, otherwise the very protocol could not be defined properly (there would not be a given order in the succeeding of the phases). In our analysis we also assume Axiom~\ref{axiom:atomicity}, atomicity of composition, that is a sufficient condition to grant Corollary~\ref{cor:parallelcompositionofpurestates}.\\
Finally, instead of the purification Axiom~\ref{axiom:purification}, we take the following weaker assumption. Before stating it, we introduce the notion of dynamically faithful state.
\begin{defn}[Dynamically faithful state]
	A state $\sigma\in\st{AC}$ is dynamically faithful for system $\rA$ if for any couple of transformations $\tA,\,\tA^\prime\in\transf{A}{B}$ one has
	\begin{equation*}
		\tA|\sigma)_{\rA\rC}=\tA^\prime|\sigma)_{\rA\rC}\Longrightarrow\tA=\tA^\prime\, .
	\end{equation*}
\end{defn}
We are now in position to state the "new" axiom.
\begin{axiom}
	\label{axiom:7}
	For every system $\rA$ there exist a system $\tilde{\rA}$ and a pure state $\Psi^{(\rA)}\in\mathsf{St}(\rA\tilde{\rA})$ that is dynamically faithful for system $\rA$. Furthermore, for every system $\rB$ and for every bipartite state $R\in\mathsf{St}(\rB\tilde{\rA})$ such that 
	\begin{equation*}
		\begin{aligned}
			\Qcircuit @C=1.2em @R=.7em @! R { 
				\multiprepareC{1}{R}&
				\poloFantasmaCn{\rB}\qw&
				\measureD{e}\\
				\pureghost{R}&
				\poloFantasmaCn{\tilde{\rA}}\qw&
				\qw\\
			}
		\end{aligned}
		\, = \,
		\begin{aligned}
			\Qcircuit @C=1.2em @R=.7em @! R { 
				\multiprepareC{1}{\Psi^{(\rA)}}&
				\poloFantasmaCn{\rA}\qw&
				\measureD{e}\\
				\pureghost{\Psi^{(\rA)}}&
				\poloFantasmaCn{\tilde{\rA}}\qw&
				\qw\\
			}
		\end{aligned}\,,
	\end{equation*}
	there exists a purification of $R$. If $\Phi$, $\Phi^\prime\in\mathsf{St}(\rC\rB\tilde{\rA})$ are two purifications of $R$ then they are connected by a reversible transformation $\tU\in\mathsf{Transf}(\rC)$.
\end{axiom}
In quantum theory, the existence of mixed faithful states is a direct consequence of local discriminability. In a theory with both local discriminability and purification there exist also dynamically faithful states that are pure. In Axiom \ref{axiom:7} we do not require local discriminability nor purification but the existence of dynamically faithful pure states. In addition we also require that a purification exists only for those states that have the same marginal of the dynamically faithful pure ones and that this purification is unique up to a reversible transformation on the purifying system.\\

An immediate consequence of Axioms \ref{axiom:atomicity} and \ref{axiom:7} is the following lemma on the uniqueness of purification.
\begin{lemma}[Uniqueness of the purification up to channels on the purifying systems]
	\label{lemma:uniqPuri}
	Let $\Psi\in\st{AB}$ and $\Psi^\prime\in\st{AC}$ be two purification of $\rho\in\st{A}$. Then there exists a channel $\tC\in\transf{B}{C}$ such that
	\begin{equation*}
		\begin{aligned}
			\Qcircuit @C=1.2em @R=.7em @! R { 
				\multiprepareC{1}{\Psi^\prime}&
				\poloFantasmaCn{\rA}\qw&
				\qw\\
				\pureghost{\Psi^\prime}&
				\poloFantasmaCn{\rC}\qw&
				\qw\\
			}
		\end{aligned}
		\, = \,
		\begin{aligned}
			\Qcircuit @C=1.2em @R=.7em @! R { 
				\multiprepareC{1}{\Psi}&
				\qw&
				\poloFantasmaCn{\rA}\qw&
				\qw&\qw\\
				\pureghost{\Psi}&
				\poloFantasmaCn{\rB}\qw&
				\gate{\tC}&
				\poloFantasmaCn{C}\qw&
				\qw\\
			}
		\end{aligned}\,.
	\end{equation*}
	Moreover, channel $\tC$ has the form
	\begin{equation*}
		\begin{aligned}
			\Qcircuit @C=1.2em @R=.7em @! R { 
				&
				\poloFantasmaCn{\rB}\qw&
				\gate{\tC}&
				\poloFantasmaCn{\rC}\qw&
				\qw
			}
		\end{aligned}
		\, = \,
		\begin{aligned}
			\Qcircuit @C=1.2em @R=.7em @! R { 
				\prepareC{\varphi_0}&
				\poloFantasmaCn{\rC}\qw&
				\multigate{1}{\tU}&
				\poloFantasmaCn{\rB}\qw&
				\measureD{e}\\
				&
				\poloFantasmaCn{\rB}\qw&
				\ghost{\tU}&
				\poloFantasmaCn{C}\qw&
				\qw\\
			}
		\end{aligned}\,.
	\end{equation*}
	for some pure state $\varphi_0\in\st{C}$ and some reversible channel $\tU\in\transf{B}{C}$.
\end{lemma}
\begin{proof}
	Let $\eta$ and $\varphi_0$ be an arbitrary pure state of B and C, respectively. Then, due to Axiom \ref{axiom:atomicity}, $|\Psi^\prime)_{\rA\rC}|\eta)_{\rB}$ and $|\Psi)_{\rA\rB}|\varphi_0)_{\rC}$ are still two pure states and so they are both purifications of $\rho$ with the same purifying system $\rB\rC$. Due to Axiom \ref{axiom:7}, we have
	\begin{equation*}
		\begin{aligned}
			\Qcircuit @C=1.2em @R=.7em @! R { 
				\multiprepareC{1}{\Psi^\prime}&
				\poloFantasmaCn{\rA}\qw&
				\qw\\
				\pureghost{\Psi^\prime}&
				\poloFantasmaCn{\rC}\qw&
				\qw\\
				\prepareC{\eta}&
				\poloFantasmaCn{\rB}\qw&
				\qw
			}
		\end{aligned}
		\, = \,
		\begin{aligned}
			\Qcircuit @C=1.2em @R=.7em @! R { 
				\multiprepareC{1}{\Psi}&
				\qw&
				\poloFantasmaCn{\rA}\qw&
				\qw&\qw\\
				\pureghost{\Psi}&
				\poloFantasmaCn{\rB}\qw&
				\multigate{1}{\tU}&
				\poloFantasmaCn{\rC}\qw&
				\qw\\
				\prepareC{\varphi_0}&
				\poloFantasmaCn{\rC}\qw&
				\ghost{\tU}&
				\poloFantasmaCn{B}\qw&
				\qw\\
			}
		\end{aligned}\,.
	\end{equation*}	
	Applying the deterministic effect $e_\rB$ on system $\rB$ we obtain the thesis, with $\tC=(e|_\rB\;\tU|\varphi_0)$.
\end{proof}

\subsection{Reversible dilation of channels}
Before starting the proof of the impossibility of bit commitment it is in order to derive some useful results about reversible dilations of channels.
\begin{lemma}
	\label{lemma:lemma23}
	Let $R\in\mathsf{St}(\rB\tilde{\rA})$ be a state such that
	\begin{equation*}
		\begin{aligned}
			\Qcircuit @C=1.2em @R=.7em @! R { 
				\multiprepareC{1}{R}&
				\poloFantasmaCn{\rB}\qw&
				\measureD{e}\\
				\pureghost{R}&
				\poloFantasmaCn{\tilde{\rA}}\qw&
				\qw
			}
		\end{aligned}
		\, = \,
		\begin{aligned}
			\Qcircuit @C=1.2em @R=.7em @! R { 
				\multiprepareC{1}{\Psi^{(\rA)}}&
				\poloFantasmaCn{\rA}\qw&
				\measureD{e}\\
				\pureghost{\Psi^{(\rA)}}&
				\poloFantasmaCn{\tilde{\rA}}\qw&
				\qw\\
			}
		\end{aligned}\,.
	\end{equation*}	
	where $\Psi^{(\rA)}\in\mathsf{St}(\rA\tilde{\rA})$ is a pure dynamically faithful state for system $\rA$. Then there exist a system $\rC$, a pure state $\varphi_0\in\mathsf{St}(\rB\rC)$, and a reversible channel $\tU\in\mathsf{Transf}(\rA\rB\rC)$ such that 
	\begin{equation}
		\label{eq:revchannel}
		\begin{aligned}
			\Qcircuit @C=1.2em @R=.7em @! R { 
				\multiprepareC{1}{R}&
				\poloFantasmaCn{\rB}\qw&
				\qw\\
				\pureghost{R}&
				\poloFantasmaCn{\tilde{\rA}}\qw&
				\qw
			}
		\end{aligned}
		\, = \,
		\begin{aligned}
			\Qcircuit @C=1.2em @R=.7em @! R { 
				\prepareC{\varphi_0}&
				\poloFantasmaCn{\rB\rC}\qw&
				\multigate{1}{\tU}&
				\poloFantasmaCn{\rA\rC}\qw&
				\measureD{e}\\
				\multiprepareC{1}{\Psi^{(\rA)}}&
				\poloFantasmaCn{\rA}\qw&
				\ghost{\tU}&
				\poloFantasmaCn{B}\qw&
				\qw\\
				\pureghost{\Psi^{(\rA)}}&
				\qw&
				\poloFantasmaCn{\tilde{\rA}}\qw&
				\qw&\qw\\
			}
		\end{aligned}\,.
	\end{equation}	
	Moreover the channel $\tV\in\mathsf{Transf}(\rA\rightarrow\rA\rB\rC)$ defined by $\tV:=\tU|\varphi_0)_{\rB\rC}$ is unique up to reversible channels on $\rA\rC$.
\end{lemma}
\begin{proof}
	Take a purification of $R$, say $\Psi_R\in\mathsf{St}(\rC\rB\tilde{\rA})$ for some purifying system $\rC$ (existence of such a purification is granted by Axiom \ref{axiom:7}). One has
	\begin{equation*}
		\begin{aligned}
			\Qcircuit @C=1.2em @R=.7em @! R { 
				\multiprepareC{2}{\Psi_R}&
				\poloFantasmaCn{\rC}\qw&
				\measureD{e}\\
				\pureghost{\Psi_R}&
				\poloFantasmaCn{\rB}\qw&
				\measureD{e}\\
				\pureghost{\Psi_R}&
				\poloFantasmaCn{\tilde{\rA}}\qw&
				\qw
			}
		\end{aligned}
		\, = \,
		\begin{aligned}
			\Qcircuit @C=1.2em @R=.7em @! R { 
				\multiprepareC{1}{R}&
				\poloFantasmaCn{\rB}\qw&
				\measureD{e}\\
				\pureghost{R}&
				\poloFantasmaCn{\tilde{\rA}}\qw&
				\qw\\
			}
		\end{aligned}
		\, = \,
		\begin{aligned}
			\Qcircuit @C=1.2em @R=.7em @! R { 
				\multiprepareC{1}{\Psi^{(\rA)}}&
				\poloFantasmaCn{\rA}\qw&
				\measureD{e}\\
				\pureghost{\Psi^{(\rA)}}&
				\poloFantasmaCn{\tilde{\rA}}\qw&
				\qw\\
			}
		\end{aligned}
	\end{equation*}	
	that is, the pure state $\Psi_R$ and $\Psi^{(\rA)}$ have the same marginal on system $\tilde{\rA}$. Applying the uniqueness of purification expressed by Lemma \ref{lemma:uniqPuri} one then obtains
	\begin{equation*}
		\begin{aligned}
			\Qcircuit @C=1.2em @R=.7em @! R { 
				\multiprepareC{2}{\Psi_R}&
				\poloFantasmaCn{\rC}\qw&
				\qw\\
				\pureghost{\Psi_R}&
				\poloFantasmaCn{\rB}\qw&
				\qw\\
				\pureghost{\Psi_R}&
				\poloFantasmaCn{\tilde{\rA}}\qw&
				\qw
			}
		\end{aligned}
		\, = \,
		\begin{aligned}
			\Qcircuit @C=1.2em @R=.7em @! R { 
				\prepareC{\varphi_0}&
				\poloFantasmaCn{\rB\rC}\qw&
				\multigate{1}{\tU}&
				\poloFantasmaCn{\rA}\qw&
				\measureD{e}\\
				\multiprepareC{1}{\Psi^{(\rA)}}&
				\poloFantasmaCn{\rA}\qw&
				\ghost{\tU}&
				\poloFantasmaCn{\rB\rC}\qw&
				\qw\\
				\pureghost{\Psi^{(\rA)}}&
				\qw&
				\poloFantasmaCn{\tilde{\rA}}\qw&
				\qw&\qw
			}
		\end{aligned}\, .
	\end{equation*}
	Applying the deterministic effect on system $\rC$ on both sides, one then proves Eq.~\eqref{eq:revchannel}. Moreover, if $\tV^\prime\coloneqq\tU^\prime|\varphi^\prime_0)_{\rB\rC}$ is channel such that Eq.~\eqref{eq:revchannel} holds, then the pure states $\tV|\Psi^{(\rA)})_{\rA\tilde{\rA}}$ and $\tV^\prime|\Psi^{(\rA)})_{\rA\tilde{\rA}}$ have the same marginal on system $\rB\tilde{\rA}$. Uniqueness of purification then implies
	\begin{equation*}
		\begin{aligned}
			\Qcircuit @C=1.2em @R=.7em @! R { 
				& &
				\amultigate{1}{\tV^\prime}&
				\poloFantasmaCn{\rA\rC}\qw&
				\qw\\
				\multiprepareC{1}{\Psi^{(\rA)}}&
				\poloFantasmaCn{\rA}\qw&
				\ghost{\tV^\prime}&
				\poloFantasmaCn{\rB}\qw&
				\qw\\
				\pureghost{\Psi^{(\rA)}}&
				\qw&
				\poloFantasmaCn{\tilde{\rA}}\qw&
				\qw&\qw
			}
		\end{aligned}
		\, = \,
		\begin{aligned}
			\Qcircuit @C=1.2em @R=.7em @! R { 
				& & 
				\amultigate{1}{\tV}&
				\poloFantasmaCn{\rA\rC}\qw&
				\gate{\tW}&
				\poloFantasmaCn{\rA\rC}\qw&
				\qw\\
				\multiprepareC{1}{\Psi^{(\rA)}}&
				\poloFantasmaCn{\rA}\qw&
				\ghost{\tV}&
				\poloFantasmaCn{\rB}\qw&
				\qw&\qw&\qw\\
				\pureghost{\Psi^{(\rA)}}&
				\qw&
				\poloFantasmaCn{\tilde{\rA}}\qw&
				\qw&\qw&\qw&\qw
			}
		\end{aligned}
	\end{equation*}
	for some reversible channel $\tW\in\mathsf{Transf}(\rA\rC)$. Since $\Psi^{(\rA)}$ is dynamically faithful for $\rA$, this implies $\tV^\prime=\tW\tV$.
\end{proof}
We now give the definition of dilatation and reversible dilatation.
\begin{defn}[Dilatation of a channel]
	A dilatation of channel $\tC\in\transf{A}{B}$ is a channel $\tV\in\transf{A}{BE}$ such that 
	\begin{equation*}
		\begin{aligned}
			\Qcircuit @C=1.2em @R=.7em @! R {
				&
				\poloFantasmaCn{\rA}\qw&
				\gate{\tC}&
				\poloFantasmaCn{\rB}\qw&
				\qw
			}
		\end{aligned}\, = \,
		\begin{aligned}
			\Qcircuit @C=1.2em @R=.7em @! R { 
				& &
				\amultigate{1}{\tV}&
				\poloFantasmaCn{\rE}\qw&
				\measureD{e}\\
				&
				\poloFantasmaCn{\rA}\qw&
				\ghost{\tV}&
				\poloFantasmaCn{\rB}\qw&
				\qw\\
			}
		\end{aligned}\, .
	\end{equation*}
	We refer to system $\rE$ as to the \textnormal{environment}.
\end{defn}
\begin{defn}[Reversible dilation]
	A dilatation $\tV\in\transf{A}{BE}$ is called reversible if there exists a system $\rE_0$ such that $\rA\rE_0\simeq\rB\rE$ and
	\begin{equation*}
		\begin{aligned}
			\Qcircuit @C=1.2em @R=.7em @! R { 
				& &
				\amultigate{1}{\tV}&
				\poloFantasmaCn{\rE}\qw&
				\qw\\
				&
				\poloFantasmaCn{\rA}\qw&
				\ghost{\tV}&
				\poloFantasmaCn{\rB}\qw&
				\qw\\
			}
		\end{aligned}
		\, = \,
		\begin{aligned}
			\Qcircuit @C=1.2em @R=.7em @! R { 
				\prepareC{\varphi_0}&
				\poloFantasmaCn{\rE_0}\qw&
				\multigate{1}{\tU}&
				\poloFantasmaCn{\rE}\qw&
				\qw\\
				&
				\poloFantasmaCn{\rA}\qw&
				\ghost{\tU}&
				\poloFantasmaCn{\rB}\qw&
				\qw\\
			}
		\end{aligned}
	\end{equation*}
	for some pure state $\varphi_0\in\mathsf{St}(\rE_0)$ and some reversible channel $\tU\in\mathsf{Transf}(\rA\rE_0\rightarrow\rB\rE)$.
\end{defn}
According to the above definitions, we have the following dilatation theorem:
\begin{theorem}[Reversible dilatation of channels]
	\label{thm:RevDila}
	Every channel $\tC\in\transf{A}{B}$ has a reversible dilatation $\tV\in\transf{A}{BE}$. If $\tV$, $\tV^\prime\in\transf{A}{BE}$ are two reversible dilatations of the same channel, then they are connected by a reversible transformation on the environment, namely
	\begin{equation*}
		\begin{aligned}
			\Qcircuit @C=1.2em @R=.7em @! R { 
				& &
				\amultigate{1}{\tV^\prime}&
				\poloFantasmaCn{\rE}\qw&
				\measureD{e}\\
				&
				\poloFantasmaCn{\rA}\qw&
				\ghost{\tV^\prime}&
				\poloFantasmaCn{\rB}\qw&
				\qw\\
			}
		\end{aligned}
		\, = \,
		\begin{aligned}
			\Qcircuit @C=1.2em @R=.7em @! R { 
				& &
				\amultigate{1}{\tV}&
				\poloFantasmaCn{\rE}\qw&
				\measureD{e}\\
				&
				\poloFantasmaCn{\rA}\qw&
				\ghost{\tV}&
				\poloFantasmaCn{\rB}\qw&
				\qw\\
			}
		\end{aligned}
	\end{equation*}
	\begin{equation*}
		\Longrightarrow
		\begin{aligned}
			\Qcircuit @C=1.2em @R=.7em @! R { 
				& &
				\amultigate{1}{\tV^\prime}&
				\poloFantasmaCn{\rE}\qw&
				\qw\\
				&
				\poloFantasmaCn{\rA}\qw&
				\ghost{\tV^\prime}&
				\poloFantasmaCn{\rB}\qw&
				\qw\\
			}
		\end{aligned}
		\, = \,
		\begin{aligned}
			\Qcircuit @C=1.2em @R=.7em @! R { 
				& &
				\amultigate{1}{\tV}&
				\poloFantasmaCn{\rE}\qw&
				\gate{\tW}&
				\poloFantasmaCn{\rE}\qw&
				\qw\\
				&
				\poloFantasmaCn{\rA}\qw&
				\ghost{\tV}&
				\qw&
				\poloFantasmaCn{\rB}\qw&
				\qw&\qw\\
			}
		\end{aligned}
	\end{equation*}
	for some reversible channel $\tW\in\mathsf{Transf}(\rE)$.
\end{theorem}
\begin{proof}
	Let us store the channel $\tC$ in the faithful state $\Psi^{(\rA)}\in\mathsf{St}(\rA\tilde{\rA})$, thus getting the state $R_\tC$:
	\begin{equation*}
		\begin{aligned}
			\Qcircuit @C=1.2em @R=.7em @! R { 
				\multiprepareC{1}{R_\tC}&
				\poloFantasmaCn{\rB}\qw&
				\qw\\
				\pureghost{R_\tC}&
				\poloFantasmaCn{\tilde{\rA}}\qw&
				\qw\\
			}
		\end{aligned}
		\, \coloneqq \,
		\begin{aligned}
			\Qcircuit @C=1.2em @R=.7em @! R { 
				\multiprepareC{1}{\Psi}&
				\poloFantasmaCn{\rA}\qw&
				\gate{\tC}&
				\poloFantasmaCn{\rB}\qw&
				\qw\\
				\pureghost{\Psi}&
				\qw&
				\poloFantasmaCn{\tilde{\rA}}\qw&
				\qw&\qw\\
			}
		\end{aligned}\, .
	\end{equation*}
	Since $\tC$ is a channel, it satisfy the normalization condition
	\begin{equation*}
		\begin{aligned}
			\Qcircuit @C=1.2em @R=.7em @! R { 
				&
				\poloFantasmaCn{\rA}\qw&
				\gate{\tC}&
				\poloFantasmaCn{\rB}\qw&
				\measureD{e}
			}
		\end{aligned}
		\, = \,
		\begin{aligned}
			\Qcircuit @C=1.2em @R=.7em @! R { 
				&
				\poloFantasmaCn{\rA}\qw&
				\measureD{e}
			}
		\end{aligned}\, ,
	\end{equation*}
	which implies
	\begin{equation*}
		\begin{aligned}
			\Qcircuit @C=1.2em @R=.7em @! R {
				\multiprepareC{1}{R_\tC}&
				\poloFantasmaCn{\rB}\qw&
				\measureD{e}\\
				\pureghost{R_\tC}&
				\poloFantasmaCn{\tilde{\rA}}\qw&
				\qw\\					
			}
		\end{aligned}
		\, = \,
		\begin{aligned}
			\Qcircuit @C=1.2em @R=.7em @! R {
				\multiprepareC{1}{\Psi^{(\rA)}}&
				\poloFantasmaCn{\rA}\qw&
				\gate{\tC}&
				\poloFantasmaCn{\rB}\qw&
				\measureD{e}\\
				\pureghost{\Psi^{(\rA)}}&
				\qw&
				\poloFantasmaCn{\tilde{\rA}}\qw&
				\qw&\qw\\					
			}
		\end{aligned}
		\, = \,
		\begin{aligned}
			\Qcircuit @C=1.2em @R=.7em @! R {
				\multiprepareC{1}{\Psi^{(\rA)}}&
				\poloFantasmaCn{\rA}\qw&
				\measureD{e}\\
				\pureghost{\Psi^{(\rA)}}&
				\poloFantasmaCn{\tilde{\rA}}\qw&
				\qw\\					
			}
		\end{aligned}\, .
	\end{equation*}				
	Now, applying Lemma \ref{lemma:lemma23} we obtain 
	\begin{equation*}
		\begin{aligned}
			\Qcircuit @C=1.2em @R=.7em @! R {
				\multiprepareC{1}{R_\tC}&
				\poloFantasmaCn{\rB}\qw&
				\qw\\
				\pureghost{R_\tC}&
				\poloFantasmaCn{\tilde{\rA}}\qw&
				\qw\\					
			}
		\end{aligned}
		\, = \,
		\begin{aligned}
			\Qcircuit @C=1.2em @R=.7em @! R {
				\prepareC{\varphi_0}&
				\poloFantasmaCn{\rB\rC}\qw&
				\multigate{1}{\tU}&
				\poloFantasmaCn{\rA\rC}\qw&
				\measureD{e}\\
				\multiprepareC{1}{\Psi^{(\rA)}}&
				\poloFantasmaCn{\rA}\qw&
				\ghost{\tU}&
				\poloFantasmaCn{\rB}\qw&
				\qw\\
				\pureghost{\Psi^{(\rA)}}&
				\qw&
				\poloFantasmaCn{\tilde{\rA}}\qw&
				\qw&\qw\\					
			}
		\end{aligned}\, .
	\end{equation*}	
	Since $\Psi^{(\rA)}$ is dynamically faithful for system $\rA$, this implies 
	\begin{equation*}
		\begin{aligned}
			\Qcircuit @C=1.2em @R=.7em @! R {
				&
				\poloFantasmaCn{\rA}\qw&
				\gate{\tC}&
				\poloFantasmaCn{\rB}\qw&
				\qw\\					
			}
		\end{aligned}
		\, = \,
		\begin{aligned}
			\Qcircuit @C=1.2em @R=.7em @! R {
				\prepareC{\varphi_0}&
				\poloFantasmaCn{\rB\rC}\qw&
				\multigate{1}{\tU}&
				\poloFantasmaCn{\rA\rC}\qw&
				\measureD{e}\\
				&
				\poloFantasmaCn{\rA}\qw&
				\ghost{\tU}&
				\poloFantasmaCn{\rB}\qw&
				\qw\\				
			}
		\end{aligned}\, .
	\end{equation*}	
	Therefore, $\tV \coloneqq \tU|\varphi_0)_{\rB\rC}$ is a reversible dilatation of $\tC$, with $\rE_0\coloneqq\rB\rC$ and $\rE\coloneqq\rA\rC$. Finally, the uniqueness clause in Lemma \ref{lemma:lemma23} implies uniqueness of the dilatation.
\end{proof}
Moreover, two reversible dilatations of the same channel with different environments are related as follows.
\begin{lemma}
	\label{lemma:corollary}
	Let $\tV\in\transf{A}{BE}$ and $\tV^\prime\in\mathsf{Transf}(\rA\rightarrow\rB\rE^\prime)$ be two reversible dilatations of the same channel $\tC$, with generally different environment $\tE$ and $\rE^\prime$. Then there is a channel $\tL$ from $\rE$ to $\rE\rE^\prime$ such that
	\begin{equation*}
		\begin{aligned}
			\Qcircuit @C=1.2em @R=.7em @! R {
				& &
				\amultigate{1}{\tV^\prime}&
				\poloFantasmaCn{\rE^\prime}\qw&
				\qw\\
				&
				\poloFantasmaCn{\rA}\qw&
				\ghost{\tV^\prime}&
				\poloFantasmaCn{\rB}\qw&
				\qw\\								
			}
		\end{aligned}
		\, = \,
		\begin{aligned}
			\Qcircuit @C=1.2em @R=.7em @! R {
				& & & &
				\amultigate{1}{\tL}&
				\poloFantasmaCn{\rE}\qw&
				\measureD{e}\\
				& &
				\amultigate{1}{\tV}&
				\poloFantasmaCn{\rE}\qw&
				\ghost{\tL}&
				\poloFantasmaCn{\rE^\prime}\qw&
				\qw\\
				&
				\poloFantasmaCn{\rA}\qw&
				\ghost{\tV}&
				\qw&
				\poloFantasmaCn{\rB}\qw&
				\qw&\qw\\
			}
		\end{aligned}\, .
	\end{equation*}	
	The channel $\tL$ has the form
	\begin{equation*}
		\begin{aligned}
			\Qcircuit @C=1.2em @R=.7em @! R {
				& &
				\amultigate{1}{\tL}&						
				\poloFantasmaCn{\rE}\qw&
				\qw\\
				&
				\poloFantasmaCn{\rE}\qw&
				\ghost{\tL}&
				\poloFantasmaCn{\rE^\prime}\qw&
				\qw\\
			}
		\end{aligned}
		\, = \,
		\begin{aligned}
			\Qcircuit @C=1.2em @R=.7em @! R {
				\prepareC{\eta_0}&
				\poloFantasmaCn{\rE^\prime}\qw&
				\multigate{1}{\tU}&
				\poloFantasmaCn{\rE}\qw&
				\qw\\
				&
				\poloFantasmaCn{\rE}\qw&
				\ghost{\tU}&
				\poloFantasmaCn{\rE^\prime}\qw&
				\qw\\
			}
		\end{aligned}
	\end{equation*}	
	for some pure state $\eta_0$ and some reversible transformation $\tU\in\mathsf{Transf}(\rE\rE^\prime)$.
\end{lemma}
\begin{proof}
	Apply $\tV$ and $\tV^\prime$ to the faithful state $\Psi^{(\rA)}$ and then use the uniqueness of purification stated in Lemma \ref{lemma:uniqPuri}.
\end{proof}

\subsection{Casually ordered channels and channels with memory}
The last step before the proof of the theorem is the notion of casually ordered channels and channels with memory.
\begin{defn}[Casually ordered bipartite channel]
	A bipartite channel $\tC$ from $\rA_1\rA_2$ to $\rB_1\rB_2$ is casually ordered if there is a channel $\tD$ from $\rA_1$ to $\rB_1$ such that $(e|_{\rB_2}\tC=\tD\otimes(e|_{\rA_2}$. Diagrammatically,
	\begin{equation}
		\label{eq:casuallyorderedbipartitechannels}
		\begin{aligned}
			\Qcircuit @C=1.2em @R=.7em @! R {
				&
				\poloFantasmaCn{\rA_1}\qw&
				\multigate{1}{\tC}&
				\poloFantasmaCn{\rB_1}\qw&
				\qw\\
				&
				\poloFantasmaCn{\rA_2}\qw&
				\ghost{\tC}&
				\poloFantasmaCn{\rB_2}\qw&
				\measureD{e}
			}
		\end{aligned}
		\, = \,
		\begin{aligned}
			\Qcircuit @C=1.2em @R=.7em @! R {
				&
				\poloFantasmaCn{\rA_1}\qw&
				\gate{\tD}&
				\poloFantasmaCn{\rB_1}\qw&
				\qw\\
				&
				\qw&
				\poloFantasmaCn{\rA_2}\qw&
				\qw&
				\measureD{e}
			}	
		\end{aligned}\,.
	\end{equation}
\end{defn}
Eq.~\eqref{eq:casuallyorderedbipartitechannels} means that the channel $\tC$ does not allow for signaling from the input system $\rA_2$ to the output system $\rB_1$. In a relativistic context, this can be interpreted as $\rB_1$ being outside the casual future of $\rA_2$.
\begin{defn}[Sequence of two channel with memory]
	A bipartite channel $\tC$ from $\rA_1\rA_2$ to $\rB_1\rB_2$ can be realized as a \textnormal{sequence of two channels with memory} if there exist two system $\rE_1,\,\rE_2$, called \textnormal{memory systems}, and two channels $\tC_1\in\mathsf{Transf}(\rA_1\rightarrow\rB_1\rE_1)$ and $\tC_2\in\mathsf{Transf}(\rA_2\rE_1\rightarrow\rB_2\rE_2)$ such that $\tC=(e|_{\rE_2}(\tC_2\otimes\tI_{\rB_1})(\tI_{\rA_2}\otimes\tC_1)$. Diagrammatically,
	\begin{equation}
		\label{eq:channelwithmemory}
		\begin{aligned}
			\Qcircuit @C=1.2em @R=.7em @! R {
				&
				\poloFantasmaCn{\rA_1}\qw&
				\multigate{1}{\tC}&
				\poloFantasmaCn{\rB_1}\qw&
				\qw\\
				&
				\poloFantasmaCn{\rA_2}\qw&
				\ghost{\tC}&
				\poloFantasmaCn{\rB_2}\qw&
				\qw\\
			}
		\end{aligned}
		\, = \,
		\begin{aligned}
			\Qcircuit @C=1.2em @R=.7em @! R {
				&
				\poloFantasmaCn{\rA_1}\qw&
				\multigate{1}{\tC_1}&
				\poloFantasmaCn{\rB_1}\qw&
				\qw
				& &
				\poloFantasmaCn{\rA_2}\qw&
				\multigate{1}{\tC_2}&
				\poloFantasmaCn{\rB_2}\qw&
				\qw\\
				& & 
				\aghost{\tC_1}&
				\qw&
				\poloFantasmaCn{\rE_1}\qw&
				\qw&\qw&
				\ghost{\tC_2}&
				\poloFantasmaCn{\rE_2}\qw&
				\measureD{e}\\
			}
		\end{aligned}\,.
	\end{equation}
\end{defn}
For causally ordered bipartite channels the dilatation theorem implies the following result:
\begin{theorem}[Casual ordering is memory]
	A bipartite channel $\tC$ from $\rA_1\rA_2$ to $\rB_1\rB_2$ is casually ordered if and only if it can be realized as a sequence of two channels with memory. Moreover, the channels $\tC_1,\,\tC_2$ in Eq.~\eqref{eq:channelwithmemory} can be always chosen such that $\tC_2\tC_1$ is a reversible dilatation of $\tC$.
\end{theorem}
\begin{proof}
	If Eq.~\eqref{eq:channelwithmemory} holds, the channel $\tC$ is clearly casually ordered, with the channel $\tD$ given by $\tD\coloneqq(e|_{\rE_1}\tC_1$. Conversely, suppose that $\tC$ is casually ordered. Take a reversible dilatation of $\tC$, say $\tV\in\mathsf{Transf}(\rA_1\rA_2\rightarrow\rB_1\rB_2\rE)$, and a reversible dilatation of $\tD$, say $\tV_1\in\mathsf{Transf}(\rA_1\rightarrow\rB_1\rE_1)$. Now by the definition of casually ordered channels (Eq.~\eqref{eq:casuallyorderedbipartitechannels}) we have
	\begin{equation*}
		\begin{aligned}
			\Qcircuit @C=1.2em @R=.7em @! R {
				&
				\poloFantasmaCn{\rA_1}\qw&
				\multigate{2}{\tV}&
				\poloFantasmaCn{\rB_1}\qw&
				\qw\\
				&
				\poloFantasmaCn{\rA_2}\qw&
				\ghost{\tV}&
				\poloFantasmaCn{\rB_2}\qw&
				\measureD{e}\\
				& &
				\aghost{\tV}&
				\poloFantasmaCn{\rE}\qw&
				\measureD{e}
			}
		\end{aligned}
		\, = \,
		\begin{aligned}
			\Qcircuit @C=1.2em @R=.7em @! R {
				&
				\poloFantasmaCn{\rA_1}\qw&
				\multigate{1}{\tV_1}&
				\poloFantasmaCn{\rB_1}\qw&
				\qw\\
				& &
				\aghost{\tV_1}&
				\poloFantasmaCn{\rE_1}\qw&
				\measureD{e}\\
				&
				\qw&
				\poloFantasmaCn{\rA_2}\qw&
				\qw&
				\measureD{e}
			}
		\end{aligned}\, .
	\end{equation*}
	This means that $\tV$ and $\tV_1\otimes\tI_{\rA_2}$ are two reversible dilatations of the same channel. By the uniqueness of the reversible dilatation expressed by Lemma \ref{lemma:corollary} we then obtain
	\begin{equation*}
		\begin{aligned}
			\Qcircuit @C=1.2em @R=.7em @! R {
				&
				\poloFantasmaCn{\rA_1}\qw&
				\multigate{2}{\tV}&
				\poloFantasmaCn{\rB_1}\qw&
				\qw\\
				&
				\poloFantasmaCn{\rA_2}\qw&
				\ghost{\tV}&
				\poloFantasmaCn{\rB_2}\qw&
				\qw\\
				& &
				\aghost{\tV}&
				\poloFantasmaCn{\rE}\qw&
				\qw
			}
		\end{aligned}
		\, = \,
		\begin{aligned}
			\Qcircuit @C=1.2em @R=.7em @! R {
				&
				\poloFantasmaCn{\rA_1}\qw&
				\multigate{1}{\tV_1}&
				\qw&
				\poloFantasmaCn{\rB_1}\qw&
				\qw&
				\qw\\
				& &
				\aghost{\tV_1}&
				\poloFantasmaCn{\rE_1}\qw&
				\multigate{2}{\tL}&
				\poloFantasmaCn{\rB_2}\qw&
				\qw\\
				&
				\qw&
				\poloFantasmaCn{\rA_2}\qw&
				\qw&
				\ghost{\tL}&
				\poloFantasmaCn{\rE}\qw&
				\qw\\
				& & & &
				\aghost{\tL}&
				\poloFantasmaCn{\rE_1\rA_2}\qw&
				\measureD{e}
			}
		\end{aligned}\, .
	\end{equation*}
	Once we have defined $\rE_2\coloneqq\rE\rE_1\rA_2$ it only remains to observe that the above diagram is nothing but the thesis, with $\tC_1=\tV_1$ and $\tC_2=\tL$. By construction, $\tC_2\tC_1$ is a reversible dilatation of $\tC$.
\end{proof}
The definition of casually ordered bipartite channel is easily extended to the multipartite case. Here we will only report the definition and the two main theorems of the theory. For their demonstrations we remand to the original article \cite{chiribella} since they are still valid also from our three Axiom as starting point.
\begin{defn}[Causally ordered channel]
	An N-partite channel $\tC^{(N)}$ from $\rA_1\dots\rA_N$ to $\rB_1\dots\rB_N$ is causally ordered if for every $k\le N$ there is a channel $\tC^{(k)}$ from $\rA_1\dots\rA_k$ to $\rB_1\dots\rB_k$ such that
	\begin{equation*}
		\begin{aligned}
			\Qcircuit @C=1.2em @R=.7em @! R {
				&
				\poloFantasmaCn{\rA_1}\qw&
				\multigate{5}{\tC^{(N)}}&
				\poloFantasmaCn{\rB_1}\qw&
				\qw\\
				& \vdots &
				\aghost{\tC^{(N)}}&
				\vdots&
				\\
				&
				\poloFantasmaCn{\rA_k}\qw&
				\ghost{\tC^{(N)}}&
				\poloFantasmaCn{\rB_k}\qw&
				\qw\\
				&
				\poloFantasmaCn{\rA_{k+1}}\qw&
				\ghost{\tC^{(N)}}&
				\poloFantasmaCn{\rB_{k+1}}\qw&
				\measureD{e}\\
				& \vdots &
				\aghost{\tC^{(N)}}&
				\vdots&
				\\
				&
				\poloFantasmaCn{\rA_{N}}\qw&
				\ghost{\tC^{(N)}}&
				\poloFantasmaCn{\rB_{N}}\qw&
				\measureD{e}\\
			}
		\end{aligned}
		\, = \,
		\begin{aligned}
			\Qcircuit @C=1.2em @R=.7em @! R {
				&
				\poloFantasmaCn{\rA_1}\qw&
				\multigate{2}{\tC^{(k)}}&
				\poloFantasmaCn{\rB_1}\qw&
				\qw\\
				& \vdots &
				\aghost{\tC^{(k)}}&
				\vdots&
				\\
				&
				\poloFantasmaCn{\rA_1}\qw&
				\ghost{\tC^{(k)}}&
				\poloFantasmaCn{\rB_1}\qw&
				\qw\\
				&
				\qw&
				\poloFantasmaCn{\rA_{k+1}}\qw&
				\qw&
				\measureD{e}\\
				&\vdots& & \vdots&\\
				&
				\qw&
				\poloFantasmaCn{\rA_{N}}\qw&
				\qw&
				\measureD{e}\\
			}
		\end{aligned}\, .
	\end{equation*}
\end{defn}
The definition means that the output systems $\rB_1\dots\rB_k$ are outside the casual future of any input system $\rA_l$ with $l>k$. Causally ordered channels can be characterized as follows.
\begin{theorem}[Causal ordering is memory for general N]
	\label{theorem:causalorderingforgeneralN}
	An N-partite channels $\tC^{(N)}$ from $\rA_1\dots\rA_N$ to $\rB_1\dots\rB_N$ is causally ordered if and only if there exist a sequence of memory systems $\{\rE_k\}_{k=0}^N$ with $\rE_0=\rI$ and a sequence of channels $\{\tV_k\}_{k=1}^N$, with $\tV_k\in\mathsf{Transf}(\rA_k\rE_{k-1}\rightarrow\rB_k\rE_k)$ such that
	\begin{equation*}
		\begin{aligned}
			\Qcircuit @C=1.2em @R=.7em @! R {
				&
				\poloFantasmaCn{\rA_1}\qw&
				\multigate{2}{\tC^{(N)}}&
				\poloFantasmaCn{\rB_1}\qw&
				\qw\\
				& \vdots &
				\aghost{\tC^{(N)}}&
				\vdots&
				\\
				&
				\poloFantasmaCn{\rA_N}\qw&
				\ghost{\tC^{(N)}}&
				\poloFantasmaCn{\rB_N}\qw&
				\qw
			}
		\end{aligned}
		\, = \,
		\begin{aligned}
			\Qcircuit @C=1.1em @R=.7em @! R {
				&
				\poloFantasmaCn{\rA_1}\qw&
				\multigate{1}{\tV_1}&
				\poloFantasmaCn{\tB_1}\qw&
				\qw& & 
				\poloFantasmaCn{\rA_2}\qw&
				\multigate{1}{\tV_2}&
				\poloFantasmaCn{\rB_2}\qw&
				\qw&
				\cdots& &
				\poloFantasmaCn{\rA_N}\qw&
				\multigate{1}{\tV_N}&
				\poloFantasmaCn{\rB_N}\qw&
				\qw\\
				& &
				\aghost{\tV_1}&
				\qw&
				\poloFantasmaCn{\rE_1}\qw&
				\qw&
				\qw&
				\ghost{\tV_2}&
				\poloFantasmaCn{\rE_2}\qw&
				\qw&
				\cdots& &
				\poloFantasmaCn{\rE_{N-1}}\qw&
				\ghost{\tV_N}&
				\poloFantasmaCn{\rE_N}\qw&
				\measureD{e}\\
			}
		\end{aligned}\, .
	\end{equation*}
	Moreover, $\tV_N\dots\tV_1$ is a reversible dilation of $\tC^{(N)}$.
\end{theorem}
We have also a uniqueness result:
\begin{corollary}[Uniqueness of the reversible dilatation]
	Let $\{\tV_k\}_{k=1}^N$, $\tV_k\in\mathsf{Transf}(\rA_k\rE_{k-1}\rightarrow\rB_k\rE_k)$ be a reversible realization of the casually ordered channel $\tC^{(N)}$ as a sequence of channels with memory, as in Theorem \ref{theorem:causalorderingforgeneralN}. Suppose that $\{\tV_k^\prime\}_{k=1}^N$, $\tV_k^\prime\in\mathsf{Transf}(\rA_k\rE_{k-1}^\prime\rightarrow\rB_k\rE_kì^\prime)$ is another reversible realization of $\tC^{(N)}$ as a sequence of channels with memory. Then there exist a channel $\tR$ from $\rE_N$ to $\rE_N^\prime$ such that
	\begin{equation*}
		\begin{aligned}
			&\begin{aligned}
				\Qcircuit @C=1.2em @R=.7em @! R {
					&
					\poloFantasmaCn{\rA_1}\qw&
					\multigate{1}{\tV_1^\prime}&
					\poloFantasmaCn{\tB_1}\qw&
					\qw& & 
					\poloFantasmaCn{\rA_2}\qw&
					\multigate{1}{\tV_2^\prime}&
					\poloFantasmaCn{\rB_2}\qw&
					\qw&
					\cdots& &
					\poloFantasmaCn{\rA_N}\qw&
					\multigate{1}{\tV_N^\prime}&
					\poloFantasmaCn{\rB_N}\qw&
					\qw\\
					& &
					\aghost{\tV_1^\prime}&
					\qw&
					\poloFantasmaCn{\rE_1^\prime}\qw&
					\qw&
					\qw&
					\ghost{\tV_2^\prime}&
					\poloFantasmaCn{\rE_2^\prime}\qw&
					\qw&
					\cdots& &
					\poloFantasmaCn{\rE_{N-1}^\prime}\qw&
					\ghost{\tV_N}&
					\poloFantasmaCn{\rE_N^\prime}\qw&
					\qw\\
				}
			\end{aligned}\, =\\
			&=\,
			\begin{aligned}
				\Qcircuit @C=1.2em @R=.7em @! R {
					&
					\poloFantasmaCn{\rA_1}\qw&
					\multigate{1}{\tV_1}&
					\poloFantasmaCn{\tB_1}\qw&
					\qw& & 
					\poloFantasmaCn{\rA_2}\qw&
					\multigate{1}{\tV_2}&
					\poloFantasmaCn{\rB_2}\qw&
					\qw&
					\cdots& &
					\poloFantasmaCn{\rA_N}\qw&
					\multigate{1}{\tV_N}&
					\poloFantasmaCn{\rB_N}\qw&
					\qw&\qw&\qw\\
					& &
					\aghost{\tV_1}&
					\qw&
					\poloFantasmaCn{\rE_1}\qw&
					\qw&
					\qw&
					\ghost{\tV_2}&
					\poloFantasmaCn{\rE_2}\qw&
					\qw&
					\cdots& &
					\poloFantasmaCn{\rE_{N-1}}\qw&
					\ghost{\tV_N}&
					\poloFantasmaCn{\rE_N}\qw&
					\gate{\tR}&
					\poloFantasmaCn{\rE_N^\prime}\qw&
					\qw\\
				}
			\end{aligned}\, .
		\end{aligned}
	\end{equation*}
\end{corollary}

\subsection{No bit commitment}

The bit commitment protocol defined in Section~\ref{sec:defn-bit-commitment} is generally implemented with sequences of channels with memory, that can be used to describe sequences of moves of Alice or Bob. In this scenario, the memory systems are the private systems available to a party, while the other input-output systems are the systems exchanged in the communication with the other party.\\
We recall that for every theory that satisfy the starting hypothesis, this proof has an absolutely general validity: for every kind of input states and every possible strategy adopted, including both atomic and non-atomic transformations.
\begin{theorem}[No perfectly secure bit commitment]
	If a N-round protocol is perfectly concealing, then there is a perfect cheating.
\end{theorem}
\begin{proof}
	We first prove the impossibility for protocols that do not involve the exchange of classical information. Let $\tA_0,\,\tA_1\in\mathsf{Transf}(\rA_1...\rA_N\rightarrow\rB_1...\rB_N\rF_N)$ two causally ordered N-partite channels (here the last output system of the causally ordered channels is the bipartite system $\rB_N,\,F_N$), representing Alice's move to encode the bit value $b=0,\,1$, respectively. The system $\rF_N$ is the system sent from Alice to Bob at the opening phase in order to unveil the value of the bit. If the protocol is perfectly concealing, then the reduced channels before the opening phase must be indistinguishable, namely $(e|_{\rF_N}\tA_0=(e|_{\rF_N}\tA_1\coloneqq\tC$. Now, due to Theorem \ref{thm:RevDila}, there exist two reversible dilatations $\tV_0\in\mathsf{Transf}(\rA_1...\rA_N\rightarrow\rB_1...\rB_N\rF_N\rG_0)$ and $\tV_1\in\mathsf{Transf}(\rA_1...\rA_N\rightarrow\rB_1...\rB_N\rF_N\rG_1)$ for $\tA_0$ and $\tA_1$, respectively. Since $\tV_0$ and $\tV_1$ are also two  dilatations of the channel $\tC$, due to Lemma \ref{lemma:corollary} there is a channel $\tR$ from $\rF_N\rG_0$ to $\rF_N\rG_1$ such that $\tV_1=\tR\tV_0$. Applying this channel to her private systems, Alice can switch from $\tV_0$ to $\tV_1$ just before the opening. Discarding the auxiliary system $\rG_1$, this yields channel $\tA_1$.\\
	The cheating is perfect, since Alice can play the strategy $\tV_0$ until the end of the commitment and decide the bit value before the opening without being detected by Bob. The above reasoning can be extended to N-round protocols involving the exchange of classical information. Indeed classical messages can be modeled by measure-and-prepare channels where the observation states are perfectly distinguishable. The fact that some systems can only be prepared in perfectly distinguishable states will be referred as to "communication interface" of the protocol. In this case, to construct Alice's cheating strategy we can first take the reversible dilatations $\tV_0$, $\tV_1$ and the channel $\tR$ such that $\tV_1=\tR\tV_0$. In order to comply with the communication interface protocol, one can compose $\tV_0$ and $\tV_1$ with classical channels on all system that must be "classical" before the opening, thus obtaining two channels $\tD_0$ and $\tD_1$ that are no longer reversible but still satisfy $\tD_1=\tR\tD_0$. Discarding the auxiliary system $\rG_1$ and, of required by the communication interface, applying a classical channels on $\rF_N$, Alice then obtains channel $\tA_1$. Again, this strategy allows Alice to decide the value of the bit just before the opening without being detected.
\end{proof}

\chapter{PR-Boxes}
\label{chap:pr-box}
In this Chapter we will analyse the probabilistic theory \cite{barret2007,barret2005,d'ariano-tosini} corresponding to the popular PR-boxes model introduced in Ref.~\cite{PRorigial}.\\

In particular, in the first Section we will retrace the crucial steps and the underlying reasons that led to the development of the PR-box model.\\
Then, in the second Section, after the formalization of the PR-boxes model in the language of operational probabilistic theories, we will report our results: the general POVM that grants perfect discriminability between any two pure states in the bi-partite case, the existence and uniqueness of the purification exclusively for the maximally mixed state (again in the bi-partite scenario) and finally some consideration on the general case of N-partite boxes.

\section{Why PR-Boxes?}

One of the most striking feature of quantum theory is certentantly non-locality. In fact, since the very beginning of the theory the incompleteness of the Copenhagen interpretation of quantum mechanics in relation to the violation of local causality was one of the main discussed aspect. In 1935 Einstein Podolsky and Rosen published the famous article of the EPR paradox \cite{EPR}. The thought experiment generated a great deal of interest in the following years. Their notion of a "complete description" was later formalized by the suggestion of \textit{hidden variables} that determine the statistics of measurement results, but to which an observer does not have access.\\
In 1964 John Bell proved that some predictions of quantum mechanics cannot be reproduced by any theory of local physical variables \cite{bell}. Although Bell worked within non-relativistic quantum theory, the definition of local variable is relativistic: a local variable can be influenced only by events in its backward light cone, not by events outside, and can influence events in its forward light cone only. Quantum mechanics, which does not allow us to transmit signals faster than light (super-luminal signalling), preserves relativistic causality. But quantum mechanics does not always allow us to consider distant systems as separate, as Einstein assumed.\\
Now quantum non-locality has been experimentally verified under different physical assumptions. Any physical theory that aims at superseding or replacing quantum theory should account for such experiments and therefore must also be non-local in this sense; quantum non-locality is a property of the universe that is independent of our description of nature.\\

So quantum non-locality is an essential feature of quantum theory but it often appear in a negative light. In 1994 Popescu and Rohrlich published a work \cite{PRorigial} where they proposed to show quantum non-locality in a more positive light. They investigated the inversion of the logical approach to quantum mechanics, considering quantum non-locality as an axiom instead of as a theorem and wondering what non-locality together with relativistic causality would imply.\\
They found that quantum theory is only one of a class of non-local theories consistent with causality, and not even the most non-local. In fact, in a certain sense, non-locality can be quantified.\\

In 1969, John Clauser, Michael Horne, Abner Shimony, and Richard Holt reformulated Bell's inequality in a manner that best suites experimental testing, the homonymous CHSH inequality \cite{CHSH}. CHSH inequality, restricted to any classical theory, states that a particular algebraic combination of correlations lies between -2 and 2. This bound is obviously violated in quantum mechanics, where in fact the CHSH inequality allows a maximum value given by Cirel'son's theorem as $2\sqrt{2}$ \cite{cirelson}. However, Popescu and Rohrlich wrote down a set of correlations that return a value of 4 for the CHSH expression, the maximum value algebraically possible, and that yet are non(super-luminar)-signalling. A question that now rises spontaneously is why does quantum theory not allow these strongly non-local correlations.\\

In the hope of making further progress with this question these correlations have been investigated in the context of a theory with well-defined dynamics.\\
Abstractly this scenario may be described by introducing two observers that have access to a black box. Each observer selects an input from a range of possibilities and obtains an output. The box determines a joint probability for each output pair given each input pair. It is clear that a quantum state provides a particular example of such a box, with input corresponding to measurement choice and output to measurement outcome. More generally boxes can be divided into different types. Some will allow the observers to signal to one another via their choice of input, and correspond to two-way classical channels, as introduced by Shannon. Others will not allow signalling - it is well known, for example, that any box corresponding to an entangled quantum state will not. This is necessary for compatibility between quantum mechanics and special relativity. Among the non-signalling boxes, some will violate a Bell-type inequality, and we refer to any such a box as non-local. As we have described above, in terms of our boxes, there are some boxes that are non-signalling but are more non-local than any box allowed by quantum theory.

\section{PR-Boxes as OPT}
\label{sec:prbox-as-opt}
The model describes $N$ correlated boxes (in the original paper \cite{PRorigial} it was $N=2$) in a casual context. Each box is represented by the same elementary system $\rA$, with the $N$ correlated boxes represented by the composite system $\rA^{\otimes N}$. We will consider the simplest situation where each box has both input and output as binary variables. On each elementary system $\rA_i$, $i=1,\ldots,N$ only two atomic binary observation tests are allowed, say $B^{(0)}\coloneqq\{b^{(0)},e_\rA-b^{(0)}\}$ or $B^{(1)}\coloneqq\{b^{(1)},e_\rA-b^{(1)}\}$, with $b^{(0)},b^{(1)}\in\eff{A}$ atomic effects, and $e_\rA$ the deterministic effect of system $\rA$. Notice that the since the model is causal, the deterministic effect $e_\rA$ is unique and it is $e_{\rA^{\otimes N}}=\otimes^N e_\rA$.\\
The probabilistic model is typically presented in terms of the probability function $P:(b_1,b_2,\ldots, b_N|B_1,B_2,\ldots B_N)\mapsto [0,1]$, with $B_i\in\{B^{(0)},B^{(1)}\}$ and $b_i\in\{b^{(0)},b^{(1)}\}$, for every $i=1,\ldots,N$, which returns the probability of the outcomes $b_1,b_2,\ldots,b_N$ given the observation tests $B_1,B_2,\ldots B_N$. The constraint imposed on the function $P$ is \textit{no-signalling}, i.e.
\begin{equation}
	\label{eq:no-signalling}	
	\begin{aligned}
		\sum_{b_k=b^{(0)},e-b^{(0)}} &P(b_1,\ldots, b_N|B_1,\ldots B_N)=\\
		=&P(b_1,\ldots, b_{k-1},b_{k+1},\ldots,b_N|B_1,\ldots B_{k-1},B_{k+1},\ldots, B_N).
	\end{aligned}
\end{equation}

\subsection{States, effects and transformations}

We will start our analysis with the boxes of the elementary system $\rA$ that are necessary local. They are described in terms of the following probability distribution
\begin{equation}
	\label{eq:prbox1}
	p_{\alpha\beta}(a|x)=
	\begin{cases}
		1& a=\alpha x\oplus \beta\\
		0 & \text{otherwise}
	\end{cases},
\end{equation}
with $\alpha,\,\beta=0,1$. This four probability distributions correspond to the four pure states of system $\rA$.\\
In fact the elementary system $\rA$ has dimension $\dim(\rA)=3$, namely its states are described by vectors in $\mathbb{R}^3$ ($\mathsf{St}_\mathbb{R}(\rA)=\mathsf{Eff}_\mathbb{R}(\rA)=\mathbb{R}^3$). The four pure normalized states \eqref{eq:prbox1} can be represented by the following vectors
\begin{equation}
	\label{eq:local-states}
	\begin{aligned}
		\omega_0 =
		\begin{pmatrix}
			1\\0\\1
		\end{pmatrix},\:
		\omega_1 =
		\begin{pmatrix}
			0\\-1\\1
		\end{pmatrix},\:
		\omega_2 =
		\begin{pmatrix}
			-1\\0\\1
		\end{pmatrix},\:
		\omega_3 =
		\begin{pmatrix}
			0\\1\\1
		\end{pmatrix},
	\end{aligned}
\end{equation}
where the correspondence between $\omega_n$ and the probability rule $p_{\alpha\beta}$ is given by $\alpha\beta=$(binary form of $n$).\\
The convex set of states normalized is then represented by a square (see the square in the plane $z=1$ in Fig.~\ref{fig:elementary-system}).\\
\begin{SCfigure}
	\begin{overpic}[width=.5\columnwidth]{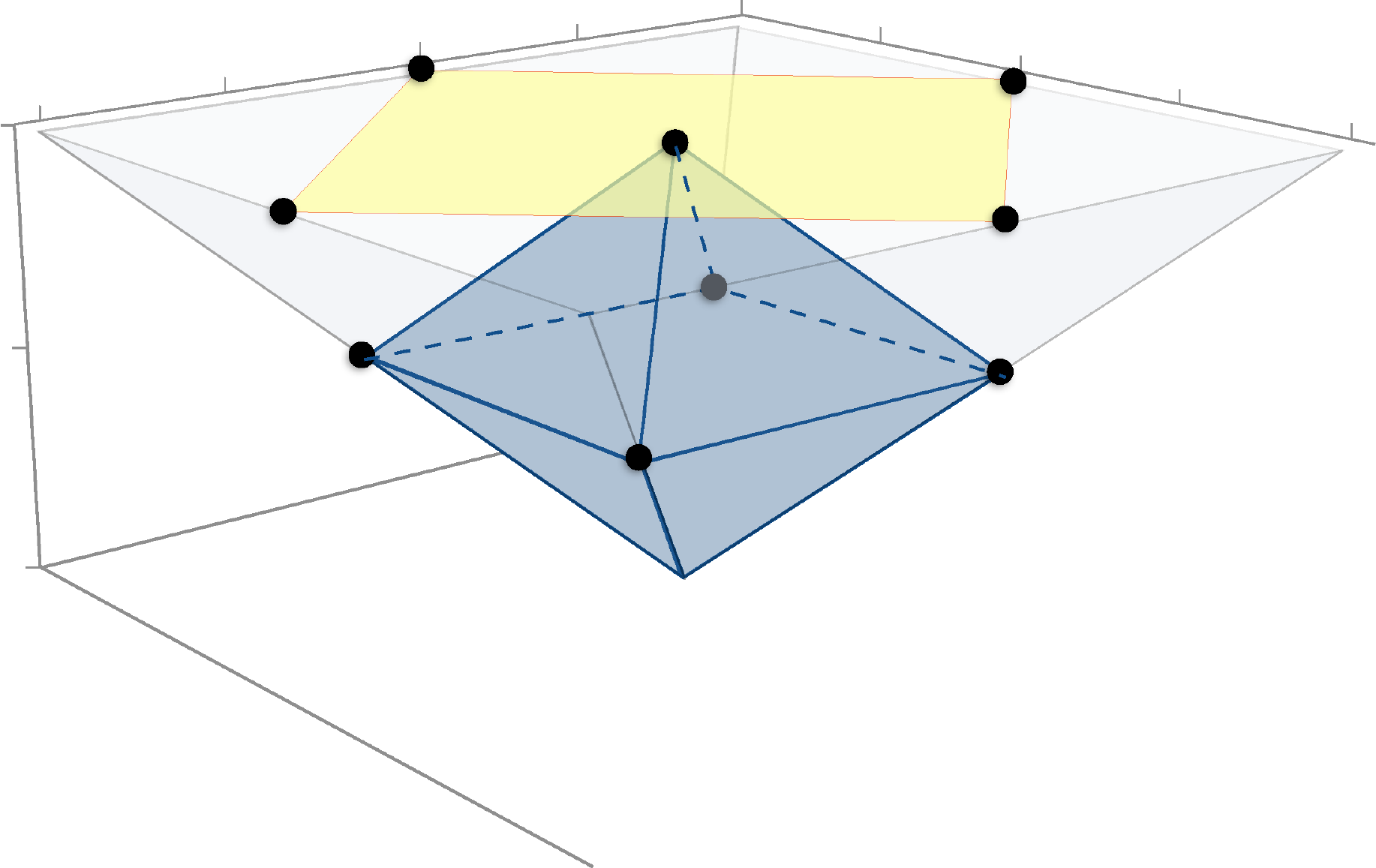}
		\put (15, 44) {\color{red}$\omega_1$}
		\put (67, 54) {\color{red}$\omega_3$}
		\put (29, 54) {\color{red}$\omega_0$}
		\put (74, 44) {\color{red}$\omega_2$}
		\put (48, 27) {\color{blue1}$b_1$}
		\put (73, 33) {\color{blue1}$b_2$}
		\put (22, 33) {\color{blue1}$b_0$}
		\put (50, 38) {\color{blue1}$b_3$}
		\put (45, 53) {\color{blue1}$\bar e$}
		\put (72.5, 60) {\color{black}$0$}
		\put (52.5, 64) {\color{black}$1$}
		\put (93, 56) {\color{black}$-1$}
		\put (-2.5, 52.5) {\color{black}$1$}
		\put (-2, 36.5) {\color{black}$0$}
		\put (-4.5, 21) {\color{black}$-1$}
		\put (0, 56) {\color{black}$-1$}
		\put (29, 61) {\color{black}$0$}
	\end{overpic}
	\caption{\textbf{ Elementary system of the PR-boxes theory.} This picture depicts the ``squit'' elementary system often considered in generalized probabilistic theories (in analogy to the ``bit'' and the ``qubit'' which are elementary systems of classical and quantum theory respectively).  The system is specified by its sets of states and effects here represented as vectors in $\mathbb{R}^3$.  The convex set of normalized states is represented by the square at the top, while the convex set of effects corresponds to the truncated cone.}
	\label{fig:elementary-system}
\end{SCfigure}
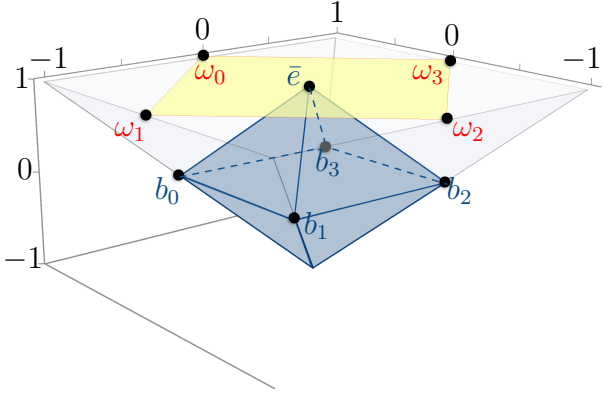

The set of effects of system $\rA$ is defined as the set of vectors $b$ such that $0 \le \Tr[b^T \omega] \le 1$ for every state $\omega\in\st{A}$, with $p_{y|x} = \Tr[{b^{(y)}}^T \omega_x]$ the rule providing the probability associated to an effect $b^{(y)}$ on a state $\omega_x$. This leads to the truncated cone of effects in Fig.~\ref{fig:elementary-system}, with extremal points given by
\begin{align*}
	b^{(0)} =
	\begin{pmatrix}
		\frac{1}{2}\\[2pt]\frac{1}{2}\\[2pt]\frac{1}{2}
	\end{pmatrix},\;
	b^{(1)} =
	\begin{pmatrix}
		-\tfrac{1}{2}\\[2pt]\tfrac{1}{2}\\[2pt]\tfrac{1}{2}
	\end{pmatrix},\;
	b^{(2)} =
	\begin{pmatrix}
		-\tfrac{1}{2}\\[2pt]-\tfrac{1}{2}\\[2pt]\tfrac{1}{2}
	\end{pmatrix},\;
	b^{(3)} =
	\begin{pmatrix}
		\tfrac{1}{2}\\[2pt]-\tfrac{1}{2}\\[2pt]\tfrac{1}{2}
	\end{pmatrix}.
\end{align*}
The deterministic effect (the effect $e_\rA$ such that $\Tr[e_\rA^T \omega]=1$ for every state $\omega\in\st{A}$) is the vector $e_\rA=(0,0,1)^T$ (when no ambiguity arises we will simply denote the deterministic effect with $e$).\\

Our analysis now proceeds towards the composite system $\rA\otimes \rA$. In this case the probability function $p(a,b|x,y)$ form a table with $2^4$ entries, although these are not all independent due to the constraints of Eq.~\eqref{eq:no-signalling}. The dimension of the set of boxes is found by subtracting the number of independent constraints from $2^4$, and turns out to be 8 (we notice that $8=\dim\mathsf{St}_1(\rA\otimes\rA)$, while $\dim\mathsf{St}_\mathbb{R}(\rA\otimes \rA)=\dim\mathsf{St}_\mathbb{R}(\rA)\dim\mathsf{St}_\mathbb{R}(\rA)=9$, so due to Theorem~\ref{theorem:product-rule-for-composite-systems} it satisfies local discriminability). In this case we will have no more a square like in Fig.~\ref{fig:elementary-system} but a polytope with 24 vertices. The vertices will be boxes that satisfy all of the constraints and
saturates a sufficient number of the positivity constraints to be uniquely determined. These 24 bilocal pure states may be divided into two classes.\\
The \textit{local boxes}, given by the following 16 probability distributions
\begin{equation}
	\label{eq:14} 
	p(a,b|x,y)=
	\begin{cases}
		1& a=\alpha x\oplus\beta\\
		1& b=\gamma y\oplus\delta\\
		0& \text{otherwise}
	\end{cases},
\end{equation}
with $\alpha,\beta,\gamma,\delta=0,1$, and the \textit{non-local boxes}, given by the 8 probability distributions
\begin{equation}
	\label{eq:15}
	p_{\alpha\beta\gamma}(a,b|x,y)=
	\begin{cases}
		1/2& a\oplus b=xy\oplus \alpha x\oplus \beta y\oplus\gamma\\
		0 & \text{otherwise}
	\end{cases}
\end{equation}
with $\alpha,\beta,\gamma=0,1$.\\
For convenience we can represent states and effects as $3\times 3$ real matrices rather than as vectors in $\mathbb{R}^9$.
The 16 local state of the bipartite system are nothing else that the factorized pure states
\begin{align}
	\label{eq:local-bipartite-states}
	\Omega_{4i+j} \coloneqq \omega_i \otimes \omega_j^T,
\end{align}
where $i, j \in \{ 0, 1, 2, 3 \}$, and the 8 non-local states (playing the role of entangled states) are represented by the following matrices
\begin{equation}
	\label{eq:non-local-bipartite-states}
	\begin{aligned}
		\Omega_{16} & := \frac12
		\begin{pmatrix}
			-1 & 1 & 0 \\
			1 & 1 & 0 \\
			0 & 0 & 2
		\end{pmatrix},
		& \Omega_{17} & := \frac12
		\begin{pmatrix}
			-1 & -1 & 0 \\
			-1 & 1 & 0 \\
			0 & 0 & 2
		\end{pmatrix}, 
		& \Omega_{18}  := \frac12
		\begin{pmatrix}
			1 & -1 & 0 \\
			-1 & -1 & 0 \\
			0 & 0 & 2
		\end{pmatrix},\\
		\Omega_{19} &:= \frac12
		\begin{pmatrix}
			-1 & 1 & 0 \\
			-1 & -1 & 0 \\
			0 & 0 & 2
		\end{pmatrix}, 
		& \Omega_{20} & := \frac12
		\begin{pmatrix}
			-1 & -1 & 0 \\
			1 & -1 & 0 \\
			0 & 0 & 2
		\end{pmatrix},
		& \Omega_{21}  := \frac12
		\begin{pmatrix}
			1 & -1 & 0 \\
			1 & 1 & 0 \\
			0 & 0 & 2
		\end{pmatrix}, \\
		\Omega_{22} & := \frac12
		\begin{pmatrix}
			1 & 1 & 0 \\
			-1 & 1 & 0 \\
			0 & 0 & 2
		\end{pmatrix},
		& \Omega_{23} & := \frac12
		\begin{pmatrix}
			1 & 1 & 0 \\
			1 & -1 & 0 \\
			0 & 0 & 2
		\end{pmatrix}.&
	\end{aligned}
\end{equation}
From a generic probability rule $p_{\alpha\beta\gamma}$ we can identify one of the above matrix $\Omega_n$ by the equation: $n=15+\left[(3+3\alpha+4\beta+6\gamma)\mod8\right]$.\\

Again, the probability associated to a bipartite effect $B_y$ on the state $\Omega_x$ is given by $p_{y|x} = \Tr[B_y^T \Omega_x]$. Accordingly, the set of bipartite effects is easily derived via the consistency condition $\Tr[B_j^T \Omega_i] \ge 0$ for every $j \in [0, 23]$. It follows that the only admissible extremal effects are the 16 factorized matrices
\begin{align}  
	B_{4i+j} := b^{(i)} \otimes {b^{(j)}}^T.
\end{align}
This is a relevant feature of PR-boxes model, whose strong correlation incapsulated in the eight non-local states $\Omega_x$, $x\in[16,23]$ are incompatible with any in principle admissible non-factorized measurements. This feature has been firstly noticed in Ref.~\cite{short2010} and later in Ref.~\cite{tosini2017} where all possible theories compatible with the squit local system have been classified (among these theories is the dual version of PR-boxes that only have factorized states but eight non-local effects). Finally, the deterministic effect for the bipartite system is $e\otimes e^T$.\\

Analogously one defines the convex set of states and the convex set of effects for the arbitrary $N$-partite system $\rA^{\otimes N}$. However, while for effect is nothing but a trivial generalization, for the states the discussion is not so straight and some interesting features arise. We will dedicate a following Section to discuss some of the most immediate aspect about the $N$-partite system, we will start our discussion from the tripartite boxes.\\

We can now turn our attention to the transformations of the theory. We focus here on the reversible transformations which are of interest for the present paper results. The set $\mathcal{U}(\rA)$ of reversible transformations of the system $\rA$ coincides with the finite group of symmetries of the square (the dihedral group of order eight $D_8$ containing four rotations and four reflections). In the chosen representation we have
\begin{equation}
	\label{eq:single-system-unitaries}
	\begin{aligned}
		\mathcal{U}(\rA)=\{U_k^s: k=0,\ldots,3, s=\pm\}\\
		U_k^s =
		\begin{pmatrix}
			\cos \frac{\pi k}2 & -s \sin \frac{\pi k}2 & 0 \\
			\sin \frac{\pi k}2 & s \cos \frac{\pi k}2 & 0 \\
			0 & 0 & 1
		\end{pmatrix}.
	\end{aligned}
\end{equation}
The matrices $U_k^{+}$ and $U_k^{-}$ representing the four rotations and the four reflections of the square respectively. When we apply them to the four pure states of Eq.~\eqref{eq:local-states} we have
\begin{align}
	\label{eq:reversible-mapping-on-local-states}
	\omega_{j+k} & = U_k^+\, \omega_{j}, & \omega_{k} & = U_k^-\, \omega_{j+k}\,,
\end{align}
for $j\in\left[0,3\right]$ and where the sum is$\mod 3$, i.e. if $j=2$ and $k=2$ then $j+k=0$.\\
We notice that these transformations are atomic. In Ref.~\cite{d'ariano-tosini} all atomic transformations of the squit system have been classified.\\
The set of reversible transformations $\mathcal{U}(A\otimes A)$ of the composite system $\rA\otimes\rA$ (which has been derived in Ref.~\cite{gross}) is  
\begin{equation}
	\label{eq:reversible}
	\mathcal{U}(\rA\otimes \rA)=\{ W^i(U_j^{s_1} \otimes U_k^{s_2})\}\quad i=0,1,\; 0\leq j,k\leq 3,\; s_1,s_2=\pm,
\end{equation}
with $W$ the \textit{swap} map, namely the map that exchanges the two subsystems. This means that any reversible map corresponds to the tensor product of single system reversible transformations with, possibly, the application of the swap. As noticed in Ref.~\cite{gross} reversible transformations cannot create entanglement. This result has been extended in Ref.~\cite{al-safi2014} to the composition of an arbitrary number $N$ of systems $\rA^{\otimes N}$, showing that also in that case the set of reversible multipartite transformations is generated by local reversible operations and permutations of systems.\\

We notice that in the PR-boxes theory any non-local bipartite pure state can be reversibly mapped to any other non-local bipartite pure state. Moreover, this mapping can be done via the local application of a single system reversible map. For example, starting from the state $\Omega_{16}$, one has
\begin{align}
	\label{eq:reversible-mapping}
	\Omega_{16+k} & = (U_k^+\otimes I)\Omega_{16}, & \Omega_{23-k} & = (U_k^-\otimes I) \Omega_{16}.
\end{align}
We finally remark that the set of reversible transformations of system $\rA$, $\tU(\rA)$, in terms of description by the probability distributions in Eqs.~\eqref{eq:prbox1},~\eqref{eq:14},~\eqref{eq:15} correspond to a local relabelling defined by the operations
\begin{equation}
	\label{eq:local-relabelling}
	\begin{aligned}
		x&\rightarrow x\oplus 1,\qquad a\rightarrow a\oplus\alpha x\oplus \gamma,\\
		y&\rightarrow y\oplus 1,\qquad b\rightarrow b\oplus\beta y\oplus \gamma.
	\end{aligned}
\end{equation}

\subsection{Discriminability between PR-boxes}
\label{sec:discriminability-PRboxes}
An important question to address when dealing with the PR-boxes theory is the following: given two deterministic state of the theory is it possible to discriminate between them? We will address to this question only for the extremal points of the polytope, i.e. the pure states, in order to analyse a way to perfect discriminate between them.\\

For local boxes it is easy to show that there exist POVMs that are able to perfect discriminate between every two of the four possible state in Eq.~\eqref{eq:local-states} \cite{d'ariano-tosini}.\\
For what concern bipartite boxes, while it is trivial to find perfectly discriminable POVMs for each pair of the 16 local boxes (since they are simply the tensor product of the 4 local boxes, so also the discriminable POVMs are just the tensor product of the perfectly discriminable POVMs of the local boxes) it is a bit more elaborate to explore the discriminability between non-local boxes.\\

To help in our analysis we introduce the following table regarding non-local boxes labelled by $(\alpha,\beta,\gamma)$ as described in Eq.~\eqref{eq:15}:

\begin{table}[h]
	\centering
	\begin{tabular}{c|c|c}
		$x$ & $y$ & $a\oplus b$ \\
		\hline
		0  &  0  & $\gamma$ \\
		0  &  1  & $\beta\oplus\gamma$ \\
		1  &  0  & $\alpha\oplus\gamma$ \\
		1  &  1  & $1\oplus\alpha\oplus\beta\oplus\gamma$ \\
	\end{tabular}
	\caption{Correlations for a generic $(\alpha,\beta,\gamma)$ non-local box for all the possible given input combinations of $x$ and $y$.}
\end{table}

It is now easy to verify that for every two different non-local boxes, i.e. for every choice of two different combinations of $(\alpha,\beta,\gamma)$: $c_1=(\alpha_1,\beta_1,\gamma_1)$ and $c_2=(\alpha_2,\beta_2,\gamma_2)$, there is always at least one input combination $(x,y)$, whose output relation is equal to 0 for $c_1$ and to 1 for $c_2$. If we denote with $e$ the deterministic effect of the bipartite system $\rA\otimes\rA$ and $a\coloneqq b^{(3(1-x))} \otimes b^{(3(1-y))}+b^{(1+x)} \otimes b^{(1+y)}$, the POVM $\{a,e-a\}$ is able to perfectly discriminate between the two chosen non-local boxes.\\

Finally, the last remark regards the discrimination between one local and one non-local bipartite box. We can help ourselves with the two TABLEs of Appendix \ref{app:tables} that follows the notation of Eqs.~\eqref{eq:14},\eqref{eq:15}. For every pair of bipartite boxes (one local and one  non-local) there is always a combination of $(x,y)$ such that in one case the outcome relation $a\oplus b$ is equal to 0 for one box and 1 for the other. So we can construct a perfectly discriminating POVM following the same strategy of above.

\subsection{Purification in the PR-box (bipartite restriction) model}
\label{sec:purification-of-maximally-mixed-states}
Until now nothing has been said about purification in the PR-box model.\\
The following discussion represent an important result of the theory, valid in the general context of $N$-partite boxes (with $N$ finite), but we will operate in the significant restriction of admitting no more than bipartite correlated boxes. For the rest of the subsection when we will refer to PR-box model/theory we will mean the model under this limitation.\\

Let begin with some general definitions and preliminary considerations.
\begin{defn}[Transitivity]
	The ability to transform any pure state into any other by means of reversible transformations will be called transitivity, meaning that the action of the set of reversible transformations is transitive on the set of pure states.
\end{defn}
\noindent Based on this definition, PR-box model clearly enjoys this property. In fact the transformations of Eq.~\eqref{eq:single-system-unitaries} is transitive on the set of pure states of Eq.~\eqref{eq:local-states}.\\
Among the numerous consequences that transitivity implies one will be of our interest, the uniqueness of the maximally mixed state. So it is in order to properly define what a maximally mixed state is.
\begin{defn}[Maximally mixed state]
	If a state is invariant under the action of every reversible transformation, then it is a maximally mixed state.
\end{defn}
\noindent We will not report here the demonstration of the uniqueness of the maximally mixed state from transitivity, for it we refer to Ref.~\cite{d'ariano-libro}.\\
From the transformations of Eq.~\eqref{eq:single-system-unitaries} and the vector representation of the pure state in Eq.~\eqref{eq:local-states} it is not difficult to write down the maximally mixed state:
\begin{equation}
	\label{eq:maximally-mixed-state}
	\mu =
	\begin{pmatrix}
		0\\0\\1
	\end{pmatrix}\,.
\end{equation}
We are now in position to state the main result of this Section.
\begin{theorem}
	Given a system $\otimes^N\rA$, the maximally mixed state $\mu^{\otimes N}\in\st{\otimes^NA}$ of Eq.~\eqref{eq:maximally-mixed-state} is the unique internal state that is purificable and its purification is unique up to a reversible transformation on the purifying systems.
\end{theorem}
\begin{proof}
	The proof is divided in two part. In the first one we will prove the thesis for $N=1$, then in the second part it is extended to an arbitrary number of systems.\\
	Given the system $\rA_1$ we will consider a system $\rA_2$ as the purifying system. The pure states of $\rA_1\otimes\rA_2$ are the 24 pure states $\Omega_{i}$, for $i=0,\ldots,23$ of Eqs.~\eqref{eq:local-bipartite-states},~\eqref{eq:non-local-bipartite-states}. We know that 16 of them (namely the local ones, that are expressed in Eq.~\eqref{eq:local-bipartite-states}) are separable states. So if we consider the marginal state obtained by applying the deterministic effect $e_{\rA_2}$ on the purifying system $\rA_2$ on these separable states we get one of the local 4 state represented in Eq.~\eqref{eq:local-states}, that are pure. However, if we repeat the same procedure on the 8 non-local pure states (represented in Eq.~\eqref{eq:non-local-bipartite-states}) we get the same state for all of them: the maximally mixed state $\mu\in\mathsf{St}(\rA_1)$. So the unique internal state that can be purificated is the maximally mixed state and since the 8 non-local bipartite pure states are all mapped to any other non-local bipartite pure state by the application of a single system reversible map, as shown in Eq.~\eqref{eq:reversible-mapping}, the purification is unique up to a reversible transformation on the purifying system.\\
	We now deal with the $N$-partite scenario. The more general state $\Psi\in\mathsf{St}(\otimes^N\rA)$ has the form:
	\begin{equation*}
		\begin{aligned}
			\Qcircuit @C=1.2em @R=.8em @! R {
				\multiprepareC{2}{\Psi}&			
				\poloFantasmaCn{\rA_1}\qw&
				\qw\\
				\pureghost{\Psi}&
				\vdots&\\
				\pureghost{\Psi_i}&
				\poloFantasmaCn{\rA_N}\qw&
				\qw\\
			}
		\end{aligned}
		\: = \:
		\begin{aligned}
			\Qcircuit @C=1.2em @R=.6em @! R {
				\multiprepareC{1}{\omega_1}&
				\poloFantasmaCn{\rA_1}\qw&
				\qw\\
				\pureghost{\omega_1}&
				\poloFantasmaCn{\rA_2}\qw&
				\qw\\
				\vdots&&\\
				\multiprepareC{1}{\omega_{M/2}}&
				\poloFantasmaCn{\rA_{M-1}}\qw&
				\qw\\
				\pureghost{\omega_{M/2}}&
				\poloFantasmaCn{\rA_M}\qw&
				\qw\\
				\prepareC{\omega_{M/2+1}}&
				\poloFantasmaCn{\rA_{M+1}}\qw&
				\qw\\
				\vdots&&\\
				\prepareC{\omega_{N-M/2}}&
				\poloFantasmaCn{\rA_{N}}\qw&
				\qw\\
			}
		\end{aligned}\, .
	\end{equation*}
	In order to be purificable, every monopartite state has to be pure or purificable and every bipartite state to be pure (since we are admitting no more than bipartite correlations, a bipartite state can be purificable only if it is pure). But if one of the states $\omega_j$ for $j=1,\ldots,N-M/2$ is pure then $\Psi$ can not be internal. Since the only purificable internal state is $\mu$, we have that the only purificable internal $N$-partite state is $\mu\otimes\ldots\otimes\mu$ (thanks to local discriminability the parallel composition of internal states is still an internal state). Furthermore, for what we have seen before, the purification of $\mu\otimes\ldots\otimes\mu$ is unique up to a reversible transformation. This transformation is nothing else that the parallel composition of the maps $(U_k\otimes\tI)$ where $U_k$ are the maps of Eq.~\eqref{eq:reversible-mapping}.
\end{proof}
\noindent\textit{Remark:} in the previous proof we noticed that the maximally mixed state of a system $\otimes^N\rA$ is not the unique mixed state that is purificable. In fact, the more general state $\Phi\in\mathsf{St}(\otimes^N\rA)$ that can be purificated is of the form: $\Phi=\phi_1\otimes\ldots\otimes\phi_N$ where
\begin{equation*}
	\phi_i=
	\begin{cases}
		&\mu\\
		&\omega_{j}\quad\text{for }j\in\left[0,3\right]\\
		&\Omega_{k}\quad\text{for }j\in\left[16,23\right]
	\end{cases}\quad\text{for }i=1,\ldots,N\,,
\end{equation*}
and its purification is still unique up to reversible transformations on the purifying systems.
\subsection{N-partite PR-boxes}
\label{sec:n-partite-boxes}
In the literature of PR-box theory a thorough and systematic study on $k$-partite correlated boxes, with $k\ge3$, has never been made. This represents an important absence within the model since it prevents the theory to be complete. In this Section we show some important consequences that emerge by just considering tripartite correlated boxes with some speculations about the complete $N$-partite model.\\

In our analysis of PR-box model integrated with tripartite boxes we make use of the classification that has been made in Ref.~\cite{tripartite-boxes}. In the article of Pironio \textit{et al.}, the no-signaling polytope is found to have 53856 extremal points, belonging to 46 inequivalent classes. The term inequivalent means that there not exist reversible local transformations that allow to move from a representative of one class to one of another class, while, inside the same class, all the extremal points are connected by local relabelling, see Eqs.~\eqref{eq:local-relabelling} (since we are dealing with three parties boxes, also permutation of the parties is a local relabelling, i.e. $x\rightarrow y$, $y\rightarrow z$, and $z\rightarrow x$ and so on for every possible permutation).\\

Firstly, it is no more granted that the maximally mixed state is the unique internal state purificable in the theory. In fact it could happen that between the 53856 pure tripartite states, there will be one whose marginal state is an internal state different from the maximally mixed one. This leads to think that increasing the number of correlated systems that we are considering the number of internal states that are purificable will also increase. Even if there are not academic works in this matter, it is a very likely and reasonable possibility and we address to future studies to investigate in this direction.\\

Secondly, even if the only internal purificable state would still be the maximally mixed one, the purification is no more unique. To see this it suffices to consider two states $\Psi_1,\Psi_2\in\mathsf{St}(\rA^{\otimes3})$ that have the following form:
\begin{equation*}
	\Psi_1 \, = \,
	\begin{aligned}
		\Qcircuit @C=1.2em @R=.8em @! R {
			\multiprepareC{1}{\Omega}&			
			\poloFantasmaCn{\rA}\qw&
			\qw\\
			\pureghost{\Omega}&
			\poloFantasmaCn{\rA}\qw&
			\qw\\
			\prepareC{\omega}&
			\poloFantasmaCn{\rA}\qw&
			\qw\\
		}
	\end{aligned}
	\quad \text{and}\quad \Psi_2 \, = \,
	\begin{aligned}
		\Qcircuit @C=1.2em @R=.8em @! R {
			\multiprepareC{2}{\Phi^{(44)}}&			
			\poloFantasmaCn{\rA}\qw&
			\qw\\
			\pureghost{\Phi^{(44)}}&
			\poloFantasmaCn{\rA}\qw&
			\qw\\
			\pureghost{\Phi^{(44)}}&
			\poloFantasmaCn{\rA}\qw&
			\qw\\
		}
	\end{aligned}\, ,
\end{equation*}
where $\omega\in\st{A}$ and $\Omega\in\mathsf{St}(\rA^{\otimes2})$ are two pure states and $\Phi^{(44)}$ is a pure state, representative of the $44^{th}$ class described in Ref.~\cite{tripartite-boxes}. If $\Omega$ is one of the non-local bipartite states then they have the same marginal, namely
\begin{equation*}
	\begin{aligned}
		\Qcircuit @C=1.2em @R=.8em @! R {
			\multiprepareC{1}{\Omega}&			
			\poloFantasmaCn{\rA}\qw&
			\qw\\
			\pureghost{\Omega}&
			\poloFantasmaCn{\rA}\qw&
			\multimeasureD{1}{e}\\
			\prepareC{\omega}&
			\poloFantasmaCn{\rA}\qw&
			\ghost{e}\\
		}
	\end{aligned}
	\, = \,
	\begin{aligned}
		\Qcircuit @C=1.2em @R=.8em @! R {
			\prepareC{\mu}&
			\poloFantasmaCn{\rA}\qw&
			\qw\\
		}
	\end{aligned}
	\, = \,
	\begin{aligned}
		\Qcircuit @C=1.2em @R=.8em @! R {
			\multiprepareC{2}{\Phi^{(44)}}&			
			\poloFantasmaCn{\rA}\qw&
			\qw\\
			\pureghost{\Phi^{(44)}}&
			\poloFantasmaCn{\rA}\qw&
			\multimeasureD{1}{e}\\
			\pureghost{\Phi^{(44)}}&
			\poloFantasmaCn{\rA}\qw&
			\ghost{e}\\
		}
	\end{aligned}\, ,
\end{equation*}
where $\mu$ is the maximally mixed state of Eq.~\eqref{eq:maximally-mixed-state} and the second equality derives straightforward once the probability rule of the $44^{th}$ class is written explicitly, as we will see in Eq.~\eqref{eq:tripartite-box-prob-rule}. So we found that $\Psi_1$ and $\Psi_2$ are two purification of the same state but since every reversible transformation $U\in\tU(\rA^{\otimes2})$ is the composition of local reversible maps, "correlation" can not be created and so
\begin{equation*}
	\begin{aligned}
		\Qcircuit @C=1.2em @R=.8em @! R {
			\multiprepareC{2}{\Psi_1}&			
			\poloFantasmaCn{\rA}\qw&
			\qw\\
			\pureghost{\Psi_1}&
			\poloFantasmaCn{\rA}\qw&
			\qw\\
			\pureghost{\Psi_1}&
			\poloFantasmaCn{\rA}\qw&
			\qw\\
		}
	\end{aligned}
	\, \ne \,
	\begin{aligned}
		\Qcircuit @C=1.2em @R=.8em @! R {
			\multiprepareC{2}{\Psi_2}&
			\qw&		
			\poloFantasmaCn{\rA}\qw&
			\qw&
			\qw\\
			\pureghost{\Psi_2}&
			\poloFantasmaCn{\rA}\qw&
			\multigate{1}{U}&
			\poloFantasmaCn{\rA}\qw&
			\qw\\
			\pureghost{\Psi_2}&
			\poloFantasmaCn{\rA}\qw&
			\ghost{U}&
			\poloFantasmaCn{\rA}\qw&
			\qw\\
		}
	\end{aligned}\quad\forall\,U\in\tU(\rA^{\otimes2})\,.
\end{equation*}

Furthermore, restricting our attention to the pure bipartite states, we have noticed that when we pick up two of these states that have the same marginal, than there is always a local reversible map from one to the other and vice-versa. This is the case for the 8 non-local bipartite boxes, that are all connected by local transformations in $\tU(\rA)$ of Eq.~\eqref{eq:reversible-mapping}. This mechanism will turn out to be exactly the one responsible for the impossibility of perfectly secure bit commitment, as we will see in detail in the next Chapter. Since in the tripartite scenario not all the tripartite non-local boxes are connected by local transformation (we remind that local relabelling is not enough to change from a class to another), it is reasonable to think that perfect bit commitment will be possible, or at least a completely new way of cheating has to be thought. With this purpose we will propose a scheme of bit commitment as the conclusive Section of the next Chapter.

\chapter{No Bit Commitment in PR-Boxes}
\label{chap:no-bc}	
In the past years numerous protocols have been proposed to realize bit commitment using PR-boxes. However, as outlined by A.J. Short, N. Gisin, and S. Popescu in Ref.~\cite{short-gisin-popescu}, ``\textit{it is surprising that the possibility that non-local correlations which are stronger than those in quantum mechanics could be used for bit commitment, because it is the very existence of non-local correlations which in quantum mechanics prevents bit commitment}''. In that article they particularly referred to the  protocol prosed by S. Wolf and J. Wullschleger \cite{wolf} and showed that it was erroneous by argument of causality. After that also Buhrman \textit{et al.} \cite{Buhrman_2006} proposed a bit commitment protocol in PR-box theory that was claimed to be unconditionally secure and where the counter-proof of Short, Gisin and Popescu did not work anymore.\\

From a OPT point of view, when dealing with PR-boxes, some issues arise since they still not have a complete and closed theory. In fact, $k$-partite boxes with $k\ge3$ have been studied only roughly and a coherent and comprehensive theory has not been proposed. As we have seen in Section~\ref{sec:n-partite-boxes}, simplistic generalizations are not adequate since admitting more than bipartite correlations alters significantly the theory. Nevertheless, in the literature not only PR-box model is generally considered admitting no more than bipartite states, but also local transformations (that are admissible in the theory) are ignored.\\

In this final Chapter we propose a proof of impossibility of perfectly secure bit commitment in PR-boxes (even if under two important limitations: pure input states and bipartite boxes, the proof includes almost all the protocols proposed in literature that make use of PR-boxes). Furthermore we will explicitly describe a cheating protocol, contextualized in OPTs, that confute both the scheme proposed by Wolf and Wullschleger and the one by Buhrman \textit{et al.}. We will show that just admitting local atomic reversible transformations the protocols proposed in literature can be cheated.\\
Our proof joins the work published by Barnum, Dahlsten, Leifer, and Toner \cite{non-classicality} where, in the framework of probabilistic theories, they prove that in all theories that are locally non-classical but do not have entanglement, there exists a bit commitment protocol that is exponentially secure in the number of systems used. If the protocol of Buhrman \textit{et al.} would have been turned out to be correct then it would have represented the first example of an unconditionally secure bit commitment protocol valid in a theory with entanglement. However the question if a theory with entanglement admits perfectly secure bit commitment is still open.\\

Finally we will sketch at the end of the Chapter a bit commitment scheme that make use of tripartite non-local boxes that is not more cheatable by local transformations and it could satisfy perfectly secure bit commitment. But advancements in the theory are necessary in order to give a definitive answer.

\section{No-Perfectly Secure Bit Commitment}

In this Section we provide an explicit proof of the impossibility of perfectly secure bit commitment in PR-box theory. However, two important limitation will be adopted. Even if we will consider arbitrary $N$-partite systems (with $N$ finite), we will not admit more than bipartite correlations (it will be taken for granted for the rest of the Section). Furthermore $\Psi_0$ and $\Psi_1$ will always be selected between the pure states. This is due to the fact that discrimination has never been studied in PR-box theory and the only results we rely on are those of Section~\ref{sec:discriminability-PRboxes} that refers to pure states. If it would turn out that the strategy exposed in Section~\ref{sec:discriminability-PRboxes} is the unique one to grant perfect discriminability between two arbitrary bipartite states then the protocol could be easily extended also to $\Psi_0,\,\Psi_1$ as arbitrary mixed states.\\

Actually we will prove our theorem in two different way. In the first proof, we will state the property of perfect bit commitment and we will show that inconsistencies arise. In the second, that we will call alternative proof, we will show that, with some shrewdness, the proof of Section~\ref{sec:lightening-noBC} can be used also in this context.

\subsection{First Proof}

In this Section we make use of the definition of the protocol given in Section~\ref{sec:defn-bit-commitment} and nomenclature of states and transformations given in Section~\ref{sec:prbox-as-opt}.\\
Before the main theorem a preliminary lemma is in order.
\begin{lemma}
	If a bit commitment protocol is correct with probability one then the two input states $\Psi_0,\,\Psi_1\in\mathsf{St}(\rA^{\otimes N}\rB^{\otimes N})$ shared by Alice ($\rA^{\otimes N}$) and Bob ($\rB^{\otimes N}$) are of the form
	\begin{equation}
		\label{eq:input-states}
		\begin{aligned}
			\Qcircuit @C=1.2em @R=.7em @! R {
				\multiprepareC{4}{\Psi_i}&			
				\poloFantasmaCn{\rA_1}\qw&
				\qw\\
				\pureghost{\Psi_i}&
				\poloFantasmaCn{\rB_1}\qw&
				\qw\\
				\pureghost{\Psi_i}&
				\vdots&\\
				\pureghost{\Psi_i}&
				\poloFantasmaCn{\rA_N}\qw&
				\qw\\
				\pureghost{\Psi_i}&
				\poloFantasmaCn{\rB_N}\qw&
				\qw
			}
		\end{aligned}
		\: = \:
		\begin{aligned}
			\Qcircuit @C=1.2em @R=.6em @! R {
				\multiprepareC{1}{\Omega_{k_1(i)}}&
				\poloFantasmaCn{\rA_1}\qw&
				\qw\\
				\pureghost{\Omega_{k_1(i)}}&
				\poloFantasmaCn{\rB_1}\qw&
				\qw\\
				\vdots&&\\
				\multiprepareC{1}{\Omega_{k_M(i)}}&
				\poloFantasmaCn{\rA_M}\qw&
				\qw\\
				\pureghost{\Omega_{k_M(i)}}&
				\poloFantasmaCn{\rB_M}\qw&
				\qw\\
				\prepareC{\omega_{k_{M+1}(i)}}&
				\poloFantasmaCn{\rA_{M+1}}\qw&
				\qw\\
				\prepareC{\omega_{k_{M+2}(i)}}&
				\poloFantasmaCn{\rB_{M+1}}\qw&
				\qw\\
				\vdots&&\\
				\prepareC{\omega_{k_{2N-M-1}(i)}}&
				\poloFantasmaCn{\rA_{N}}\qw&
				\qw\\
				\prepareC{\omega_{k_{2N-M}(i)}}&
				\poloFantasmaCn{\rB_{N}}\qw&
				\qw\\
			}
		\end{aligned}
		\quad \text{ for }i=0,\,1
	\end{equation}	
	where $k_j(i)\in\left[16,23\right]$ for $j=1,\ldots,M$ and $k_j(i)\in\left[0,3\right]$ for $j=M+1,\ldots,2N-M$.
\end{lemma}
\begin{proof}
	The thesis follows immediately from the two assumptions we made: no more than bipartite correlated boxes and pure input states. In fact, for every two pure states a perfectly discriminating procedure always exists, as analysed in Section~\ref{sec:discriminability-PRboxes}. For discriminate between two parallel compositions of pure states the parallel composition of the discriminating POVMs for each pair of states is sufficient.
\end{proof}
In the previous lemma it would be possible that Alice and Bob have in control also non-factorized bipartite states, i.e. $\Omega_{k_{M+1}}\in\mathsf{St}(\rA_{M+1}\rA_{M+3})$ for $k_{M+1}\in\left[16,23\right]$ instead of $\omega_{k_{M+1}}\otimes\omega_{k_{M+3}}\in\mathsf{St}(\rA_{M+1}\rA_{M+3})$ for $k_{M+1},k_{M+3}\in\left[0,3\right]$, however this does not carry any modification in the proof and so, to not make the notation even more troublesome, we will refer to Eq.~\eqref{eq:input-states} as the more general input states.
\begin{theorem}
	Perfect bit commitment is impossible in PR-box theory.
\end{theorem}
\begin{proof}
	If a bit commitment is perfect it means that it should be correct with probability one, perfectly concealing and perfectly binding.\\
	If it is correct with probability one then, by the previous Lemma, the input states $\Psi_0,\,\Psi_1\in\mathsf{St}(\rA^{\otimes N}\rB^{\otimes N})$ must have the form of Eq.~\eqref{eq:input-states}.\\
	If it is also perfectly concealing, then we have to impose the condition of Eq.~\eqref{eq:perfectly-concealing}, i.e. $(e|_{\rA^{\otimes N}}|\Psi_0)_{\rA^{\otimes N}\rB^{\otimes N}}=(e|_{\rA^{\otimes N}}|\Psi_1)_{\rA^{\otimes N}\rB^{\otimes N}}$:
	\begin{equation*}
		\begin{aligned}
			\Qcircuit @C=1.2em @R=.7em @! R {
				\multiprepareC{1}{\Psi_0}&
				\poloFantasmaCn{\rA^{\otimes N}}\qw&
				\measureD{e_{\rA^{\otimes N}}}\\
				\pureghost{\Psi_0}&
				\poloFantasmaCn{\rB^{\otimes N}}\qw&
				\qw\\
			}
		\end{aligned}
		\: = \:
		\begin{aligned}
			\Qcircuit @C=1.2em @R=.7em @! R {
				\multiprepareC{1}{\Omega_{k_1(0)}}&
				\poloFantasmaCn{\rA_1}\qw&
				\measureD{e_\rA}\\
				\pureghost{\Omega_{k_1(0)}}&
				\poloFantasmaCn{\rB_1}\qw&
				\qw\\
				\vdots&&\\
				\multiprepareC{1}{\Omega_{k_M(0)}}&
				\poloFantasmaCn{\rA_M}\qw&
				\measureD{e_\rA}\\
				\pureghost{\Omega_{k_M(0)}}&
				\poloFantasmaCn{\rB_M}\qw&
				\qw\\
				\prepareC{\omega_{k_{M+2}(0)}}&
				\poloFantasmaCn{\rB_{M+1}}\qw&
				\qw\\
				\vdots&&\\
				\prepareC{\omega_{k_{2N-M}(0)}}&
				\poloFantasmaCn{\rB_{N}}\qw&
				\qw\\
			}
		\end{aligned}
		\: = \:
		\begin{aligned}
			\Qcircuit @C=1.2em @R=.7em @! R {
				\multiprepareC{1}{\Omega_{k_1(1)}}&
				\poloFantasmaCn{\rA_1}\qw&
				\measureD{e_\rA}\\
				\pureghost{\Omega_{k_1(1)}}&
				\poloFantasmaCn{\rB_1}\qw&
				\qw\\
				\vdots&&\\
				\multiprepareC{1}{\Omega_{k_M(1)}}&
				\poloFantasmaCn{\rA_M}\qw&
				\measureD{e_\rA}\\
				\pureghost{\Omega_{k_M(1)}}&
				\poloFantasmaCn{\rB_M}\qw&
				\qw\\
				\prepareC{\omega_{k_{M+2}(1)}}&
				\poloFantasmaCn{\rB_{M+1}}\qw&
				\qw\\
				\vdots&&\\
				\prepareC{\omega_{k_{2N-M}(1)}}&
				\poloFantasmaCn{\rB_{N}}\qw&
				\qw\\
			}
		\end{aligned}
		\: = \:
		\begin{aligned}
			\Qcircuit @C=1.2em @R=.7em @! R {
				\multiprepareC{1}{\Psi_1}&
				\poloFantasmaCn{\rA^{\otimes N}}\qw&
				\measureD{e_{\rA^{\otimes N}}}\\
				\pureghost{\Psi_1}&
				\poloFantasmaCn{\rB^{\otimes N}}\qw&
				\qw\\
			}
		\end{aligned}\,.
	\end{equation*}
	So, for the systems from $\rB_{M+1}$ to $\rB_{N}$ we have that $\omega_{k_{M+2s}(0)}=\omega_{k_{M+2s}(1)}$ for $s=1,\ldots,N-M$. For the systems from $\rB_{M+1}$ to $\rB_{N}$ we can not make any further deductions since these 8 pure states have all the same marginal.\\
	
	In any case if the protocol is correct with probability one and perfectly concealing it can not be perfectly binding.\\
	In fact, for the non-local bipartite states there will be one of the local transformations of Eq.~\eqref{eq:single-system-unitaries} such that
	\begin{equation}
		\begin{aligned}
			\Qcircuit @C=1.2em @R=.7em @! R {
				\multiprepareC{1}{\Omega_{k_j(0)}}&
				\poloFantasmaCn{\rA}\qw&
				\gate{U_{k_j}}&
				\poloFantasmaCn{\rA}\qw&
				\qw\\
				\pureghost{\Omega_{k_j(i)}}&
				\qw&
				\poloFantasmaCn{\rB}\qw&
				\qw&\qw\\
			}
		\end{aligned}\,
		\: = \:
		\begin{aligned}
			\Qcircuit @C=1.2em @R=.7em @! R {
				\multiprepareC{1}{\Omega_{k_j(1)}}&
				\poloFantasmaCn{\rA}\qw&
				\qw\\
				\pureghost{\Omega_{k_j(1)}}&
				\poloFantasmaCn{\rB}\qw&
				\qw\\
			}
		\end{aligned}\,,
	\end{equation}
	as expressed by Eq.~\eqref{eq:reversible-mapping}, where $U_{k_j}\in\tU(\rA)$. Furthermore for all the other states, namely for the local states $\omega_{k_{M+2s+1}(i)}$ with $s=0,\ldots,N-M-1$, $i=0,1$, it is immediate that there will be local transformations $U_{k_{M+2s+1}}\in\tU(\rA)$, $s=0,\ldots,N-M-1$ chosen from the 4 of Eq.~\eqref{eq:single-system-unitaries} that will permit to Alice to switch unnoticed by Bob from $\Psi_0$ to $\Psi_1$ and vice-versa.
\end{proof}

\subsection{Alternative Proof}

In Section~\ref{sec:lightening-noBC} we pointed out the three sufficient conditions that a theory has to satisfy in order to ensure the impossibility of perfectly secure bit commitment: causality, atomicity of composition and the one required in Axiom~\ref{axiom:7}.\\

Clearly PR-box theory is manifestly causal, since Eq.~\eqref{eq:no-signalling} imposes exactly the no-signalling constraint.\\
Furthermore, also atomicity of composition is fulfilled by PR-box theory. The parallel composition of atomic operations is still atomic due to local discriminability, see Ref.~\cite{manessi}. To verify that also the sequential composition of atomic operation is still atomic it only need to consider all the atomic operations in the theory, see Ref.~\cite{d'ariano-tosini}, and straightforwardly compute their composition.\\

For what concern Axiom~\ref{axiom:7} some considerations are in order. We required the existence of at least one dynamically faithful pure state, $\Psi^{(\rA)}$, and the existence of a purification for every state $R$ that has the same marginal of $\Psi^{(\rA)}$. Actually the latter assumption is excessive. If we look at Theorem~\ref{thm:RevDila}, where this assumption needs to work, it would be enough that every state $|R)_{\rB\tilde{\rA}}$ that is obtained from $|\Psi^{(\rA)})_{\rA\tilde{\rA}}$ by a local channel $\tC$, i.e.  $|R)_{\rB\tilde{\rA}}\coloneqq(\tC\otimes\tI_{\tilde{\rA}})|\Psi^{(\rA)})_{\rA\tilde{\rA}}$, is purificable. In fact, if $R$ is so defined, it certainly has the same marginal of $\Psi^{(\rA)}$: 
\begin{equation*}
	\begin{aligned}
		\Qcircuit @C=1.2em @R=.7em @! R {
			\multiprepareC{1}{R}&
			\poloFantasmaCn{\rB}\qw&
			\measureD{e}\\
			\pureghost{R}&
			\poloFantasmaCn{\tilde{\rA}}\qw&
			\qw\\					
		}
	\end{aligned}
	\, = \,
	\begin{aligned}
		\Qcircuit @C=1.2em @R=.7em @! R {
			\multiprepareC{1}{\Psi^{(\rA)}}&
			\poloFantasmaCn{\rA}\qw&
			\gate{\tC}&
			\poloFantasmaCn{\rB}\qw&
			\measureD{e}\\
			\pureghost{\Psi^{(\rA)}}&
			\qw&
			\poloFantasmaCn{\tilde{\rA}}\qw&
			\qw&\qw\\					
		}
	\end{aligned}
	\, = \,
	\begin{aligned}
		\Qcircuit @C=1.2em @R=.7em @! R {
			\multiprepareC{1}{\Psi^{(\rA)}}&
			\poloFantasmaCn{\rA}\qw&
			\measureD{e}\\
			\pureghost{\Psi^{(\rA)}}&
			\poloFantasmaCn{\tilde{\rA}}\qw&
			\qw\\					
		}
	\end{aligned}\, .
\end{equation*}
Now if we limit to consider $\tC$ as a local reversible channel (since in the theory the reversible transformations are atomic), being $\Psi^{(\rA)}$ pure by hypothesis, due to atomicity of composition also $R$ is pure and hence trivially purificable. So, under this further constraint, we only need to check the existence of a dynamically faithful pure state within the theory.\\
We can help ourselves by the fact that the purification of an internal state is dynamically faithful (and obviously pure). This result was published in Ref.~\cite{chiribella} for theories that satisfy local discriminability and purification but it can easily extended to PR-box theory given the existence of at least one internal state that is purificable: $\mu$.\\

Now the answer is immediate, all the non-local bipartite extremal point of the 8-dimensional polytope are dynamically faithful pure states. At this point it is easy to see that we can just choose one of the pure states in Eq.~\eqref{eq:non-local-bipartite-states}, as the dynamically faithful pure state in Axiom~\ref{axiom:7}.\\
Finally, it is straightforward to verify that the purification is unique up to a local reversible transformation on the purifying system, see Eq.~\eqref{eq:reversible-mapping}.\\

The further limitation that we imposed, on the atomicity of $\tC$, practically requires that for every commitment protocol the two encodings $\tA_0,\,\tA_1\in\mathsf{Transf}(\rA_1...\rA_N\rightarrow\rB_1...\rB_N\rF_N)$ have the marginal atomic: $(e|_{\rF_N}\tA_0=(e|_{\rF_N}\tA_1\coloneqq\tC$. But from Eq.~\eqref{eq:reversible} we know that every element of the set of reversible transformation of the composite system $\rA\otimes\rA$ is nothing else than the tensor product of local reversible transformations (with eventually the swap map). So if we limit to consider only atomic local transformations the initial new constraint is satisfied.\\

Finally, thanks to local discriminability, given two pure faithful states $\Psi^{(\rA)}\in\mathsf{St}(\rA\tilde{\rA})$ and $\Psi^{(\rB)}\in\mathsf{St}(\rB\tilde{\rB})$ for system $\rA$ and $\rB$, respectively, then also $\Psi^{(\rA)}\otimes\Psi^{(\rB)}\in\mathsf{St}(\rA\tilde{\rA}\rB\tilde{\rB})$ is a pure dynamically faithful state for the compound system $\rA\rB$ (a rigorous proof can be found in Ref.~\cite{chiribella}) and so the previous discussion is easily generalizable to any $N$-partite system.\\

In conclusion, the proof of impossibility of perfectly secure bit commitment of Section~\ref{sec:lightening-noBC} can be easily extended to include PR-box theory (limited to no more than bipartite correlations and only reversible transformations).

\section{Unconditionally Secure Bit Commitment}
\label{sec:cheating-BC}
As outlined before, in Ref.~\cite{Buhrman_2006} Buhrman \textit{et al.} proposed a bit commitment protocol that was claimed to be unconditionally secure. They would like to show that superstrong non-local correlations in the form of non-local boxes enable to solve cryptographic problems otherwise known to be impossible. In particular, their result would imply that the no-signaling principle and secure computation are compatible in principle.\\
However, we now prove how it would be possible for Alice to perfectly cheat, i.e. with null probability of being detected by Bob, making use of the local reversible atomic transformations of Eq.~\eqref{eq:reversible-mapping}.\\

In our analysis, we begin from dealing with only one bipartite PR-box and we suppose that Alice's input is the committed bit (even if this is not a well defined bit commitment protocol it is as well an instructive example in order to simplify the following analysis).\\
Alice and Bob share the non-local bipartite PR-box $\Omega_{18}\in\mathsf{St}_\mathbb{R}(\rA\rB)$ (to which corresponds the probability rule $p_{000}$) but they have access only to system $\rA$ and $\rB$, respectively. We can summarise the protocol as follows.\\

\noindent\textbf{COMMIT:}
\begin{itemize}
	\item Alice select her committed bit $x$, inputs it and obtain output bit $a$;
	\item Bob inputs a random bit $y$ and obtains output bit $b$.
\end{itemize}
\textbf{REVEAL:}
\begin{itemize}
	\item Alice sends $x$ and $a$ to Bob;
	\item Bob checks to see if $a\oplus b=xy$. If this relation is true, Bob accepts $x$ as the revealed bit, otherwise he knows that Alice has cheated and rejects Alice’s revelation.
\end{itemize}
It is easy to see that Alice has a probability to cheat successfully equal to $\frac{1}{2}$.\\
With the help of local transformation Alice can make the probability of successful cheating equal to 1. In fact, it is sufficient to find $(\alpha,\beta,\gamma)$ of Eq.~\eqref{eq:15} such that, for given $x$ and $y$ and for the output couple $(a,b)$ (i.e. such that $p_{\alpha\beta\gamma}(a,b|x,y)\ne0$) exists a generic function $f:\{0,1\}\longrightarrow\{0,1\}$ such that, given $a^\prime=f(a)$ and $x^\prime=x\oplus1$, $p_{000}(a^\prime,b|x^\prime,y)\ne0$.\\
Mathematically, to find suitable $(\alpha,\beta,\gamma)$ it is sufficient to resolve
\begin{equation}\label{eq:cheating}
	\begin{cases}
		& a\oplus b=xy\oplus \alpha x\oplus \beta y\oplus\gamma\\
		& a^\prime\oplus b=x^\prime y
	\end{cases}.
\end{equation}
We find that $\beta=1$ and $f(a)=a\oplus\alpha x\oplus\gamma$, for every choice of $\alpha$ and $\gamma$.\\
In conclusion, if Alice perform a local transformation (given by Eq.~\eqref{eq:single-system-unitaries}) such that $p_{000}\longrightarrow p_{\alpha1\gamma}$ and then inputs her bit $x$ and gets output $a$, she can reveal $x^\prime$ and $a^\prime=a\oplus\alpha x\oplus\gamma$ to Bob, who will accept with probability 1.\\

We can now consider the unconditionally secure bit-commitment protocol proposed in Ref.~\cite{Buhrman_2006} where $2n+1$ non-local bipartite PR-boxes in the state $\Omega_{18}\in\st{AB}$ are shared between Alice and Bob.\\
The authors found that Alice probability of successfully cheating is at maximum equal to 1/2 but can be asymptotically reduced if the protocol is repeated $k$ times. However, using local transformation Alice can cheat without being detected with probability 1. We refer to the original article about the commitment protocol and we outline only the "cheating procedure".\\

\noindent\textbf{COMMIT:}
\begin{itemize}
	\item Alice wants to commit to bit $c$ but to send to Bob bit $c^\prime=c\oplus1$. So she chooses $x\in\{0,1\}^{2n+1}$ by choosing the first $2n$ bits such that $|x^\prime_1...x^\prime_{2n}|_{11}$ is even where $x_i^\prime=x_i\oplus1$ for $i=1,2,...,2n$ (given a string of even length $x$, $|x|_{11}$ is the number of substring "11" in $x$ starting at an odd position) and then choosing $x_{2n+1}=c$ (analogously she can choose $x\in\{0,1\}^{2n+1}$ such that $|x^\prime_1...x^\prime_{2n}|_{11}$ is odd and $x_{2n+1}=c\oplus1$);
	\item Alice, for each of the $2n+1$ shared boxes, apply a local transformation changing the probability law from $p_{000}\longrightarrow p_{\alpha1\gamma}$, then se puts the bits $x_1,x_2,...,x_{2n+1}$ into the boxes $1,2,...,2n+1$. Let $a_1,a_2,...,a_{2n+1}$ be Alice's output bits from the boxes;
	\item Alice computes the parity of all the "cheated" output bits
	$A^\prime=\oplus_{i=0}^{2n+1}a^\prime_i$ and send $A^\prime$ to Bob, where $a^\prime_i=a_i\oplus\alpha x_i\oplus\gamma$ for $i=1,2,...,2n+1$;
	\item Bob randomly chooses a string $y\in_\mathbb{R}\{0,1\}^{2n+1}$ and puts the bits $y_1,y_2,...,y_{2n+1}$ into his boxes. We call the output bits from his boxes $b_1,b_2,...,b_{2n+1}$.
\end{itemize}
Then the \textbf{REVEAL} phase:
\begin{itemize}
	\item Alice sends $c^\prime$, her string $x^\prime$ (where $x_i^\prime=x_i\oplus1$ for $i=1,2,...,2n+1$) and all her $2n+1$ "cheated" outputs bits (i.e. $a^\prime_i$) to Bob;
	\item Bob checks if Alice's data is consistent: $\forall i\in\{0,1\}^{2n+1}$, $x^\prime_i\cdot y_i=a^\prime_i\oplus b_i$ and $|x^\prime_1...x^\prime_{2n}|_{11}+x^\prime_{2n+1}+c^\prime$ is even. Since he finds no error, he accepts $c^\prime$ as the committed bit.
\end{itemize}
In fact, according to Eq.~\eqref{eq:cheating}, every couple $(x^\prime_i,a_i^\prime)$ for $i=1,2,...,2n+1$ sent by Alice to Bob	satisfies $x_i^\prime\cdot y=a_i^\prime\oplus b$ and by the proposed choices of $x\in\{0,1\}^{2n+1}$ also the parity constraints are
secured.

\section{Bit Commitment in Tripartite Scenario}

We have stressed that the proof of impossibility of perfectly secure bit commitment in PR-box theory was limited by considering no more than bipartite correlated boxes in $N$-partite systems. In fact, if we take in consideration just only tripartite boxes, the scenario changes considerably.\\
In this Section we highlight that a scheme of bit commitment protocol like the one in Ref.~\cite{Buhrman_2006} that would make use of tripartite non-local boxes would not subject to the cheating by local reversible transformations and it would represent a possible perfectly (or at least unconditionally) secure bit commitment protocol.\\

In Section~\ref{sec:n-partite-boxes} we used the classification done in Ref.~\cite{tripartite-boxes} where the tripartite correlated boxes were divided in 46 non equivalent classes, i.e. 46 classes whose states are not connected by local reversible transformations. In order to see it directly, it is sufficient to write explicitly the probability rules for those classes. As an example we write out the representatives of the probability rule for three non-local tripartite classes (number 44, 45, and 46 in Ref.~\cite{tripartite-boxes}):
\begin{equation}
	\label{eq:tripartite-box-prob-rule}
	\begin{aligned}
		\text{Class 44: }& p(a,b,c|x,y,z)=
		\begin{cases}
			1/4& a\oplus b\oplus c=xyz\\
			0 & \text{otherwise}
		\end{cases}\\
		\text{Class 45: }& p(a,b,c|x,y,z)=
		\begin{cases}
			1/4& a\oplus b\oplus c=xy\oplus xz\\
			0 & \text{otherwise}
		\end{cases}\\
		\text{Class 46: }& p(a,b,c|x,y,z)=
		\begin{cases}
			1/4& a\oplus b\oplus c=xy\oplus xz\oplus yz\\
			0 & \text{otherwise}
		\end{cases}\\
	\end{aligned},
\end{equation}
and we note that the local relabelling of Eq.~\eqref{eq:local-relabelling} plus the permutations of the parties do not permit to move from every representative of one class to any one of any other class.\\

Now to build our bit commitment protocol we decide to encode $b=0,1$ by choosing as input state $\Psi_0,\Psi_1\in\mathsf{St}(\rA^{\otimes3})$ one representative of the $44^{th}$ class and one of the $45^{th}$, say the ones in Eq.~\eqref{eq:tripartite-box-prob-rule} for the sake of simplicity.\\
Following the strategy in Section~\ref{sec:discriminability-PRboxes} we can construct a POVM that is able to discriminate between them, for example $\{a,e-a\}$ where $a\coloneqq b^{(0)} \otimes b^{(3)}\otimes b^{(0)} + b^{(0)}\otimes b^{(1)}\otimes b^{(2)}+ b^{(2)} \otimes b^{(1)}\otimes b^{(0)}+b^{(2)} \otimes b^{(3)}\otimes b^{(2)}$ and $e=e_\rA\otimes e_\rA\otimes e_\rA$. So the protocol is correct with probability one.\\
Furthermore $\Psi_0$ and $\Psi_1$ are not connected by any local reversible transformation (as pointed out before) and so the protocol is also perfectly binding. At the same time the two input states have the same marginal on Bob system, and so it is also perfectly concealing. Namely, it is not difficult to derive from Eq.~\eqref{eq:tripartite-box-prob-rule} that
\begin{equation*}
	\begin{aligned}
		\Qcircuit @C=1.2em @R=.8em @! R {
			\multiprepareC{2}{\Psi_0}&			
			\poloFantasmaCn{\rA}\qw&
			\qw\\
			\pureghost{\Psi_0}&
			\poloFantasmaCn{\rA}\qw&
			\multimeasureD{1}{e}\\
			\pureghost{\Psi_0}&
			\poloFantasmaCn{\rA}\qw&
			\ghost{e}\\
		}
	\end{aligned}
	\, = \,
	\begin{aligned}
		\Qcircuit @C=1.2em @R=.8em @! R {
			\prepareC{\mu}&
			\poloFantasmaCn{\rA}\qw&
			\qw\\
		}
	\end{aligned}
	\, = \,
	\begin{aligned}
		\Qcircuit @C=1.2em @R=.8em @! R {
			\multiprepareC{2}{\Psi_1}&			
			\poloFantasmaCn{\rA}\qw&
			\qw\\
			\pureghost{\Psi_1}&
			\poloFantasmaCn{\rA}\qw&
			\multimeasureD{1}{e}\\
			\pureghost{\Psi_1}&
			\poloFantasmaCn{\rA}\qw&
			\ghost{e}\\
		}
	\end{aligned}\, .
\end{equation*}
In conclusion, by simply choosing as our input states two suitable non-local tripartite boxes we build a protocol that is correct with probability 1, perfectly concealing and seems to be also perfectly binding (we are considering only reversible transformations). Naturally, since the theory is not complete, we would not like to fall in the same mistake of claiming our bit commitment protocol perfectly (or, when applied in Ref.~\cite{Buhrman_2006}, unconditionally) secure in PR-box theory but, less then unexpected turns in the theory, PR-box theory could really represent the first example of a theory with entanglement and bit commitment.

\clearpage
\thispagestyle{empty}
\phantom{a}
\newpage

\chapterstyle{thesis3}
\chapter*{Conclusions}
\label{conclusions}
\addcontentsline{toc}{chapter}{Conclusions}
\thispagestyle{plain}

In this thesis we have formalized the bit commitment protocol in the operational language to investigate its feasibility in a more general context than quantum theory. In this way we are willing to make the first step in the understanding of the relation that exists between bit commitment and the operational axioms of quantum theory. In particular we focused on the study of BC in PR-box theory, a theory that is more non-local than the quantum one but where the purification property does not hold.\\

In performing this analysis we were forced to investigate new aspects of PR-box theory, since it is far away to be closed and complete. Some of the most remarkable results that we achieved are the following. Firstly, we described a strategy that grants to always find a POVM able to perfectly discriminate between any two bipartite pure states. In addition, we proved that the maximally mixed state $\mu$, is purificable. Furthermore, if the theory is limited to no more than bipartite correlated boxes, then $\mu$ is the unique internal state that is purificable and its purification is also unique up to reversible transformations on the purifying system. Then, we were able to show how simplistic generalizations to the arbitrary $N$-partite case are not appropriate. For example just admitting tripartite boxes we showed how the purification of the maximally mixed state is not more unique, and how it is not even granted that $\mu$ is still the unique internal state that is purificable.\\

\thispagestyle{plain}

After the presentation of the PR-box theory in Chapter~\ref{chap:pr-box}, in Chapter~\ref{chap:no-bc} we presented the results about the impossibility of bit commitment in PR-boxes.\\
In the literature of BC performed on non-local boxes, these have been rarely considered as part of a coherent theory and in fact tripartite correlated boxes or local reversible transformations (that are admissible in the theory) have always been neglected.\\
By simply taking in consideration the reversible transformations we proposed a proof of impossibility of perfectly secure bit commitment in a PR-box theory limited by the following constraint:
\begin{enumerate}
	\item no more than bipartite correlated boxes admitted;
	\item only reversible transformations considered;
	\item only pure input states.
\end{enumerate}
Even under these three important limitation our scenario is still enough general to include all the protocol proposed in literature. Furthermore we were also able to adapt the solid proof in Ref.~\cite{chiribella} (in same context as above) to obtain an identical result for the impossibility of bit commitment in PR-boxes. In addition, even if we dealt only with perfectly secure bit commitment, we explicitly described a scheme in which Alice is able to cheat perfectly in the protocol proposed in Ref.~\cite{Buhrman_2006}, that was claimed to be unconditionally secure.\\

Finally we relaxed the limitations that we imposed on the theory and we proposed a protocol that seems to be perfectly secure. However, since the theory is not complete, we address future studies to investigate this and the other questions that remain unanswered in this thesis.\\

First of all, if the discriminating strategy could include not only pure states but all of them, then the proof of impossibility of perfectly secure BC could be extended to non-pure input states, too. Anyhow, the great unknown variable is still the integration of generic $N$-partite non-local correlated boxes in the theory. The consequences could be very surprising, for example it would even be possible that the number of internal states that are purificable asymptotically increases increasing $N$, and so that the purification principle could hold in the limit $N\rightarrow\infty$.\\

In conclusion we presented some results on the impossibility of perfectly secure bit commitment in PR-boxes and we precisely pointed out their limit of validity, that, even if under considerable assumptions, still have an important comparison with the literature on the subject.\\
However, as we have outlined many times in this work, to achieve definitive results, other progresses in the fundamental aspects of the theory are absolutely necessary.
\thispagestyle{plain}

\clearpage
\thispagestyle{empty}
\phantom{a}
\newpage

\chapterstyle{thesis}
\appendix
\chapter{Tables}
\label{app:tables}

\begin{table}[h]
	\centering
	\begin{tabular}{c|c}
		\label{tab:local}
		($\alpha\beta\gamma\delta$) & $a\oplus b$ \\
		\hline
		0000 & 0\\
		0001 & 1\\
		0010 & $y$\\
		0011 & $y\oplus1$\\
		0100 & 1\\
		0101 & 0\\
		0110 & $y\oplus1$\\
		0111 & $y$\\
		1000 & $x$\\
		1001 & $x\oplus1$\\
		1010 & $x\oplus y$\\
		1011 & $x\oplus y\oplus1$\\
		1100 & $x\oplus1$\\
		1101 & $x$\\
		1110 & $x\oplus y\oplus1$\\
		1111 & $x\oplus y$\\
	\end{tabular}
	\caption{Outcome relations for the 16 bilocal-boxes.}
\end{table}
\begin{table}
	\centering
	\begin{tabular}{c|c}
		\label{tab:nonlocal}
		($\alpha\beta\gamma$) & $a\oplus b$\\
		\hline
		000 & $xy$ \\
		001 & $xy\oplus 1$ \\
		010 & $xy\oplus y$ \\
		011 & $xy\oplus y\oplus1$ \\
		100 & $xy\oplus x$ \\
		101 & $xy\oplus x\oplus1$ \\
		110 & $xy\oplus x\oplus y$ \\
		111 & $xy\oplus x\oplus y\oplus1$ \\
	\end{tabular}
	\caption{Outcome relations for the 8 bipartite nonlocal-boxes.}
\end{table}

\clearpage
\thispagestyle{empty}
\phantom{a}
\backmatter

\chapterstyle{default}

\printindex 

\end{document}